\documentclass[11pt]{article}
\usepackage{amsmath}
\usepackage{amsfonts}
\usepackage{amssymb}
\usepackage{amsthm}
\usepackage[mathscr]{euscript}
\usepackage{bbm}
\usepackage{bm}
\usepackage{mathrsfs}
\usepackage{prodint}
\usepackage{graphicx}
\usepackage[left=1in,right=1in,top=1in,bottom=1in]{geometry}
\usepackage{setspace}
\usepackage{xcolor}

\usepackage{authblk}
\usepackage[hypertexnames=false]{hyperref}

\usepackage[authoryear]{natbib}
\usepackage{multibib}
\newcites{supp}{References}

\usepackage{algorithm}
\usepackage{algorithmicx}
\usepackage[noend]{algpseudocode}
\algnewcommand{\IFor}[1]{\State\algorithmicfor\ #1\ \algorithmicdo}
\algnewcommand{\EndIFor}{\unskip\ }

\title{Prediction Sets Adaptive to\\Unknown Covariate Shift}

\date{}
\author{Hongxiang Qiu}
\author{Edgar Dobriban}
\author{Eric Tchetgen Tchetgen\footnote{Author e-mail addresses: 
\texttt{qiuhx@wharton.upenn.edu},
\texttt{dobriban@wharton.upenn.edu},
\texttt{ett@wharton.upenn.edu}
}}

\affil{Department of Statistics, The Wharton School, University of Pennsylvania}

\allowdisplaybreaks
\newtheorem{theorem}{Theorem}
\newtheorem{lemma}{Lemma}
\newtheorem{corollary}{Corollary}

\theoremstyle{definition}
\newtheorem{remark}{Remark}
\newtheorem{condition}{Condition}

\DeclareMathOperator{\logit}{logit}
\DeclareMathOperator{\expit}{expit}

\newcommand{\real}{{\mathbb{R}}}
\newcommand{\modelspace}{{\mathcal{M}}}
\newcommand{\ind}{{\mathbbm{1}}}
\newcommand{\expect}{{\mathbb{E}}}

\newcommand{\funclass}{{\mathcal{F}}}

\newcommand{\intd}{{\mathrm{d}}}

\newcommand{\smallo}{{\mathrm{o}}}
\newcommand{\bigO}{{\mathrm{O}}}

\newcommand{\Prob}{{\mathrm{Pr}}}
\newcommand{\const}{{\mathscr{C}}}

\newcommand{\IF}{{\mathrm{IF}}}
\newcommand{\error}{{\mathrm{error}}}
\newcommand{\conf}{{\mathrm{conf}}}
\newcommand{\Gcomp}{{\mathrm{Gcomp}}}
\newcommand{\weighted}{{\mathrm{weight}}}
\newcommand{\onestep}{{\mathrm{1Step}}}
\newcommand{\rejectsample}{{\mathrm{RS}}}

\newcommand{\tmle}{{\mathrm{TMLE}}}

\newcommand{\train}{{\mathrm{train}}}
\newcommand{\test}{{\mathrm{test}}}

\newcommand\independent{\protect\mathpalette{\protect\independenT}{\perp}}
\def\independenT#1#2{\mathrel{\rlap{$#1#2$}\mkern2mu{#1#2}}}

\usepackage{custom_tex}
\renewcommand\epsilon{\ep}

\begin{document}

\maketitle

\begin{abstract}
    Predicting sets of outcomes---instead of unique outcomes---is a promising solution to uncertainty quantification in statistical learning. Despite a rich literature on constructing prediction sets with statistical guarantees, adapting to unknown covariate shift---a prevalent issue in practice---poses a serious
    unsolved challenge.
    In this paper, we 
    show that prediction sets with finite-sample coverage guarantee are uninformative and
    propose a novel flexible distribution-free method, PredSet-1Step, to efficiently construct prediction sets with an asymptotic coverage guarantee under unknown covariate shift.
    We formally show that our method is \textit{asymptotically probably approximately correct}, having well-calibrated coverage error with high confidence for large samples. 
    We illustrate that it
    achieves nominal
    coverage in a number of experiments
    and a data set concerning HIV risk prediction in a South African cohort study.
    Our theory hinges on a new bound for the convergence rate of
    the coverage of Wald confidence intervals based on general
    asymptotically linear estimators.
\end{abstract}

\tableofcontents

\section{Introduction} \label{sec: intro}

With recent advances in data acquisition, computing, and fitting algorithms, modern statistical machine learning methods can often produce accurate predictions. 
However, a key statistical challenge
is to accurately quantify the uncertainty of the predictions. 
At the moment, it remains a subject of active research how to properly 
quantify uncertainty for the most powerful algorithms, such as deep neural nets and random forests.
The difficulty is salient because 
in many applications, there are some instances whose outcomes are
intrinsically
difficult to predict accurately.
In a classification problem, for such objects, it may be more desirable to produce a small prediction set that covers the truth with high probability, instead of outputting a single prediction. 
Reliable prediction sets can be especially important in safety-critical applications, such as in medicine \protect\citep{Kitani2012,Moja2014,Bojarski2016,Berkenkamp2017,Gal2017,Ren2017,Malik2019}.
The idea of such prediction sets has a rich statistical history dating back at least to the pioneering works of \protect\cite{Wilks1941}, \protect\cite{Wald1943}, \protect\cite{scheffe1945non}, and \protect\cite{tukey1947non,tukey1948nonparametric}.

To address this challenge, 
there is an emerging body of work on constructing prediction sets with coverage guarantees under various assumptions
\protect\citep[see, e.g.,][]{Bates2021,Chernozhukov2018,dunn2018distribution,Lei2014,lei2013distribution,lei2015conformal,Lei2018,Park2020,Sadinle2019}. 
Most of these methods have theoretical coverage guarantees when the data distribution for which the predictions are constructed matches that from which the predictive model was generated.
Among these, one of the best known methods is conformal prediction (CP) \protect\citep[see, e.g.,][]{saunders1999transduction,vovk1999machine,vovk2005algorithmic, Chernozhukov2018,dunn2018distribution,Lei2014,lei2013distribution,Lei2018}. Conformal prediction can guarantee a high probability of covering a new observation, where the probability is marginal over the entire dataset and the new observation.

Moreover, inductive conformal prediction \protect\citep{papadopoulos2002inductive}---where the data at hand is split into a training set and a calibration set, satisfies a \textit{training-set conditional}, or \textit{probably approximately correct} (PAC) guarantee \protect\citep{Vovk2013,Park2020}. 
A prediction set learned from data is PAC if, over the randomness in the data, there is a high probability that its coverage error is low for new observations. 
This guarantee decouples the randomness in data at hand and the randomness in new observations. This allows a more fine-tuned control over the probability of error.
This guarantee is a generalization of the notion of tolerance regions of \protect\cite{Wilks1941} and \protect\cite{Wald1943} to the setting of supervised learning. As a generalization in another direction, the method in \protect\cite{Bates2021} provides risk-controlling prediction sets, which have low prediction risk with high probability over the randomness in the data.

The aformentioned methods are valid when the new observation and the data at hand are drawn from the same population, but this condition might fail to hold in applications. This phenomenon has been referred to in statistical machine learning as \textit{dataset shift} \protect\citep[see, e.g.,][]{quinonero2009dataset,shimodaira2000improving,Sugiyama2012}.
More specifically, an important form of dataset shift is \textit{covariate shift}: a change of only the distribution of input covariates (or features), with an unchanged distribution of the outcome given covariates. For example, the shift may arise due to a change in the sampling probabilities of different sub-populations or individuals in surveys or designed experiments.
Another setting is the assessment of future risks based on current data, such as predicting an individual's risk of a disease based on the patient's features. Here the features can shift (e.g., as the conditions of the patient change), but the distribution of the outcome given the features may be unchanged \protect\citep{quinonero2009dataset}.
Other examples of covariate shift include changes in the color and lighting of image data \protect\citep{Hendrycks2019}, or even adversarial attacks that slightly perturb the data points \protect\citep{Szegedy2014}. In both cases, the distribution of labels given input covariates is unchanged.

A concrete example of covariate shift arises in a data set concerning HIV risk prediction in a South African cohort that was analyzed in \protect\cite{Tanser2013}, and is also studied in our paper. 
The empirical distribution of HIV prevalence across communities in the source (urban and rural communities) and target populations (peri-urban communities) are presented in Figure~\protect\ref{fig: cov shift main}, and a severe covariate shift is present. The distributions of the outcome given covariates in the two populations appear to be similar.

Another example of covariate shift arises in causal inference. As discussed in \protect\cite{Lei2021} and studied in this paper, under standard causal assumptions, predicting counterfactual (hypothetical) outcomes can be formulated as a prediction problem under covariate shift. In this setup, the two covariate distributions are those in the two treatment groups (treated and untreated), which may be different in observational data due to confounding.

\begin{figure}
    \centering
    \includegraphics[scale=0.7]{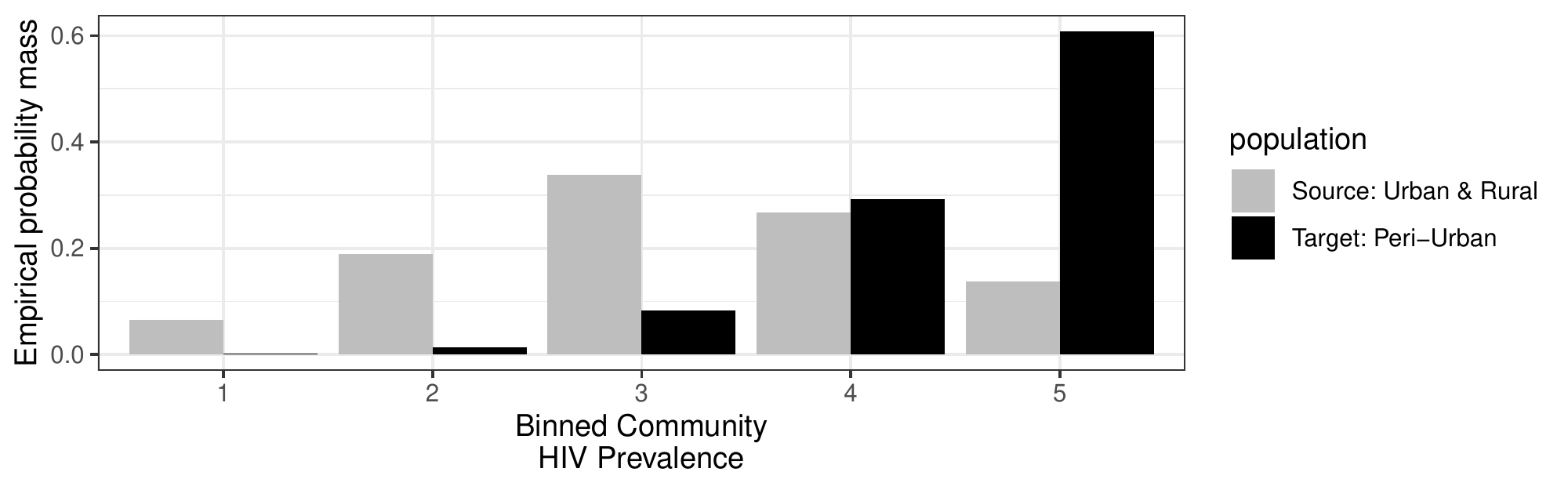}
    \caption{Empirical distribution of a covariate (binned community HIV prevalence with categories encoded by 1--6) in the two populations of the data concerning HIV risk prediction in a South African cohort.}
    \label{fig: cov shift main}
\end{figure}

\begin{figure}
    \centering
    \includegraphics[scale=0.3]{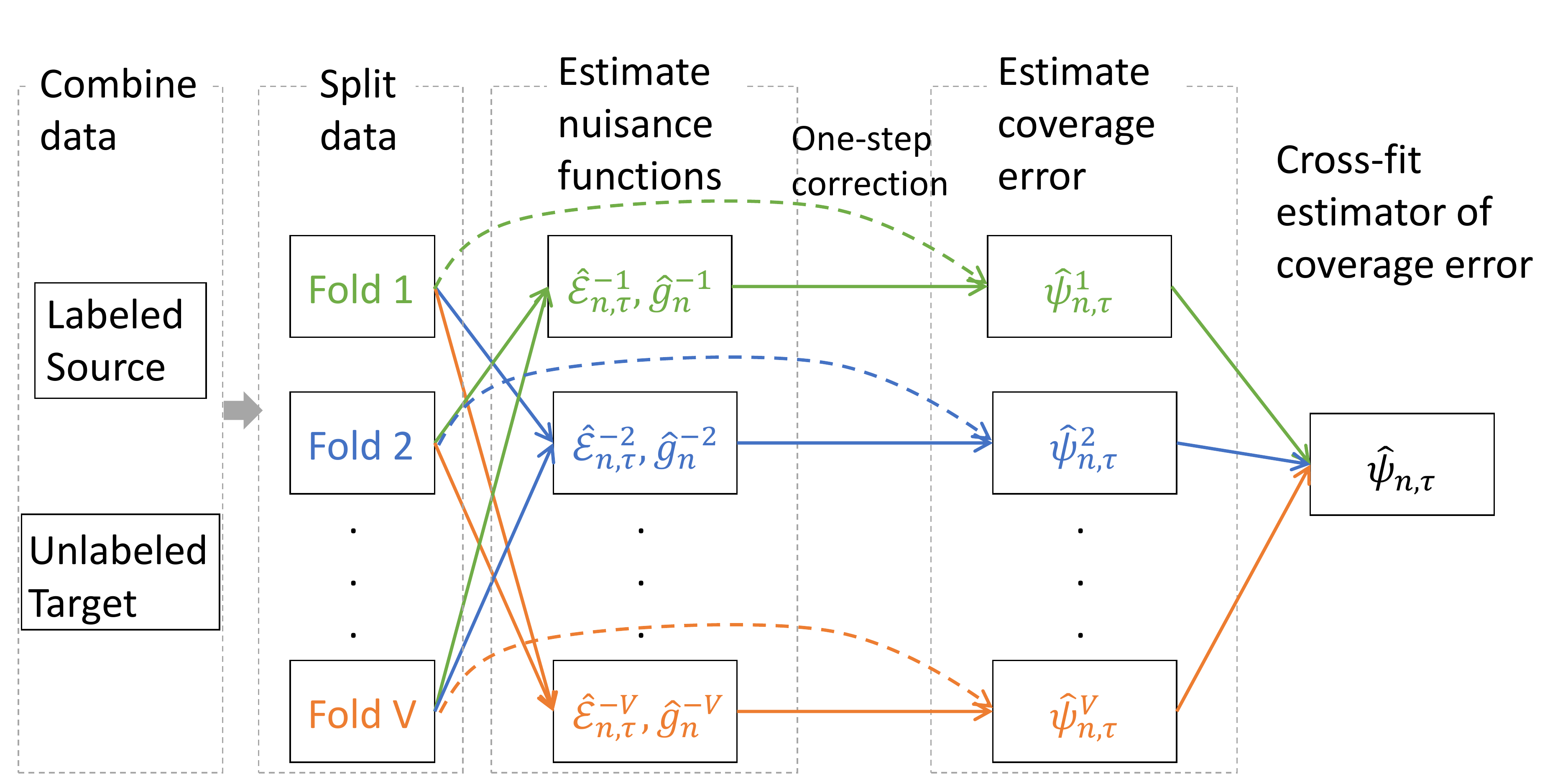}
    \caption{Overall procedure of the cross-fit one-step corrected estimator $\psi_{n,\tau}$ of the coverage error corresponding to the prediction set with threshold $\tau$, which forms the basis of our PredSet-1Step method.}
    \label{fig: CV-one step flowchart}
\end{figure}

In presence of covariate shift, prediction coverage guarantees may not hold if one assumes no covariate shift.
Possible solutions have only recently formally been studied.
\protect\cite{Tibshirani2019} studied conformal prediction under covariate shift, assuming that the likelihood ratio is known \textit{a priori}.
\protect\cite{park2021pac} studied the PAC property of inductive conformal prediction (or, PAC prediction sets) under covariate shift. 
Their methods rely on knowing the covariate shift, i.e., the likelihood ratio of the covariate distribution in the target population to that in the source population, or on bounding its smoothness, which may not always be practical.

\protect\cite{Cauchois2020} studied conformal prediction that is robust to a specified level of deviation of the target population from the source population. On the other hand, \protect\cite{Lei2021} studied conformal prediction under covariate shift without assuming that the likelihood ratio is known, and allowed estimation of this ratio instead. 
In this paper, we focus on 
the PAC property.
In inductive conformal prediction under no covariate shift or known covariate shift, 
the PAC property can be obtained even though this method was developed to obtain marginal validity (see \protect\cite{Vovk2013} for the case without covariate shift and \protect\cite{park2021pac} for the case with known covariate shift). 
However, to our knowledge, 
PAC property results for inductive conformal prediction under completely unknown covariate shift have not yet been obtained.

In this paper, we focus on achieving a PAC guarantee and show that PAC prediction sets under unknown covariate shift are uninformative. We next propose novel methods to construct prediction sets that are \textit{asymptotically PAC} (APAC) as the sample size grows to infinity, with a convergence rate that we unravel.
Our main method, \textit{PredSet-1Step}, is based on asymptotically efficient one-step corrected estimators of the true coverage error and the associated Wald confidence intervals. The procedure to construct the estimator is illustrated in Figure~\protect\ref{fig: CV-one step flowchart}, and the procedure to construct prediction sets afterwards is illustrated in Figure~\protect\ref{fig: illustrate}
in the Supplemental Material (see notations in the rest of this paper). 
PredSet-1Step heavily relies on semiparametric efficiency theory \protect\citep[see, e.g.,][]{levit1974optimality,Pfanzagl1985,Pfanzagl1990,Newey1990,vandervaart1991,Bickel1993,vandervaart1996,Bickel1993,Chernozhukov2018debiasedML,kennedy2022semiparametric} to obtain improved convergence rates.
PredSet-1Step may also be used to construct asymptotically risk-controlling prediction sets \protect\citep{Bates2021}.

This paper is organized as follows. We introduce the problem setup, present a negative result on PAC prediction sets, and present identification results under unknown covariate shift in Section~\protect\ref{sec: setup}.
In Section~\protect\ref{sec: overview methods}, we provide an overview of our proposed methods. 
We describe our method to estimate the likelihood ratio, and present pathwise differentiability results of the miscoverage in the target population; akin to those of \cite{Hahn1998}. These form the basis of our proposed methods.
We next describe our proposed PredSet-1Step method, which builds on cross-fitting/double machine learning \citep{schick1986asymptotically, Chernozhukov2018debiasedML}, along with its theoretical properties, in Section~\protect\ref{sec: efficient estimation method}. We show in Corollary~\protect\ref{corollary: CV one-step APAC} that PredSet-1Step yields APAC prediction sets with an error in the confidence level that is typically of order $n^{1/4}$ multiplied by the square root of the product of the convergence rates of estimators of two nuisance functions. These results are based on a novel bound on the difference between the realized and nominal coverage for Wald confidence intervals based on general asymptotically linear estimators (Theorem~\protect\ref{thm: general CI coverage}). We then present simulation studies in Section~\protect\ref{sec: simulation} and data analysis results in Section~\protect\ref{sec: data analysis}.

We present further results in the Supplemental Material. In Section~\protect\ref{sec: rejection sampling method},
we propose an extension, PredSet-RS, of the rejection sampling method from \protect\citet{park2021pac}.
In Section~\protect\ref{sec: TMLE},
we present an alternative approach to PredSet-1Step, PredSet-TMLE, a targeted maximum likelihood estimation (TMLE) implementation of our efficient influence function based approach \protect\citep{VanderLaan2006}.
We describe two methods to construct asymptotically risk-controlling prediction sets in Section~\protect\ref{sec: ARCPS}.
These methods are slightly modified versions of PredSet-1Step and PredSet-TMLE. The proofs of our theoretical results can be found in Section~\protect\ref{sec: proof}. 
We discuss
the PAC property, comparing it with marginal validity, in Section~\protect\ref{sec: discuss PAC}. 
We finally clarify a connection between causal inference and covariate shift, based on which we may apply methods for covariate shift to obtain well-calibrated prediction sets for individual treatment effects (ITEs), in Section~\protect\ref{section: causal and covariate shift}.
Our proposed methods are implemented in an R package available at \url{https://github.com/QIU-Hongxiang-David/APACpredset}.

\section{Problem setup and assumptions} \label{sec: setup}

\subsection{Basic setting}

Suppose one has observed labeled data from a \textit{source population}, and unlabeled data from a \textit{target population}.
Denote a prototypical full (but unobserved) data point as $\bar{O} := (A,X,Y) \sim \bar{P}^0$, where $A \in \{0,1\}$ is the indicator of the data point being drawn from the source population ($A=1$) or the target population ($A=0$), $X \in \mathcal{X}$ are the covariates, and $Y \in \mathcal{Y}$ is the outcome, label or dependent variable to be predicted. The observed data points are of the form $O:= (A,X,AY) \sim P^0$. 
In other words, in the observed data, outcomes (dependent variables) are observed only from the source population, and are missing from the target population (encoded as zero for notational convenience).

The observed data consists of $n$ independently and identically distributed (i.i.d.) observed data points $O_i \sim P^0$ ($i \in [n]:=\{1,2,\ldots,n\}$). Let $s: \mathcal{X} \times \mathcal{Y} \rightarrow \real$ be a given scoring function. For example, when $Y$ is a discrete variable, $s(x,y)$ may be an estimator of the probability of $Y=y$ given $X=x$ that has been trained from a held-out data set drawn from the source population. When $Y$ is a continuous variable, $s(x,y)$ may be an estimator of the conditional density of $Y$ at $y$ given $X=x$ or $-|y-\hat{y}(x)|$ for a given prediction model $\hat{y}$.
The function $s$ can be arbitrary user-specified mapping.
We treat $s$ as a fixed function throughout this paper; as shown in the above examples, in practice $s$ can be learned from a separate training set.

Let $\mathscr{B} \subseteq 2^\mathcal{Y}$ be the Borel $\sigma$-algebra of $\mathcal{Y}$, which is assumed to be a topological space.
We refer to a map $C: \mathcal{X} \mapsto \mathscr{B}$ that assigns to each input $x\in \mathcal{X}$ a prediction set simply as a \textit{prediction set}.
Our goal is to construct a prediction set
that is
\textit{asymptotically probably approximately correct} (APAC)
in the target population.
In other words, the prediction set should be asymptotically training-set-conditionally valid
in the target population. 

To be more precise, we first review a few related concepts. A prediction set $C$ is \textit{approximately correct} in the target population if the true coverage error in the target population, $\Prob_{\bar{P}^0}(Y \notin C(X) \mid A=0)$, is less than or equal to a given target upper bound $\alpha_\error \in (0,1)$.

An estimated prediction set $\hat{C}$ constructed from the data is \textit{probably approximately correct} (PAC) in the target population, 
with miscoverage level (also termed \textit{content}) $\alpha_\error$ and confidence level 
$1-\alpha_\conf$ ($\alpha_\conf \in (0,1)$), if, for a $\hat{C}$-independent draw $(X,Y)$ from the target population, 
$$\Prob_{P^0}(\Prob_{\bar{P}^0}(Y \notin \hat{C}(X) \mid A=0,\hat{C}) \leq \alpha_\error) \geq 1-\alpha_\conf.$$ 
In other words, $\hat{C}$ is PAC if we have confidence at least $1-\alpha_\conf$ that the true coverage error of the estimated prediction set $\hat{C}$ in the target population is below the desired level $\alpha_\error$.

\begin{remark}
    For conciseness in notations, in the rest of the paper, we may drop the distribution over which a probability is taken over when this distribution is clear (e.g., $\bar{P}^0$ or $P^0$) from the context. For example, we may write the above PAC guarantee as $\Prob(\Prob(Y \notin \hat{C}(X) \mid A=0,\hat{C}) \leq \alpha_\error) \geq 1-\alpha_\conf$.
\end{remark}

Methods to construct PAC prediction sets under covariate shift have been proposed when $Y$ is observed in data points drawn from the target population, or when the distribution shift from the source to the target population is known \protect\citep{Vovk2013,Park2020,park2021pac}. 
RCPS have also been constructed, without considering covariate shift \protect\citep{Bates2021}, while the problem with covariate shift has not been addressed, to our knowledge.

However, in our setting, neither the outcomes from the target population nor the distribution shift is known. Due to these unknown nuisance parameters, we have the following negative result on nontrivial prediction sets with a finite-sample marginal or PAC  coverage guarantee.
\begin{lemma} \label{lemma: trivial finite sample prediction set}
    Suppose that $\mathcal{X}$ and $\mathcal{Y}$ are Euclidean spaces. Let $\bar{\modelspace}^*$ be the set of all distributions $\bar{P}^0$ on the full data point $\bar{O}$ such that unknown covariate shift is present (namely Conditions~\ref{cond: positivity of P(A)}--\ref{cond: target dominated by source} in Section~\ref{sec: identification} hold), and the joint distribution of $(X,Y)$ is absolutely continuous with respect to the Lebesgue measure on $\mathcal{X} \times \mathcal{Y}$.
    Suppose that a (possibly randomized) prediction set $\hat{C}$ is PAC in the target population, that is,
    $$\Prob \left( \Prob(Y \notin \hat{C}(X) \mid A=0,\hat{C}) \leq \alpha_\error \right) \geq 1-\alpha_\conf$$
    for all $\bar{P^0} \in \bar{\modelspace}^*$. Then, for any $\bar{P^0} \in \bar{\modelspace}^*$ and a.e. $y \in \mathcal{Y}$ with respect to the Lebesgue measure,
    $$\Prob (y \notin \hat{C}(X) \mid A=0) \leq \alpha_\error+\alpha_\conf.$$
\end{lemma}
If $\alpha_\error+\alpha_\conf<1$,
Lemma~\protect\ref{lemma: trivial finite sample prediction set} indicates that any PAC prediction set $\hat{C}$ in the target population under unknown covariate shift is essentially uninformative since it will contain almost any possible outcome with a nonzero probability for any data-generating mechanism.
This lack of information can be clearly seen in the simple illustrative case where the support of $Y$ is $\real$ and $Y \independent X$. In this case, it would be desirable to obtain a PAC prediction set that outputs, for example, an estimated central or highest-density $1-\alpha_\error$ probability region of the distribution $Y \mid A=0$. 
However, Lemma~\protect\ref{lemma: trivial finite sample prediction set} implies that such a PAC prediction set does not exist, and that a PAC prediction set would instead cover almost every $y \in \real$ with probability at least $1-(\alpha_\error+\alpha_\conf)$ 
with respect to $X$. 
A similar negative result holds when $Y$ is discrete. We prove this lemma by (i) obtaining a similar negative result for prediction sets with finite-sample marginal coverage guarantees (Lemma~\protect\ref{lemma: trivial finite sample prediction set marginal}
in the Supplemental Material), and (ii) using Theorem~2 and Remark~4 in \protect\citet{Shah2020}. The proof can be found in Section~\protect\ref{sec: proof trivial finite sample prediction set} in the Supplemental Material.

Because of this negative result on finite-sample coverage guarantee, in this paper we choose to relax the validity criterion to an asymptotic one. 
It turns out that this way,
we can account for unknown covariate shift.
Recall that $n$ is the sample size used to estimate the prediction set.
\begin{definition}
    A sequence of estimated prediction sets $(\hat{C}_n)_{n\ge 1}$ is asymptotically probably approximately correct (APAC) if
\begin{equation} \label{eq: APAC statement}
    \Prob(\Prob(Y \notin \hat{C}_n(X) \mid A=0,\hat{C}_n) \leq \alpha_\error) \geq 1-\alpha_\conf+\smallo(1)
\end{equation}
as $n \rightarrow \infty$, where the $\smallo(1)$ term tends to zero as $n \rightarrow \infty$.
\end{definition}
In other words, a sequence of APAC prediction sets $\hat{C}_n$ is almost PAC for sufficiently large $n$.
We will further quantify the magnitude of the $o(1)$ error in the confidence level.
We use an estimated prediction set $\hat{C}_n$ and a sequence $(\hat{C}_n)_{n\ge 1}$ interchangeably in this paper and may say that $\hat{C}_n$ is APAC.
Further, we treat $\alpha_\error$ and $\alpha_\conf$ as fixed.

\begin{remark}
    \textit{Risk-controlling  prediction set} (RCPS) \protect\citep{Bates2021} is more general than but similar to PAC. Our proposed PredSet-1Step method can be readily applied to constructing asymptotic RCPS (ARCPS) with a slight modification. We introduce the concept of ARCPS and describe the modified method in Section~\protect\ref{sec: ARCPS} in the Supplemental Material.
\end{remark}

\protect\citet{Vovk2013} derived the PAC property of inductive conformal predictors \protect\citep{papadopoulos2002inductive}, and \protect\citet{Park2020} presented a nested perspective \protect\citep{vovk2005algorithmic}.
While the formulations of \protect\citet{Vovk2013} and \protect\citet{Park2020} are equivalent, we follow \protect\citet{Park2020}.
For thresholds $\tau \in \bar{\real}$, where $\bar{\real} := \real \cup \{\pm \infty\}$,
we consider nested prediction sets \protect\citep{vovk2005algorithmic} of the form 
$$C_\tau: x \mapsto \{y \in \mathcal{Y}: s(x,y) \geq \tau\}.$$ 
Since $C_{\tau_1}(x) \subseteqq C_{\tau_2}(x)$ for any $x$ and any $\tau_1 \ge \tau_2$, typical measures of the size of $C_\tau$ (such as the cardinality or Lebesgue measure) are nonincreasing functions of $\tau$. 
Therefore, to obtain an APAC prediction set with a small size, given a finite set $\mathcal{T}_n \subseteqq \bar{\real}$ of candidate thresholds, we select a threshold $\hat{\tau}_n \in \mathcal{T}_n$ such that $\Prob(\Prob(Y \notin C_{\hat{\tau}_n}(X) \mid A=0) \leq \alpha_\error) \geq 1-\alpha_\conf+\smallo(1)$ and $\hat{\tau}_n$ is as large as possible. Our methods
under this setting are our main contribution.

We summarize our main result informally. Formal results can be found in later sections. The algorithm to estimate the coverage error $\Psi_\tau(P^0)=\Prob(Y \notin C_\tau(X) \mid A=0)$ 
corresponding to the threshold $\tau$ used in PredSet-1Step is Algorithm~\protect\ref{alg: CV one step}. We will show in Corollary~\protect\ref{corollary: CV one-step APAC} that, under certain conditions, the prediction set $C_{\hat{\tau}^\onestep_n}$ with the threshold $\hat{\tau}^\onestep_n$ selected by PredSet-1Step is APAC:
$$\Prob(\Psi_{\hat{\tau}^\onestep_n}(P^0) \leq \alpha_\error) \geq 1-\alpha_\conf- \const \Delta_{n,\epsilon},$$
where $\const$ is an absolute positive constant and $\Delta_{n,\epsilon}$ is typically of order $n^{1/4}$ multiplied by the square root of the product of the convergence rates of two nuisance function estimators. 
This 
corollary
relies on a novel result on the convergence rate of Wald confidence interval coverage for general asymptotically linear estimators (Theorem~\protect\ref{thm: general CI coverage}). 
The result bounds the difference between the true and the nominal coverage by three error terms: (i) the difference between the estimator and a sample mean, (ii) the estimation error of the asymptotic variance, and (iii) the difference of the distribution of the sample mean from its limiting normal distribution.

\begin{remark} \label{rmk: compare with moment equation}
Beyond the APAC criterion, an alternative approach is to find a prediction set $C$ that approximately solves
$$\min \quad \mathrm{size}(C) \qquad \text{subject to} \quad \widehat{\Prob}(Y \notin C(X) \mid A=0) \leq \alpha_\error,$$
where $\widehat{\Prob}(Y \notin C(X) \mid A=0)$ is an estimator of $\Prob(Y \notin C(X) \mid A=0)$ and $\mathrm{size}(C)$ is a measure of the size of the prediction set $C$. 
This approach has been considered in \protect\citet{Yang2022}, and generally results in
smaller prediction sets than the APAC ones we consider in this paper.
The reason is that the APAC guarantee requires approximately controlling the confidence level $1-\alpha_\conf$ to achieve the desired coverage error level $\alpha_\error$ over the data, which leads to some conservativeness.
This difference can also been seen from the PAC guarantee in \protect\citet{Yang2022} taking the form
\begin{equation} \label{eq: PAAC statement}
    \Prob(\Prob(Y \notin \hat{C}_n(X) \mid A=0,\hat{C}_n) \leq \alpha_\error + \smallo_p(1)) \geq 1-\alpha_\conf.
\end{equation}
To distinguish from the APAC guarantee in \protect\eqref{eq: APAC statement}, we call the guarantee in \protect\eqref{eq: PAAC statement} a \textit{probably asymptotically approximately correct} (PAAC) guarantee. 
The difference between APAC \protect\eqref{eq: APAC statement} and PAAC \protect\eqref{eq: PAAC statement} is in the asymptotically vanishing approximation error: 
in APAC \protect\eqref{eq: APAC statement}, the approximation is on the confidence level; in PAAC \protect\eqref{eq: PAAC statement}, the approximation is on the coverage error. This difference may seem subtle but has substantial impact on the performance of prediction sets satisfying these guarantees. 
We illustrate this difference by interpreting APAC \protect\eqref{eq: APAC statement} and PAAC \protect\eqref{eq: PAAC statement} in words: 
APAC \protect\eqref{eq: APAC statement} states that, with confidence approaching the desired level $1-\alpha_\conf$, 
the true coverage error does not exceed the desired level $\alpha_\error$, but may frequently be a little conservative; PAAC \protect\eqref{eq: PAAC statement} states that, with confidence at least $1-\alpha_\conf$, the true coverage error does not exceed the desired level $\alpha_\error$ by much, but may frequently exceed $\alpha_\error$ by a little. We also illustrate the difference in Figure~\protect\ref{fig: APAC PAAC}.
In some applications, having a high confidence guarantee on the desired level $\alpha_\error$ of true coverage error at a price of slight conservativeness may be desirable, for example, for safety purposes.
\end{remark}

\begin{figure}
    \centering
    \includegraphics[scale=0.7]{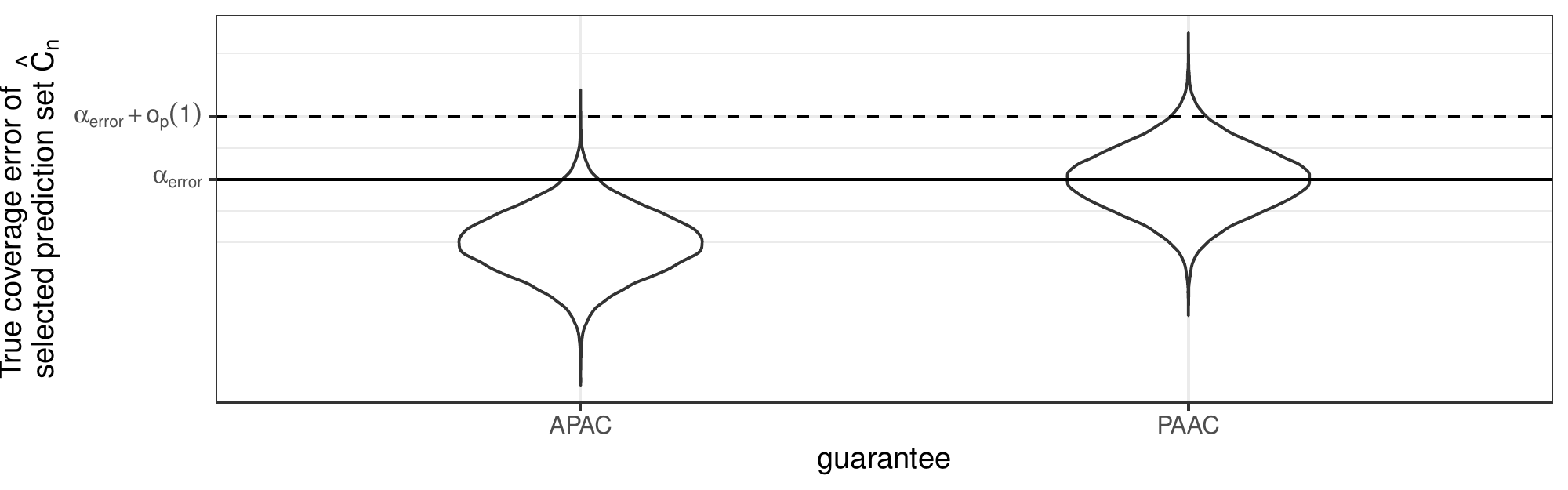}
    \caption{Exemplar sampling distributions of the true coverage error $\Prob(Y \notin \hat{C}_n(X) \mid A=0,\hat{C}_n)$ of prediction sets $\hat{C}_n$ satisfying APAC and PAAC guarantees, respectively.} \label{fig: APAC PAAC}
\end{figure}

We conclude this section by introducing a few more notations. We use $\const$ to denote an absolute positive constant that may vary line by line. For two scalar sequences $(a_n)_{n\ge 1}$ and $(b_n)_{n\ge 1}$, we use $a_n \lesssim b_n$ to denote that for some constant $\const>0$ and all $n \ge 1$, $a_n \leq \const b_n$, and we define $\gtrsim$ similarly. We use $a_n \simeq b_n$ to denote that both $a_n \lesssim b_n$ and $a_n \gtrsim b_n$ hold. We also adopt the little-o and big-O notations. For a probability distribution $P$ and a scalar $p\ge 1$, we use $\| \cdot \|_{P,p}$ to denote the $L^p(P)$-norm of a function.

\subsection{Identification} \label{sec: identification}

Without any further assumptions, it is impossible to estimate $\Prob(Y \notin C(X) \mid A=0)$ for an arbitrary prediction set $C$, since the joint distribution $(X,Y) \mid A=0$ of $(X,Y)$ in the target population cannot be identified due to $Y$ missing in the data.
We make a few assumptions, following the standard setting in the covariate shift literature \protect\citep[see, e.g.,][]{shimodaira2000improving,quinonero2009dataset,Sugiyama2012}, so that $\Prob(Y \notin C(X) \mid A=0)$ can be identified as a functional of the true distribution $P^0$ on the observed data.

Let $P^0_A$ denote the marginal distribution of $A$ under $P^0$, $P^0_{X \mid a}$ denote the distribution of $X \mid A=a$ under $P^0$ for $a \in \{0,1\}$, $\bar{P}^0_{Y \mid x,a}$ the distribution of $Y \mid X=x,A=a$ under the full data distribution $\bar{P}^0$, and $P^0_{Y \mid x} :=\bar{P}^0_{Y \mid x,1}$. 
For any distribution $P$ of the observed data point $O$, we define these marginal and conditional distributions similarly and denote them with similar notations except that the superscript $0$ denoting $P^0$ is dropped. For example, $P_{X \mid a}$ stands for the distribution of $X \mid A=a$ under $P$. It will also be convenient to define the loss function 
$$Z_\tau: (x,y) \mapsto \ind(y \notin C_\tau(x))$$ for any $\tau \in \bar{\real}$.
Our first condition is: 
\begin{condition}[Data available from both populations] \label{cond: positivity of P(A)}
$0 < \Prob(A=1) < 1$.
\end{condition}

This condition ensures that data points from both source and target populations are collected in sufficient quantity, and that the conditional distributions introduced above are well defined. In practice, this condition requires that a reasonable amount of data from both populations is collected.

Next, we state the key \textit{covariate shift} assumption \protect\citep[see, e.g.,][]{shimodaira2000improving,quinonero2009dataset,Sugiyama2012}, which is central to our paper.

\begin{condition}[Covariate shift: Identical conditional outcome distribution] \label{cond: same Y|X}
The conditional distribution of $Y \mid X=x$ in the target population is identical to that in the source population for all $x \in \mathcal{X}$.\footnote{Formally, this has to hold almost surely with respect to a given probability measure over $\mathcal{X}$, with respect to which all distributions of $X$ considered are absolutely continuous; however, we simplify the statement for clarity.} Mathematically, $\bar{P}^0_{Y \mid x,1}=\bar{P}^0_{Y \mid x,0}=P^0_{Y \mid x}$.
\end{condition}

Condition~\protect\ref{cond: same Y|X} is similar to the missing at random assumption in the missing data literature \protect\citep[see, e.g.,][]{Little2019}. 
It holds, for instance, if in the target we observe the same $Y$ (e.g., does a car face left or right), with $X$ from a different distribution (images from cities A vs B).
Finally, we have an assumption to ensure that the target population overlaps with the source population.

\begin{condition}[Dominance of covariate distributions] \label{cond: target dominated by source}
The marginal distribution of $X$ in the target population, $P^0_{X \mid 0}$, is dominated by that in the source population, $P^0_{X \mid 1}$; that is, the Radon-Nikodym derivative $w_0 := \intd P^0_{X \mid 0}/\intd P^0_{X \mid 1}$ is well defined.
\end{condition}

We assume that Conditions~\protect\ref{cond: positivity of P(A)}--\protect\ref{cond: target dominated by source} hold throughout this paper. 
For any distribution $P$ of the observed data point $O$ satisfying Condition~\protect\ref{cond: target dominated by source} and any $\tau \in \bar{\real}$, we define the functionals
\begin{align} \label{qdef}
    &w_P := \intd P_{X \mid 0}/\intd P_{X \mid 1} \quad \text{and} \quad \mathcal{E}_{P,\tau} : x \mapsto
    \Prob_P(Y \notin C_\tau(X) \mid X=x,A=1).
\end{align}
We will also replace $P$ in the subscripts of these and other quantities with $0$ when referring to the functional components of $P^0$.
Here, $w_P$ is the likelihood ratio between target and source covariate distributions under $P$, and $\mathcal{E}_{P,\tau}$ is the covariate-conditional coverage error of the prediction set $C_\tau$ in the source population.
It is not hard to show that, under the distribution $\bar{P}$ of the complete but unobserved data, we can also express $\mathcal{E}_{P,\tau}$ in terms of the target population as $\mathcal{E}_{P,\tau}(x)=\Prob_{\bar{P}}(Y \notin C_\tau(X) \mid X=x,A=0)$.
Further, we define
\begin{align*}
    &\Psi^\Gcomp_\tau: P \mapsto \expect_P[\mathcal{E}_{P,\tau}(X) \mid A=0] \quad \text{and} \quad \Psi^\weighted_\tau: P \mapsto \expect_P[w_P(X) Z_\tau(X,Y) \mid A=1].
\end{align*}
One can verify that for $j \in \{\Gcomp,\weighted\}$,
\begin{equation} \label{eq: identification}
    \Psi^j_\tau(P^0) = \Prob_{\bar{P}^0}(Y \notin C_\tau(X) \mid A=0). 
\end{equation}
In other words, although $\Psi^\Gcomp_\tau(P^0)$ and $\Psi^\weighted_\tau(P^0)$ take as inputs 
different components of $P^0$, both correspond to the same functional of $P^0$, the coverage error of the prediction set $C_\tau$ in the target population. We will use $\Psi_\tau$ to denote these two functionals when we need not distinguish their mathematical expressions. In other words, $\Psi_\tau(P^0)$ equals the probability that $Y \notin C_\tau(X)$ in the ``covariate shifted'' population where $A=0$:
\begin{equation} \label{psitau}
    \Psi_\tau(P^0) = \Prob_{\bar{P}^0}(Y \notin C_\tau(X) \mid A=0). 
\end{equation}
Borrowing terminology from causal inference and missing data, we refer to $\Psi^\Gcomp_\tau(P^0)$ and $\Psi^\weighted_\tau(P^0)$ as the G-computation formula and the weighted formula, respectively.

\begin{remark} \label{rmk: parameter depend on components}
Both $\Psi^\Gcomp_\tau$ and $\Psi^\weighted_\tau$ take as inputs only certain components of the distribution rather than the entire distribution, and hence may be computed as long as the relevant components are defined. For example, $\Psi^\Gcomp_\tau(P)$ is defined if the distribution $P_{X \mid 0}$ and $\mathcal{E}_{P,\tau}$ are defined. 
We will specify only the required components when defining our estimators.
\end{remark}

\begin{remark} \label{rmk: causal and covariate shift}
There is a connection between counterfactuals in causal inference and covariate shift, as pointed out in \protect\citet{Lei2021}. We discuss this connection in more detail in Section~\protect\ref{section: causal and covariate shift}
in the Supplemental Material.
\end{remark}

\section{Overview and preliminaries of proposed method} \label{sec: overview methods}

In all methods we propose, 
we assume that a 
finite set $\mathcal{T}_n$ of candidate thresholds is given, with a cardinality that may grow to infinity with $n$. 
Since, as a function of $\tau$, there are at most $n+1$ versions of the observed miscoverage indicators $\{Z_\tau(X_i,Y_i) = \ind(s(X_i,Y_i) < \tau), i \in [n]\}$ in any data set, each corresponding to a threshold in the set $\{s(X_i,Y_i): i \in [n]\} \cup \infty$, this assumption is not stringent. 

Our general strategy is to construct an asymptotically valid $(1-\alpha_\conf)$-confidence upper bound (CUB) for $\Psi_\tau(P^0)$ for each threshold $\tau \in \mathcal{T}_n$, and select the largest threshold $\hat{\tau}_n \in \mathcal{T}_n$ such that, for any candidate threshold less than or equal to $\hat{\tau}_n$, the corresponding CUB is less than $\alpha_\error$. This procedure is illustrated in Figure~\protect\ref{fig: illustrate}
in the Supplemental Material.
To construct accurate approximate confidence intervals (CIs), we rely on semiparametric efficiency theory \protect\citep{bickel1982adaptive,Pfanzagl1985,Pfanzagl1990,Newey1990,vandervaart1991,Bickel1993,vanderVaart1998,Bickel1993,kennedy2022semiparametric}.

\subsection{Estimation of nuisance functions} \label{sec: reparameterize weight}

For a given threshold $\tau$, we will see in Section~\protect\ref{sec: pathwise differentiability} that it is helpful to estimate nuisance functions corresponding to the pointwise coverage error
\begin{equation}\label{q0tau}
\mathcal{E}_{0,\tau} = \mathcal{E}_{P^0,\tau} : x \mapsto \Prob_{P^0}(Y \notin C_\tau(X) \mid X=x,A=1)
\end{equation} 
and the covariate shift likelihood ratio $w_0$ from Condition \protect\eqref{cond: target dominated by source}.
An estimator $\mathcal{E}_{n,\tau}$ of $\mathcal{E}_{0,\tau}$ may be obtained with standard classification or regression algorithms in the subsample from the source population with dependent variable $Z_\tau(X,Y)$ and covariate $X$.

However, for the estimation of the likelihood ratio $w_0$, we opt for a re-parametrization to a classification problem. 
Inspired by \protect\citet{friedman2004multivariate}, \protect\citet{Bickel2007}, \protect\citet{Sugiyama2008}, and \protect\citet{Menon2016}, we use the following observation from Bayes' Theorem. 
For any distribution $P$ of the observed data point $O$ satisfying Condition~\protect\ref{cond: target dominated by source}, define $g_P : x \mapsto [0,1]$ and $\gamma_P \in (0,1)$ via
\begin{equation}\label{gpgammap}
g_P(x) := \Prob_{P}(A=1 \mid X=x),\qquad \gamma_P:=\Prob_{P}(A=1).
\end{equation}
We define $g_0$ and $\gamma_0$ similarly for $P^0$:
\begin{equation}\label{g0gamma0}
g_0(x) := \Prob_{P^0}(A=1 \mid X=x),\qquad \gamma_0:=\Prob_{P^0}(A=1).
\end{equation}
Following terminology in causal inference, we call $g_0$ the \textit{propensity score function} \protect\citep{Rosenbaum1983}. When referring to generic propensity score functions and probabilities, we will write $g$ and $\gamma$ instead of $g_P$ and $\gamma_P$.
For any propensity score function $g$ and any probability $\gamma \in (0,1)$, we define $\mathscr{W}(g,\gamma):\mathcal{X} \to [0,\infty)$ as
\begin{equation}\label{mw}
\mathscr{W}(g,\gamma)(x) := \frac{1-g(x)}{g(x)} \frac{\gamma}{1-\gamma}. 
\end{equation}
Bayes' theorem directly shows that
\begin{equation} \label{eq: reparameterize weight}
    w_0(x)=\mathscr{W}(g_0,\gamma_0)(x).
\end{equation}
We will use this reparameterization through the rest of this paper.
We can estimate $\gamma_0$ by $\gamma_n$ obtained from the empirical distribution.
Further, we may estimate $g_0$ by $g_n$ obtained with standard classification or regression algorithms with dependent variable $A$ and covariate $X$. 
In our experience, existing classification techniques are more flexible in our setting than density estimation methods.
For instance, density estimation procedures might need adjustment according to the support of variables in $X$ (e.g., bounded continuous, unbounded continuous, discrete, a mixture, etc.), while most classification methods need not make this distinction.

\subsection{Pathwise differentiability} \label{sec: pathwise differentiability}

We next present results on pathwise differentiability of the error rate parameter $\Psi_\tau$ with respect to $\modelspace$, a model that is nonparametric at $P^0$, which are akin to those of \cite{Hahn1998}; see also \protect\citep{levit1974optimality,bickel1982adaptive,Pfanzagl1985,Pfanzagl1990,Newey1990,vandervaart1991,Bickel1993,vanderVaart1998,Bickel1993}. We first briefly describe the intuition behind these terminologies. Consider a generic one-dimensional parametric submodel $(P^\epsilon_{H})_{\epsilon\in\real}$ satisfying
$\intd P^\epsilon_{H}/\intd P^0(o) \approx 1 + \epsilon H(o)$
for some function $H$ with $\expect_{P^0}[H(O)]=0$ and finite variance. We only consider \textit{regular parametric submodels} \protect\citep[see, e.g.,][for more details]{Newey1990, Bickel1993}. 
The function $H$ is called the \textit{score function} of this submodel. 
We say that a model $\modelspace$ is nonparametric at $P^0$ if, for any function $H$ with mean zero and finite variance, $P^\epsilon_{H} \in \modelspace$ for $\epsilon$ sufficiently close to zero. 
Roughly speaking, $H$ encodes the direction of local perturbations of $P^0$ in the submodel, and a nonparametric model allows any perturbation of $P^0$.

We focus on nonparametric models in the main text, in which case no information about $P^0$ is known. In particular, the estimator $(\mathcal{E}_{n,\tau},w_n)$ of $(\mathcal{E}_{0,\tau},w_0)$ may converge in probability in an $L^2(P^0)$ sense at a rate slower than or equal to $n^{-1/2}$. This rate is typically slower than the parametric rate $n^{-1/2}$ as long as the covariate $X$ has continuous components.

A parameter $\Psi: \modelspace \rightarrow \real$ is pathwise differentiable if
$\intd \Psi(P^\epsilon_{H})/\intd \epsilon|_{\epsilon=0} = \expect_{P^0}[H(O) D(O)]$
for some function $D$ with $\expect_{P^0}[D(O)]=0$ and finite variance.
This function $D$ is called a \textit{gradient} of the parameter $\Psi$ at $P^0$, since it characterizes the local change in the value of the parameter corresponding to a perturbation of $P^0$. We can then heuristically expand $\Psi(P^\epsilon_{H})$ around $\Psi(P^0)$:
\begin{align}
    &\Psi(P^\epsilon_{H}) - \Psi(P^0) \approx \epsilon \int H(o) D(o) P^0(\intd o) \nonumber \\
    &= \int (1+ \epsilon H(o)) D(o) P^0(\intd o) - \int D(o) P^0(\intd o) \approx \int D(o) (P^\epsilon_{H}-P^0)(\intd o). \label{eq: intuitive expansion}
\end{align}
In nonparametric models, the gradient $D$ is unique and also called the \textit{canonical gradient}. The above explanation is informal, and we refer the readers to Section~\protect\ref{sec: proof differentiability}
in the Supplemental Material and to \protect\citet{levit1974optimality,Pfanzagl1985,Pfanzagl1990,Bickel1993} for more details.
The pathwise differentiability of an estimator is closely related to efficiency, and is crucial for the construction of a root-$n$-consistent and asymptotically normal estimator. 

An estimator is \textit{asymptotically efficient} under a nonparametric model if it equals the estimand plus the sample mean of the canonical gradient, up to an error $\smallo_p(n^{-1/2})$. 
An asymptotically efficient estimator is root-$n$ consistent and has the smallest possible asymptotic variance among a large class of estimators called \textit{regular estimators} \protect\citep[see, e.g., Section~8.5 in][]{vanderVaart1998}. 
Hence, the result on pathwise differentiability of $\Psi_\tau$ forms the basis of constructing efficient estimators of $\Psi_\tau(P^0)$, based on which approximate CUBs can be constructed under a nonparametric model.

Now we return to our prediction set problem. 
Consider an arbitrary function $\mathcal{E}$ defined on $\mathcal{X}$ with range contained in $[0,1]$, any scalar $\gamma \in (0,1)$, and any positive scalar $\pi$.
For each $\tau \in \bar{\real}$, with $o:=(a,x,y)$, we define the function
\begin{align}
    D_\tau(P,g,\gamma): o &\mapsto \frac{a}{\gamma_P} \mathscr{W}(g,\gamma)(x) \left\{ Z_\tau(x,y) - \mathcal{E}_{P,\tau}(x) \right\}
    + \frac{1-a}{1-\gamma_P} [ \mathcal{E}_{P,\tau}(x) - \Psi^\Gcomp_\tau(P) ]. \label{dtau}
\end{align}
For notational convenience, we suppress the dependence of this gradient function on the target parameter $\Psi_\tau(P)$; this dependence is implicit through the dependence on $P$.

We require an additional bounded likelihood ratio condition, which is standard in the literature on covariate shift \protect\citep{shimodaira2000improving,quinonero2009dataset,Sugiyama2012}.

\begin{condition}[Bounded likelihood ratio] \label{cond: bounded weight}
There exists a constant $B < \infty$ such that $\sup_{x \in \mathcal{X}} w_0(x) < B$. Equivalently, there exists a constant $\delta \in (0,1)$ such that the propensity score is bounded away from zero, namely $\inf_{x \in \mathcal{X}} g_0(x) > \delta$. Here, $\delta$ may be taken as $\gamma_0/(B (1-\gamma_0)+1)$.
\end{condition}

Under Condition~\protect\ref{cond: bounded weight}, it holds that $\sup_{\tau \in \bar{\real}} \expect_{P^0}[D_\tau(P^0,g_0,\gamma_0)(O)^2] < \infty$, because $Z_\tau$ is bounded. 
The pathwise differentiability of $\Psi$ is presented in the following theorem, which is a version of the results of \cite{Hahn1998} in causal inference.

\begin{theorem}[Pathwise differentiability of $\Psi_\tau$] \label{thm: differentiability of Psi}
Under 
Conditions~\protect\ref{cond: positivity of P(A)}--\protect\ref{cond: bounded weight}, for each $\tau \in \bar{\real}$, the functional $\Psi_\tau$ from \protect\eqref{psitau},
where $\Psi_\tau(P^0) = \Prob(Y \notin C_\tau(X) \mid A=0)$ is the coverage error in the target, ``covariate shifted'' population with $A=0$,
is pathwise differentiable at $P^0$ relative to $\modelspace$
with canonical gradient $D_\tau(P^0,g_0,\gamma_0)$
from \protect\eqref{dtau}.
\end{theorem}

The proof of Theorem~\protect\ref{thm: differentiability of Psi}
is related to 
the proof of Theorem~1 in \protect\citet{Hahn1998}:
with $1-A$ being the treatment indicator $D$ in \protect\citet{Hahn1998}, $\ind(Y \notin \hat{C}_\tau(X))$ being the counterfactual outcome $Y_0$ in \protect\citet{Hahn1998}, $\Psi_\tau(P^0)$ can be written as the mean counterfactual outcome $\expect[Y_0 \mid D=1]$
in the treated group in \protect\citet{Hahn1998}. Estimating this is the main challenge in estimating the average treatment effect $\expect[Y_1 - Y_0 \mid D=1]$ on the treated;
the canonical gradient of 
$\expect[\ind(Y \notin \hat{C}_\tau(X)) \mid A=0]$ can be calculated using arguments similar to the proof of Theorem~1 in \protect\citet{Hahn1998}.
We provide the proof in Section~\protect\ref{sec: proof differentiability} in the Supplemental Material.
Since both nuisance functions $\mathcal{E}_{0,\tau}$ and $w_0$ appear in the canonical gradient, it is helpful to estimate both functions in order to construct asymptotically efficient estimators of $\Psi_\tau(P^0)$ as well as asymptotically valid CUBs. As is known in the sieve estimation literature \protect\citep[see, e.g.,][]{Chen2007,Qiu2021,Shen1997} and other nonparametric inference literature \protect\citep[see, e.g.,][]{Bickel2003,Newey1998,Newey2004}, it is possible to only estimate---for example---$\mathcal{E}_{0,\tau}$, with specific nonparametric methods, and still obtain an asymptotically efficient estimator of $\Psi_\tau(P^0)$. In this paper, we do not take these approaches and propose methods that require estimating both nuisance functions in our procedures to allow for the  most generality and flexibility in choosing estimators of nuisance functions.

\begin{remark} \label{rmk: general loss differentiability}
The pathwise differentiability of $\Psi_\tau$ does not use that the loss function $Z_\tau(x,y)=\ind(y \notin C_\tau(x))$ is binary. Therefore, our approach works for general loss functions, and we may construct asymptotically efficient estimators in that setting. 
Then we can construct asymptotically efficient estimators for the true risk that corresponds to a general loss function for a prediction set under covariate shift. 
In particular, PredSet-1Step may be used with slight modifications for the estimation of the conditional risk function $\mathcal{E}_{0,\tau}$ for general losses.
We present the corresponding results for constructing ARCPS in Section~\protect\ref{sec: ARCPS}
in the Supplemental Material.
\end{remark}

\section{PredSet-1Step} \label{sec: efficient estimation method}

In this section, we describe the PredSet-1Step method, based on an asymptotically efficient one-step corrected estimator of $\Psi_\tau(P^0)$, along with its main theoretical properties. For each candidate $\tau \in \mathcal{T}_n$, we first construct an asymptotically efficient estimator $\hat{\psi}_{n,\tau}$ of $\Psi_\tau(P^0)$, and then obtain a consistent estimator $\hat{\sigma}_{n,\tau}^2$ of the asymptotic variance $\hat{\sigma}_{0,\tau}^2$ of $\hat{\psi}_{n,\tau}$. We finally construct a Wald CUB based on $\hat{\psi}_{n,\tau}$ and $\hat{\sigma}_{n,\tau}^2$.
We select thresholds $\tau \in \mathcal{T}_n$ for prediction sets based on the CUBs. We next describe each step in more detail.

\subsection{Cross-fit one-step corrected estimator} \label{sec: CV one-step}

In this section, we describe a cross-fit one-step corrected estimator of $\Psi_\tau(P^0)$ for a given $\tau$. 
After obtaining an estimator $\hat{\mathcal{E}}_{n,\tau}$ of $\mathcal{E}_{0,\tau}$ via parametric (e.g., logistic regression, neural nets) or nonparametric methods, 
it might be tempting to estimate $\Psi_\tau(P^0)$ by the mean of $\hat{\mathcal{E}}_{n,\tau}(X)$ among observations from the target population. In other words, $\hat{\mathcal{E}}_{n,\tau}$ is plugged into $\Psi^\Gcomp$. 
However, in general, this plug-in estimator may not be rate-optimal and may invalidate subsequent CUB construction and APAC guarantees. The reason is a bias term that may dominate the convergence of this estimator.

Fortunately, this excessive bias can be reduced by a one-step correction on the standard plug-in estimator (see, e.g., Theorem~4 in Chapter~3 of \protect\citet{LeCam1969} for early development of this idea for parametric models, and \protect\citet{Pfanzagl1985,schick1986asymptotically} as well as Section~25.8 in \protect\citet{vanderVaart1998} for more modern generalizations to semi-/non-parametric models.) 
We further incorporate cross-fitting into the procedure to relax restrictions on the techniques used to estimate nuisance functions $\mathcal{E}_{0,\tau}$ and $g_0$. This technique improves performance in small to moderate samples; see e.g., \cite{kennedy2022semiparametric} for a review.
More generally, such sample splitting ideas date back at least to \protect\citet{Hajek1962}.

Suppose that the data is split into $V$ folds of approximately equal size completely at random. We assume that $V \geq 2$ is fixed.
Common choices of $V$ include five and ten. 
Let $I_v \subseteqq [n]$ denote the index set of observations in fold $v \in [V]$. 
We use $P^{n,v}$
to denote the empirical distribution of data in
fold $v$.
We also use $P_{A}^{n,v}$ and $P_{X \mid a}^{n,v}$ to denote the empirical distribution of $A$ and $X \mid A=a$ corresponding to data in fold $v$. 

For each $\tau$, let $\hat{\mathcal{E}}_{n,\tau}^{-v}$ denote a flexible estimator of $\mathcal{E}_{0,\tau}$ obtained from data points out of fold $v$ via, for example, standard regression or supervised statistical learning tools. We also use $g_n^{-v}$ to denote a flexible estimator of $g_0$ obtained from data points out of fold $v$. We define
\begin{equation}\label{gammanminusv}
\hat{\gamma}_n^v := \Prob_{P^{n,v}}(A=1).
\end{equation}

We construct a cross-fit one-step corrected estimator using Algorithm~\protect\ref{alg: CV one step}. Direct calculation shows that the one-step corrected estimator $\hat{\psi}_{n,\tau}^v$ for fold $v$ from \protect\eqref{psi-n-v alg} equals
\begin{equation}
    \Psi^\Gcomp_\tau(\hat{P}_{\tau}^{n,v}) + \frac{1}{|I_v|} \sum_{i \in I_v} D_\tau(\hat{P}_{\tau}^{n,v},\hat{g}_n^{-v},\hat{\gamma}_n^v)(O_i). \label{psi-n-v}
\end{equation}
The key one-step correction in \protect\eqref{psi-n-v} based on the canonical gradient is similar to a correction based on a linear approximation using the gradient in \protect\eqref{dtau} at the estimated distribution $\hat{P}_{\tau}^{n,v}$. We roughly describe the intuition below in an informal manner, and refer the readers to Section~\protect\ref{sec: proof efficiency}
in the Supplemental Material for technical details. Following \protect\eqref{eq: intuitive expansion}, we expand $\Psi(P^0)$ around $\Psi(\hat{P}_{\tau}^{n,v})$:
\begin{align*}
    \Psi(P^0) &\approx \Psi(\hat{P}_{\tau}^{n,v}) + \int D_\tau(\hat{P}_{\tau}^{n,v},\hat{g}_n^{-v},\hat{\gamma}_n^v)(o) (P^0 - \hat{P}_{\tau}^{n,v})(\intd o) \\
    &= \Psi(\hat{P}_{\tau}^{n,v}) + \expect_{P^0}[D_\tau(\hat{P}_{\tau}^{n,v},\hat{g}_n^{-v},\hat{\gamma}_n^v)(O)],
\end{align*}
where, in the expectation in the second line, we treat $(\hat{P}_{\tau}^{n,v},\hat{g}_n^{-v},\hat{\gamma}_n^v)$ as fixed. The second equality follows because a gradient at $P$ has mean zero under $P$. Since the above correction term is unknown, we replace the expectation under $P^0$ with the empirical mean and thus find the one-step correction in \protect\eqref{psi-n-v}. This idea is illustrated in Figure~\protect\ref{fig: one step}. This one-step correction is crucial to ensure root-$n$ consistency and asymptotic normality of the estimator, as we illustrate in a simulation shown in Figure~\protect\ref{fig: one step improvement}
in the Supplemental Material.

\begin{algorithm}
\caption{Cross-fit one-step estimator of coverage error $\Psi_\tau(P^0)$ used in PredSet-1Step} \label{alg: CV one step}
\begin{algorithmic}[1]
\IFor{$v \in [V]$}
    Estimate $g_0$ by $\hat{g}_n^{-v}$ using data out of fold $v$.
\EndIFor
\IFor{$v \in [V]$ and $\tau \in \mathcal{T}_n$}
    Estimate $\mathcal{E}_{0,\tau}$ by $\hat{\mathcal{E}}_{n,\tau}^{-v}$ using data out of fold $v$.
\EndIFor
\For{$v \in [V]$ and $\tau \in \mathcal{T}_n$} \Comment{(Obtain a one-step corrected estimator for fold $v$)}
    \State Let $\hat{P}_{\tau}^{n,v}$ be a distribution with the following components: (i) marginal distribution of $A$ being $P_{A}^{n,v}$, (ii) conditional distribution of $X \mid A=0$ being $P_{X \mid 0}^{n,v}$, and (iii) distribution of $Z_\tau \mid X,A=1$ defined by $\hat{\mathcal{E}}_{n,\tau}^{-v}$
    \State Let $|I_v|$ be the cardinality of the index set $I_v$. Set
    \begin{equation}\label{psi-n-v alg}
        \hat{\psi}_{n,\tau}^v := \frac{\sum_{i \in I_v} (1-A_i) \hat{\mathcal{E}}_{n,\tau}^{-v}(X_i)}{\sum_{i \in I_v} (1-A_i)} + \frac{1}{|I_v|} \sum_{i \in I_v}  \frac{A_i}{\hat{\gamma}_n^v} \mathscr{W}(\hat{g}_n^{-v},\hat{\gamma}_n^v) [Z_\tau(X_i,Y_i) - \hat{\mathcal{E}}_{n,\tau}^{-v}(X_i)].
    \end{equation}
\EndFor
\IFor{$\tau \in \mathcal{T}_n$}
Obtain the cross-fit one-step corrected estimator for threshold $\tau$:
        \begin{equation}\label{psintau}
        \hat{\psi}_{n,\tau} := \frac{1}{n} \sum_{v=1}^V |I_v| \hat{\psi}_{n,\tau}^v.
        \end{equation}
\EndIFor
\end{algorithmic}
\end{algorithm}

\begin{figure}
    \centering
    \includegraphics[scale=0.33]{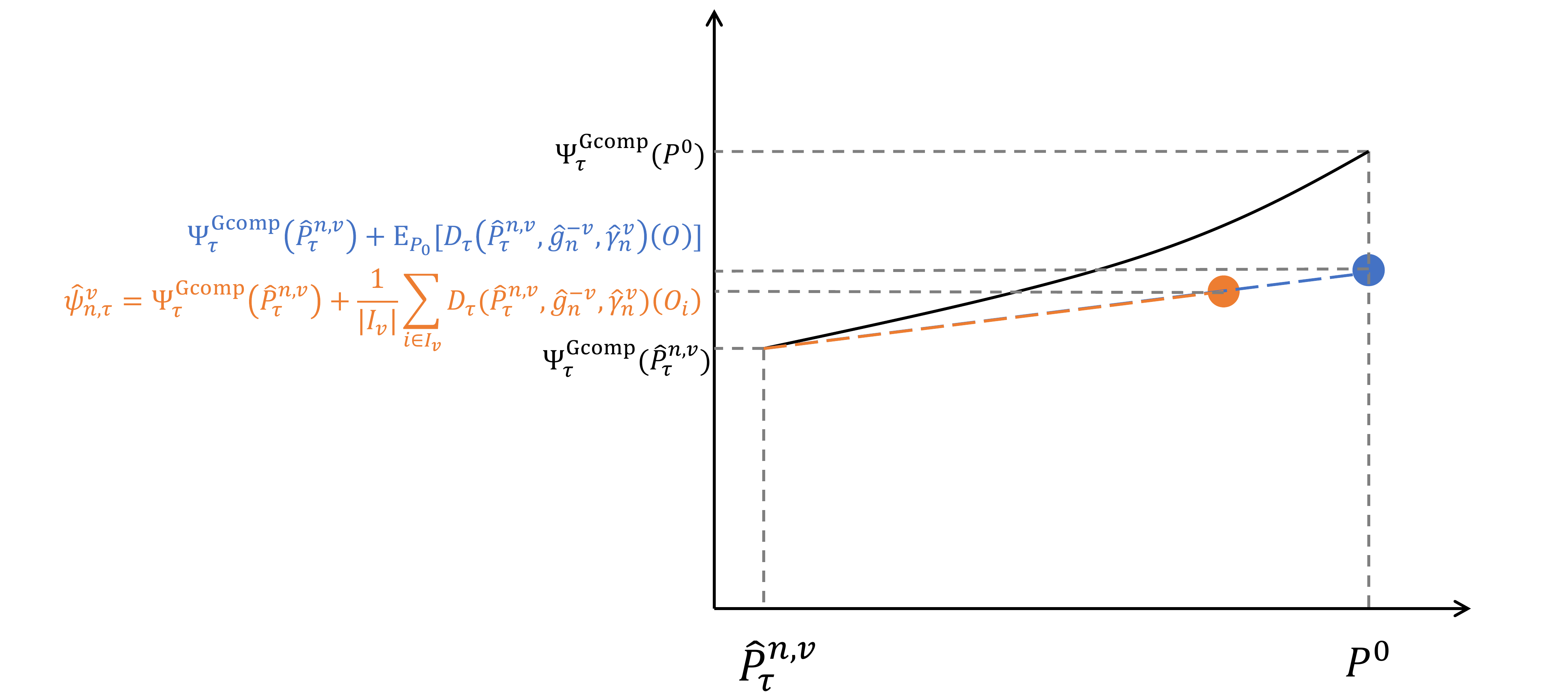}
    \caption{Illustration of the idea behind the one-step correction in \eqref{psi-n-v}. In this figure, $\Psi^\Gcomp_\tau(\hat{P}_{\tau}^{n,v}) + \expect_{P^0}[D_\tau(\hat{P}_{\tau}^{n,v},\hat{g}_n^{-v},\hat{\gamma}_n^v)(O)]$ is the ideal/oracle first-order approximation to the estimand $\Psi^\Gcomp_\tau(P^0)$ at the na\"ive plug-in estimator $\Psi^\Gcomp_\tau(\hat{P}_{\tau}^{n,v})$; $\hat{\psi}_{n,\tau}^v$ is the one-step corrected estimator.}
    \label{fig: one step}
\end{figure}

We require two additional conditions on $\hat{g}_n^{-v}$ and $\hat{\mathcal{E}}_{n,\tau}^{-v}$ for $\hat{\psi}_{n,\tau}$ to be asymptotically efficient.

\begin{condition}[Bounded propensity score estimator] \label{cond: consistent bounded}
For some non-negative sequence $(q_n)_{n\geq 0}$ tending to zero as $n\to\infty$,
with probability $1-q_n$, $\inf_{v \in [V], x \in \mathcal{X}} \hat{g}_n^{-v}(x) > \delta$, where $\delta$ is the constant introduced in Condition~\protect\ref{cond: bounded weight}.
\end{condition}

Typically, with appropriate regularization to avoid overfitting, the estimator $\hat{g}_n^{-v}$ of the propensity score is bounded away from zero 
except in extremely ill-posed datasets. 
These occur with extremely small probability (for example, if the covariate $X$ is discrete and for some $x$, all observations with $X=x$ are from the target population).  
Moreover, the user can always truncate the estimator to be bounded away from zero, in which case $q_n=0$.
Thus we often expect $q_n$ in Condition~\protect\ref{cond: consistent bounded} to decrease to zero at a much faster rate than the convergence rates of the nuisance function estimators in Condition~\protect\ref{cond: sufficient nuisance rate}. 
This does not have an effect on the asymptotic efficiency of the cross-fit one-step estimator $\hat{\psi}_{n,\tau}$, but it will affect the Wald confidence interval coverage in the next subsection.
Recall that $\| \cdot \|_{P,p}$ stands for the $L^p(P)$-norm of a function. 
Our condition for nuisance estimators is as follows.
\begin{condition}[Sufficient rates for nuisance estimators] \label{cond: sufficient nuisance rate}
The following conditions hold:
\begin{align*}
    & \expect_{P^0} \sup_{v \in [V], \tau \in \mathcal{T}_n} \| \hat{\mathcal{E}}_{n,\tau}^{-v}-\mathcal{E}_{0,\tau} \|_{P^0_{X \mid 0},2} = \smallo(1), \quad
    \expect_{P^0} \sup_{v \in [V]} \| (1-\hat{g}_n^{-v})/\hat{g}_n^{-v} - (1-g_0)/g_0 \|_{P^0_{X \mid 0},2} = \smallo(1), \\
    & \expect_{P^0} \sup_{v \in [V], \tau \in \mathcal{T}_n} \int \left| \left\{ \frac{1-\hat{g}_n^{-v}(x)}{\hat{g}_n^{-v}(x)} - \frac{1-g_0(x)}{g_0(x)} \right\} \left\{ \hat{\mathcal{E}}_{n,\tau}^{-v}(x) - \hat{\mathcal{E}}_{0,\tau}(x) \right\} \right| P^0_{X \mid 1}(\intd x) = \smallo(n^{-1/2}).
\end{align*}
\end{condition}

By the Cauchy-Schwarz inequality, a sufficient condition for the last equation in Condition~\protect\ref{cond: sufficient nuisance rate} is the following:
\begin{align*}
    & \expect_{P^0} \sup_{v \in [V], \tau \in \mathcal{T}_n} \left\| \frac{1-\hat{g}_n^{-v}}{\hat{g}_n^{-v}} - \frac{1-g_0}{g_0} \right\|_{P^0_{X \mid 1},2} \left\| \hat{\mathcal{E}}_{n,\tau}^{-v} - \mathcal{E}_{0,\tau} \right\|_{P^0_{X \mid 1},2} = \smallo(n^{-1/2}).
\end{align*}
Therefore, a sufficient condition is that both nuisance estimators converge at a rate faster than $n^{-1/4}$. 
Thus we allow for much slower rates than the parametric root-$n$ rate. This $\smallo(n^{-1/4})$ rate requirement is only a sufficient condition and is by no means necessary. Condition~\protect\ref{cond: sufficient nuisance rate} is satisfied even if one nuisance estimator converges very slowly, as long as the other nuisance estimator converges sufficiently fast to compensate for this slow convergence. We require convergence of the estimator $\mathcal{E}_{n,\tau}^{-v}$ uniformly over $\tau \in \mathcal{T}_n$ to establish uniform asymptotic efficiency in Theorem~\ref{thm: CV one-step efficiency} below. Though this assumption on convergence is stronger than assumptions typically needed to obtain efficient estimators, it may not be stringent. We illustrate this with an example in Section~\ref{sec: E uniform convergence example} 
in the Supplemental Material.

\begin{remark} \label{rmk: mixed bias}
The aforementioned phenomenon that one convergence rate can compensate for the other is similar to the mixed bias property, which is frequently observed for semiparametrically or nonparametrically efficient estimators \protect\citep{Rotnitzky2021}.
The mixed bias property often leads to double robustness.
An estimator is doubly robust if it is still consistent even when one nuisance function, but not the other, is estimated inconsistently (see, e.g., the rejoinder to discussions of \protect\citet{Scharfstein1999}, \protect\citet{Robins2000}, and \protect\citet{Bang2005}).
This double robustness property also holds for our estimator $\hat{\psi}_{n,\tau}$ of coverage error $\Psi_\tau(P^0)$, similarly to the method in \protect\citet{Yang2022}. 
In other words, $\hat{\psi}_{n,\tau}$ is consistent for $\Psi_\tau(P^0)$ even if either $\mathcal{E}_{0,\tau}$ or $g_0$ is estimated inconsistently, in which case Condition~\protect\ref{cond: sufficient nuisance rate} fails. We do, however, generally require Condition~\protect\ref{cond: sufficient nuisance rate} to hold for our proposed PredSet-1Step method except for special cases.
The reason is that PredSet-1Step further relies on asymptotically valid CUBs, which rely on the asymptotic normality of $\hat{\psi}_{n,\tau}$. If a nuisance function is estimated inconsistently and thus Condition~\protect\ref{cond: sufficient nuisance rate} fails, even though $\hat{\psi}_{n,\tau}$ is still consistent for $\Psi_\tau(P^0)$, $\hat{\psi}_{n,\tau}$ is no longer asymptotically normal in general. In this case, it is challenging, if possible at all, to construct asymptotically valid CUB.
We discuss special cases where PredSet-1Step is doubly robust under Condition~\protect\ref{cond: known nuisance} or \protect\ref{cond: Hadamard differentiable nuisance} 
in Section~\protect\ref{sec: one step DR}
in the Supplemental Material. In particular, when one nuisance function is known, our proposed procedure remains valid with the known nuisance function plugged in.
\end{remark}

This leads to our second result.
\begin{theorem}[Asymptotic efficiency of cross-fit one-step corrected estimator] \label{thm: CV one-step efficiency}
Under Conditions~\protect\ref{cond: positivity of P(A)}--\protect\ref{cond: sufficient nuisance rate},
the one-step corrected cross-fit estimator $\hat{\psi}_{n,\tau}$ from \protect\eqref{psintau} 
is an asymptotically nonparametrically efficient estimator of
the coverage error $\Psi_\tau(P^0) = \Prob(Y \notin C_\tau(X) \mid A=0)$ from \protect\eqref{psitau}, in the target, ``covariate shifted'' population where $A=0$.
Moreover, 
with the gradient $D_\tau$ from \protect\eqref{dtau}, 
the propensity score $g_0$ and the probability $\gamma_0$ of $A=1$ from \protect\eqref{g0gamma0}, and
the conditional coverage error rate
$\mathcal{E}_{0,\tau}$ from \protect\eqref{q0tau},
$$\sup_{\tau \in \mathcal{T}_n} \left| \hat{\psi}_{n,\tau} - \Psi_\tau(P^0) - \frac{1}{n} \sum_{i=1}^n D_\tau(P^0,g_0,\gamma_0)(O_i) \right| = \smallo_p(n^{-1/2}).$$
\end{theorem}

Theorem~\protect\ref{thm: CV one-step efficiency} states the same asymptotic efficiency claim that is implied by the general result Theorem~3.1 and the more concrete result Theorem~5.1 in \citet{Chernozhukov2018debiasedML}. See also Proposition 2 in \cite{kennedy2022semiparametric}.
One difference is that Theorem~\protect\ref{thm: CV one-step efficiency} concerns a uniform efficiency claim over $\tau \in \mathcal{T}_n$, which is implied by pointwise efficiency and the uniform rate condition \ref{cond: sufficient nuisance rate}; another difference arises in the proof due to the different estimation strategies for the nuisance parameter $\gamma_0$.
The proof of Theorem~\protect\ref{thm: CV one-step efficiency} can be found in Section~\protect\ref{sec: proof efficiency}
in the Supplemental Material.
As explained in Section \protect\ref{sec: pathwise differentiability}, this result implies that the one-step corrected estimator enjoys a desirable optimality property: it has the smallest possible asymptotic variance, among all regular estimators, under the nonparametric model $\modelspace$. 
 
\begin{remark} \label{rmk: one-step vs TMLE}
Although $\hat{\psi}_{n,\tau}$ is consistent for $\Psi_\tau(P^0) \in [0,1]$, this estimator itself may fall outside of the interval $[0,1]$. This possibility may harm the interpretation of $\hat{\psi}_{n,\tau}$ as an estimator of a probability. We may project $\hat{\psi}_{n,\tau}$ onto $[0,1]$, or instead use the targeted minimum-loss based estimator (TMLE) \protect\citep{VanderLaan2006,VanderLaan2018}. We present the method based on TMLE, PredSet-TMLE, in Section~\protect\ref{sec: TMLE}
in the Supplemental Material.
\end{remark}

\subsection{Wald CUB and selection of threshold} \label{sec: CV one-step CI and tau hat}

To construct Wald CUBs based on $\hat{\psi}_{n,\tau}$ using Theorem~\protect\ref{thm: CV one-step efficiency}, we need to estimate the asymptotic variance $\sigma^2_{0,\tau} := \expect_{P^0}[D_\tau(P^0,g_0,\gamma_0)(O)^2]$. We propose to use a plug-in estimator based on sample splitting. Let
\begin{equation}\label{sigmantau}
(\hat{\sigma}_{n,\tau}^v)^2 := \frac{1}{|I_v|} \sum_{i \in I_v} D_\tau(\hat{P}_{\tau}^{n,v},\hat{g}_n^{-v},\hat{\gamma}_n^v)(O_i)^2
\qquad\textnormal{ and }\qquad
\hat{\sigma}_{n,\tau} := \left[\frac{1}{n} \sum_{v=1}^V |I_v| (\hat{\sigma}_{n,\tau}^v)^2\right]^{1/2}.
\end{equation}
We propose to use $\hat{\sigma}_{n,\tau}/\sqrt{n}$ as the standard error when constructing a $(1-\alpha_\conf)$-Wald CUB of $\Psi_\tau(P^0)$. That is, we propose to use $\hat{\psi}_{n,\tau} + z_{\alpha_\conf} \hat{\sigma}_{n,\tau}/\sqrt{n}$ as an approximate $(1-\alpha_\conf)$-CUB, where 
we use $z_\alpha$ to denote the $(1-\alpha)$-quantile of the standard normal distribution for any $\alpha \in (0,1)$.

Our theoretical guarantees on the APAC property of PredSet-1Step rely on the following general result on the confidence interval coverage of Wald CIs based on asymptotically linear estimators. 
Recall that an estimator $\hat{\phi}_n$ of $\phi_0$ is asymptotically linear with influence function $\IF$ if the expansion 
$\phi_n = \phi_0 + \frac{1}{n} \sum_{i=1}^n \IF(O_i) + \smallo_p(n^{-1/2})$
holds for $\phi_n$ \protect\citep[see, e.g., Chapter~25 in][]{vandervaart1996}.
In this definition, it is implicitly assumed that $\expect_{P^0} [\IF(O)]=0$ and $\expect_{P^0}[\IF(O)^2] < \infty$.

\begin{theorem}[Coverage of Wald CIs] \label{thm: general CI coverage}
Suppose that $\hat{\phi}_n$ is an asymptotically linear estimator of $\phi_0$ with influence function $\IF$ such that $\sigma^2_0 := \expect_{P^0}[\IF(O)^2] > 0$. Let $\hat{\sigma}^2_n$ be a consistent estimator of the asymptotic variance $\sigma^2_0$. Consider the corresponding Wald $(1-\alpha)$-CUB $\hat{\phi}_n + z_\alpha \hat{\sigma}_n/\sqrt{n}$ for $\phi_0$. Assume that $\expect_{P^0} |\IF(O)|^3 =\rho_0 < \infty$. Then, for any fixed scalar $\eta>0$, there exists a universal constant $\const$ such that
\begin{align*}
    &\left|\Prob_{P^0}(\phi_0 < \hat{\phi}_n + z_\alpha \hat{\sigma}_n/\sqrt{n}) - (1-\alpha)\right| \leq \const \frac{n^{1/4}}{{\sigma_0^{1/2}}} \left\{ \expect_{P^0} \left| \hat{\phi}_n-\phi_0- \frac{1}{n} \sum_{i=1}^n \IF(O_i) \right| \right\}^{1/2} \\
    &\qquad+ \left\{ \const \frac{\expect_{P^0}[\ind(|\hat{\sigma}_n-\sigma_0| \leq \eta) |\hat{\sigma}_n-\sigma_0|]}{\sigma_0} + \Prob_{P^0}(|\hat{\sigma}_n-\sigma_0| > \eta) \right\} + \const \frac{\rho_0}{\sigma_0^3} n^{-1/2}.
\end{align*}
\end{theorem}

The three terms on the right-hand side arise from three sources: (i) the deviation of $\hat{\phi}_n-\phi_0$ from the sample mean of the influence function, (ii) the estimation error of the asymptotic variance, and (iii) the deviation of a root-$n$-scaled centered sample mean from its limiting normal distribution. In nonparametric models, the above bound typically converges to zero slower than the root-$n$-rate that is standard for parametric models, which is a phenomenon observed in some semi/non-parametric problems \protect\citep{Han2019,Zhang2011}.
We review the related literature on this problem in more detail in Section~\ref{sec: CI coverage lit review}
in the Supplemental Material.
The above bound is likely to have room for improvement, but this result suffices to prove the desired APAC property of our procedure. The proof of Theorem~\protect\ref{thm: general CI coverage} can be found in Section~\protect\ref{sec: proof Wald CI coverage}
in the Supplemental Material.

For our problem, Theorem~\protect\ref{thm: general CI coverage} alone is insufficient for results on CUB coverage for all $\tau \in \mathcal{T}_n$. For extremely large or small $\tau$, it is possible that $\Psi_\tau(P^0)=0$ or $1$ and $D_\tau(P^0,g_0,\gamma_0) \equiv 0$. Therefore, we need to consider this special case separately. For any $\epsilon \geq 0$, let 
\begin{equation}\label{teps}
\mathcal{T}^\epsilon:=\{\tau \in \bar{\real}: \sigma^2_{0,\tau} > \epsilon\}
\quad\textnormal{and}\quad
\mathcal{T}^- := \bar{\real} \setminus \mathcal{T}^0 = \{ \tau \in \bar{\real}: \sigma^2_{0,\tau} = 0 \},
\end{equation}
where $\sigma^2_{0,\tau}$ is defined at the beginning of Section \protect\ref{sec: CV one-step CI and tau hat}.
These sets of candidate thresholds depend on the true data-generating mechanism $P^0$ only, and are deterministic.
For all $\tau \in \mathcal{T}^-$, one of the following scenarios occurs: either (i) $\mathcal{E}_{0,\tau} \equiv 0$ and $\Psi_\tau(P^0)=0$, or (ii) $\mathcal{E}_{0,\tau} \equiv 1$ and $\Psi_\tau(P^0)=1$. We make the following assumption on $\hat{\mathcal{E}}_{n,\tau}^{-v}$.

\begin{condition}[Deterministic conditional coverage error estimator for extreme thresholds] \label{cond: constant Q for extreme tau}
For all $\tau \in \mathcal{T}^-$, it holds that $\hat{\mathcal{E}}_{n,\tau}^{-v} = \mathcal{E}_{0,\tau}$ for all $v \in [V]$ and all $n$.
\end{condition}

Condition~\protect\ref{cond: constant Q for extreme tau} is so mild that it is often automatically satisfied: for extremely small $\tau$ that lies in $\mathcal{T}^-$, the random variable $Z_\tau=\ind(s(X,Y)>\tau)$ is a constant equal to one. 
Since one can only observe $Z_\tau(X_i,Y_i)=1$ in any sample, it is natural to estimate $\mathcal{E}_{0,\tau}$ with $\hat{\mathcal{E}}_{n,\tau}^{-v} \equiv 1$, which equals $\mathcal{E}_{0,\tau}$. 
On the other hand, for extremely large $\tau$ that lies in $\mathcal{T}^-$, the random variable $Z_\tau$ is a constant equal to zero and it is natural to estimate $\mathcal{E}_{0,\tau}$ with $\hat{\mathcal{E}}_{n,\tau}^{-v} \equiv 0$, which also equals $\mathcal{E}_{0,\tau}$.

We have the following convergence rate of the coverage to the nominal confidence $1-\alpha_\conf$.

\begin{theorem}[Convergence rate of Wald-CUB coverage based on cross-fit one-step corrected estimator] \label{thm: convergence rate of one-step Wald CI}
Consider the cross-fit one-step corrected estimator $\hat{\psi}_{n,\tau}$ from \protect\eqref{psintau}, 
for the standard error estimator $\hat{\sigma}_{n,\tau}$ from \protect\eqref{sigmantau},
and the coverage error $\Psi_\tau(P^0) = \Prob_{\bar{P}^0}(Y \notin C_\tau(X) \mid A=0)$ from \protect\eqref{psitau}, in the target, ``covariate shifted'' population where $A=0$.
Under Conditions~\protect\ref{cond: positivity of P(A)}--\protect\ref{cond: sufficient nuisance rate},
for any fixed $\epsilon>0$, with $\mathcal{T}^\epsilon$ from \protect\eqref{teps}, it holds that
$$\sup_{\tau \in \mathcal{T}^\epsilon \cap \mathcal{T}_n} \left| \Prob(\Psi_\tau(P^0) < \hat{\psi}_{n,\tau} + z_{\alpha_\conf} \hat{\sigma}_{n,\tau}/\sqrt{n}) - (1-\alpha_\conf) \right| \lesssim \Delta_{n,\epsilon}$$
where, with
$\hat{g}_n^{-v}$ from Line~2 of Algorithm~\protect\ref{alg: CV one step}, $\hat{\gamma}_n^v$ from \protect\eqref{gammanminusv}, 
$g_0, \gamma_0$ from \protect\eqref{g0gamma0},
$\hat{\mathcal{E}}_{n,\tau}^{-v}$ from Line~4 of Algorithm~\protect\ref{alg: CV one step},
$\mathcal{E}_{0,\tau}$ from \protect\eqref{q0tau},
the marginal distribution $P^0_{X \mid 1}$ of $X$ in the source population from Condition \protect\ref{cond: target dominated by source}, and probability $1-q_n$ of having a bounded nuisance estimator from Condition~\protect\ref{cond: consistent bounded},
\begin{align}
    \begin{split}
        \Delta_{n,\epsilon} &:= n^{1/4} \epsilon^{-1/4} \sup_{v \in [V], \tau \in \mathcal{T}_n} \left\{ \expect_{P^0} \left| \int \left( \frac{1-\hat{g}_n^{-v}(x)}{\hat{g}_n^{-v}(x)} - \frac{1-g_0(x)}{g_0(x)} \right) \cdot (\hat{\mathcal{E}}_{n,\tau}^{-v}(x) - \mathcal{E}_{0,\tau}(x)) P^0_{X \mid 1}(\intd x) \right| \right\}^{1/2} + q_n \label{eq: Deltan for Wald CI coverage}
    \end{split}
\end{align}
converges to zero.
In addition, 
with $\mathcal{T}^-$ from \protect\eqref{teps},
under Condition~\protect\ref{cond: constant Q for extreme tau}, it holds that $\Pr(\Psi_\tau(P^0) \leq \hat{\psi}_{n,\tau} + z_{\alpha_\conf} \hat{\sigma}_{n,\tau}/\sqrt{n}) = 1$ for all $\tau \in \mathcal{T}^-$.
\end{theorem}
Theorem~\protect\ref{thm: convergence rate of one-step Wald CI} is a consequence of Theorem~\protect\ref{thm: general CI coverage}, and the proof can be found in Section~\protect\ref{sec: proof Wald CI coverage}
in the Supplemental Material.
The uniform bound only holds for thresholds in $\mathcal{T}^\epsilon$ for some $\epsilon>0$ because, as $\sigma_{0,\tau}^2$ tends to zero, it becomes more difficult to estimate $\sigma_{0,\tau}^2$ with a small \textit{relative} error.
The error term $\Delta_{n,\epsilon}$ in Theorem~\protect\ref{thm: convergence rate of one-step Wald CI} is essentially the square root of the product bias of the two nuisance function estimators $\hat{g}^{-v}_n$ and $\hat{\mathcal{E}}^{-v}_n$, scaled by $n^{1/4}$. This product bias term is the dominating term in the bound in Theorem~\protect\ref{thm: general CI coverage}. This dominance suggests that, when using flexible nonparametric nuisance estimators, the main challenge in improving the coverage of the Wald CI based on our proposed estimator $\hat{\psi}_{n,\tau}$ might be the product bias; improved estimators of the asymptotic variance $\sigma_{0,\tau}^2$ alone might not substantially improve the CI coverage.
We conjecture that this phenomenon might hold for a variety of efficient estimators that are constructed using semiparametric efficiency theory and involve nuisance function estimation.

Based on the Wald-CUB, we select a threshold $\hat{\tau}^\onestep_n$ to ensure that the size of prediction sets is small:
\begin{equation}\label{tos}
\hat{\tau}^\onestep_n := \max \{ \tau \in \mathcal{T}_n: \hat{\psi}_{n,\tau'} + z_{\alpha_\conf} \hat{\sigma}_{n,\tau'}/\sqrt{n} < \alpha_\error \text{ for all } \tau' \in \mathcal{T}_n \text{ such that } \tau' \leq \tau \}.
\end{equation}
This step is illustrated in Figure~\protect\ref{fig: illustrate}
in the Supplemental Material. 
This procedure for choosing a threshold based on CUBs is justified by the following general result on APAC prediction set construction based on pointwise CUBs, which is similar to Theorem~1 in \protect\citet{Bates2021} with adaptations to finite candidate threshold sets, general distributions of the score $s(X,Y)$, and asymptotic CUBs.

\begin{theorem}[Grid search threshold based on CUB] \label{thm: threshold selection based on CUB}
Given a finite set $\mathcal{T}_n$ of candidate thresholds and asymptotic $(1-\alpha_\conf)$-level CUBs  $\lambda_n(\tau)$ of $\Psi_\tau(P^0)$ valid pointwise for each $\tau \in \mathcal{T}_n$, define the selected threshold
\begin{equation}
    \hat{\tau}_n:=\max\{\tau \in \mathcal{T}_n: \lambda_n(\tau') < \alpha_\error \text{ for all } \tau' \in \mathcal{T}_n
    \text{ such that } \tau' \leq \tau\}. \label{eq: general taun}
\end{equation}
Then the prediction set with threshold $\hat{\tau}_n$ satisfies the following:
\begin{align*}
    \Prob_{P^0}(\Psi_{\hat{\tau}_n}(P^0) \leq \alpha_\error) 
    &\geq \inf_{\tau \in \mathcal{T}_n} \Prob_{P^0} \left( \lambda_n(\tau) \geq \Psi_\tau(P^0) \right)\\  
    &\geq 1 - \alpha_\conf - \sup_{\tau \in \mathcal{T}_n} \left| \Prob_{P^0} \left( \lambda_n(\tau)
    \geq \Psi_\tau(P^0) \right) - (1-\alpha_\conf) \right|.
\end{align*}
Consequently, the prediction set with threshold $\hat{\tau}_n$ is APAC if the asymptotic validity of all $\lambda_n(\tau)$ ($\tau \in \mathcal{T}_n$) is uniform; that is,
$$\inf_{\tau \in \mathcal{T}_n} \Prob_{P^0} \left( \lambda_n(\tau) \geq \Psi_\tau(P^0) \right) \geq 1-\alpha_\conf-\smallo(1),$$
which is implied by a uniform convergence of CUB coverage to the nominal level $1-\alpha_\conf$:
$$\sup_{\tau \in \mathcal{T}_n} \left| \Prob_{P^0} \left( \lambda_n(\tau) \geq \Psi_\tau(P^0) \right) - (1-\alpha_\conf) \right|=\smallo(1).$$
If the coverage of the CUB $\lambda_n(\tau)$ is at least $1-\alpha_\conf$ for all $\tau \in \mathcal{T}_n$, then the prediction set with threshold $\hat{\tau}_n$ is PAC.
\end{theorem}
In Theorem~\ref{thm: threshold selection based on CUB}, pointwise valid CUBs, rather than uniform CUBs or confidence bands, are used. 
More general results on using pointwise valid tests to control risk can be found in \protect\citet{Angelopoulos2021}.

We require an additional condition to derive the APAC guarantee of PredSet-1Step from CUB coverage results in Theorem~\protect\ref{thm: convergence rate of one-step Wald CI} and APAC results in Theorem~\protect\ref{thm: threshold selection based on CUB}.

\begin{condition}[Asymptotic variance equal to, or bounded away from, zero]\label{cond: positive variance}
Define $\tau^\dagger_n:=\min\{\tau \in \mathcal{T}_n: \Psi_\tau(P^0) > \alpha_\error\}$, where we define $\min \emptyset := \infty$. For some fixed $\epsilon>0$, it holds that $\tau^\dagger_n \in \mathcal{T}^- \cup \mathcal{T}^\epsilon$.
\end{condition}

We have dropped the dependence of $\tau^\dagger_n$ on $P^0$ from the notation for conciseness. Condition~\protect\ref{cond: positive variance} is again often automatically satisfied as long as the set of candidate thresholds $\mathcal{T}_n$ is sufficiently dense. 
Indeed, we argue that this condition holds if the candidate set $\mathcal{T}_n$ increases with the sample size $n$. Since $\inf\{\Psi_{\tau^\dagger_n}(P^0): n=1,2,\ldots\} \geq \alpha_\error>0$ by definition, either $\inf\{\Psi_{\tau^\dagger_n}(P^0): n=1,2,\ldots\}=1$ or $\alpha_\error \leq \inf\{\Psi_{\tau^\dagger_n}(P^0): n=1,2,\ldots\} < 1$. In the first case, $\Psi_{\tau^\dagger_n}(P^0)$ is trivially equal to unity, and therefore $\tau^\dagger_n \in \mathcal{T}^-$, so Condition~\protect\ref{cond: positive variance} holds. In the second case, since $\mathcal{T}_n$ is increasing with $n$, $\tau^\dagger_n$ is decreasing with $n$. Thus, for some $\delta>0$ and $N$, $\alpha_\error < \Psi_{\tau^\dagger_n}(P^0) < 1-\delta$ for all $n > N$ and thus $\tau^\dagger_n \in \mathcal{T}^{\epsilon^*}$ for some $\epsilon^*>0$. For each $n \leq N$, $\tau^\dagger_n \in \mathcal{T}^{\epsilon_n} \cup \mathcal{T}^-$ for some $\epsilon_n>0$. Condition~\protect\ref{cond: positive variance} hence holds with $\epsilon=\max\{\epsilon^*,\epsilon_1,\ldots,\epsilon_N\}$.
Condition~\protect\ref{cond: positive variance} may only fail if $\Psi_{\tau^\dagger_n}(P^0)$ can be arbitrarily close to---but not equal to---one. 
Even if $\mathcal{T}_n$ is not increasing, in all scenarios we can think of,
Condition~\protect\ref{cond: positive variance} only fails
for extremely contrived sets $\mathcal{T}_n$.

We have the following corollary of Theorem~\protect\ref{thm: convergence rate of one-step Wald CI}, our final result showing the APAC property of PredSet-1Step.

\begin{corollary}[Main result] \label{corollary: CV one-step APAC}
If Conditions~\protect\ref{cond: positivity of P(A)}--\protect\ref{cond: positive variance} hold,
then we have
\begin{equation} \label{eq: one-step guarantee}
    \Prob_{P^0}(\Psi_{\hat{\tau}^\onestep_n}(P^0) \leq \alpha_\error) \geq 1-\alpha_\conf-\const \Delta_{n,{\epsilon}}
\end{equation}
where $\Delta_{n,\epsilon}$ is defined in \protect\eqref{eq: Deltan for Wald CI coverage}. In other words, the prediction set with threshold $\hat{\tau}^\onestep_n$ is APAC.
\end{corollary}

\begin{remark} \label{rmk: bootstrap}
It might be preferable to use another CUB---rather than the Wald CUB we propose. For example, it is well known that carefully constructed bootstrap procedures can lead to better coverage for certain problems \protect\citep{Hall2013}. 
Another possibility is to efficiently estimate the asymptotic variance $\sigma_{0,\tau}^2$. 
However, this does not appear to improve the overall convergence rate, because the estimation error in the asymptotic variance is not the only term that dominates our bound on the convergence rate. 
The other dominating term is the deviation of the estimator from the sample mean of the influence function. 

The empirical performance of the above methods has sometimes been observed to be comparable to the ones without efficient estimation of the asymptotic variance or bootstrap \protect\citep[see, e.g., Chapter~28 in][]{VanderLaan2018}. To our knowledge, theory on the convergence rate of confidence interval coverage for general asymptotically linear estimators
has not been developed in the literature. The bound we obtained in Theorem~\protect\ref{thm: general CI coverage} requires the development of novel tools to propagate the difference between the estimator and the sample mean of the influence function to the difference between the true and the nominal coverage.
\end{remark}

\begin{remark} \label{rmk: efficient estimator based on weight}
PredSet-1Step relies on an efficient estimator based on the G-computation formula $\Psi^\Gcomp_\tau$. An alternative approach is to use estimators based on the weighted formula $\Psi^\weighted_\tau$. In this approach, a one-step correction is also crucial to achieving the same asymptotic efficiency.
Furthermore, for each fold $v$, we have used $\hat{\gamma}_n^v$ based on data in fold $v$ to estimate $\gamma_0$. Using the empirical estimator $\sum_{i \notin I_v} A_i/(n-|I_v|)$ based on data out of fold $v$---an approach that coincides with double/debiased machine learning \protect\citep{Chernozhukov2018debiasedML}---also leads to efficient estimators and APAC prediction sets under the same conditions. The proof is similar, with minor modifications. 
We have chosen to estimate $\gamma_0$ in the same fold because it leads to a remainder that aligns with the conventional definition of the mixed bias property \protect\citep{Rotnitzky2021}.
\end{remark}

\section{Simulations} \label{sec: simulation}

We conduct three simulation studies to investigate the performance of our methods. In the first simulation, we consider a moderate-to-high dimensional sparse setting; in the second simulation, we consider a relatively low dimensional setting; in the third simulation, we consider a relatively low dimensional setting without covariate shift. In all settings, we consider $\alpha_\error=\alpha_\conf=0.05$ and the following methods: 
(i) PredSet-1Step; 
(ii) PredSet-TMLE, described in Remark~\protect\ref{rmk: one-step vs TMLE} and Section~\protect\ref{sec: TMLE} 
in the Supplemental Material; 
(iii) PredSet-RS, a method based on rejection sampling,  described in Section~\protect\ref{sec: rejection sampling method}
in the Supplemental Material; 
(iv) plug-in, a na\"ive variant of PredSet-1Step based on a na\"ive cross-fit plug-in estimator of the true coverage error $\Psi_\tau(P^0)$; the same as PredSet-1Step except that the one-step correction in \protect\eqref{psi-n-v} is not included; 
(v) plug-in2, a method similar to PredSet-1Step based on a corss-fit estimator with the estimated likelihood ratio $w_0$ plugged into $\Psi^\weighted_\tau$; 
(vi) weighted CP, weighted Conformal Prediction \protect\citep{Tibshirani2019} with an estimated likelihood ratio and a target marginal coverage error at most $\alpha_\error$;
and
(vii) inductive CP, inductive Conformal Prediction \protect\citep{papadopoulos2002inductive}, tuned as in \protect\citet{Vovk2013,park2021pac} to ensure training-set conditional validity (i.e., the PAC property), ignoring covariate shift.
To our best knowledge, training-set conditional validity results are unknown for weighted Conformal Prediction. We still include this method for comparison and do not expect it to attain (approximate) training-set conditional validity.
Whenever no threshold can be selected, that is, the CUB corresponding to $\tau=0$ is above $\alpha_\error$, we set the selected threshold to zero. 
We consider a setting without covariate shift in the third simulation, because in this case, inductive CP has a finite sample PAC guarantee while our proposed methods do not. 
In this case, we focus on comparing our proposed methods PredSet-1Step and PredSet-TMLE with inductive CP.

For all methods incorporating covariate shift, we split the data into two folds of equal sizes ($V=2$).
When estimating the nuisance functions $\mathcal{E}_{0,\tau}$ and $g_0$, we use Super Learner \protect\citep{VanderLaan2007} with the library consisting of logistic regression, generalized additive models \protect\citep{Hastie1990}, logistic LASSO regression \protect\citep{hastie1995penalized,tibshirani1996regression} and gradient boosting \protect\citep{Mason1999,Mason2000,Friedman2001,Friedman2002} with various combinations of tuning parameters (maximum number of boosting iterations being 100, 200, 400, 800 or 1000; minimum sum of instance weights needed in a child being 1, 5 or 10). 
Super Learner is an ensemble learner that outputs a weighted average of the algorithms in the library to minimize the cross-validated prediction error. 
In all above methods except inductive CP, the candidate threshold set $\mathcal{T}_n$ is a fixed grid on the interval $[0,0.3]$ with distance between adjacent grid points being 0.05. PredSet-RS requires additional tuning parameters, and we present them in the Supplemental Material.

We consider sample sizes $n$=500, 1000, 2000 and 4000. For each sample size, we run all methods on 200 randomly generated data sets. 
We approximately calculate the true optimal threshold $\tau_0$ by generating $10^6$ samples from the target population and taking the $\alpha_\error$-th quantile of $s(X,Y)$ in the sample.
We next describe the data-generating mechanisms and the results of the three simulations.

\subsection{Moderate-to-high dimensional sparse setting} \label{sec: high dim simulation}

To generate the data, we first generate the population indicator $A \sim \mathrm{Bernoulli}(0.5)$. Given $A=a$, the covariate $X:=(X_1,\ldots,X_{20})^\top$ is a 20-dimensional random vector generated from exponential distributions as follows:
$$X_1 \sim \mathrm{Exp}(2^{1-a}), \qquad X_2 \sim \mathrm{Exp}(2^{1-a}), \qquad X_k \sim \mathrm{Exp}(1) \quad (k=3,\ldots,20),$$
where $X_1,\ldots,X_{20}$ are mutually independent. The outcome $Y$ has three labels $\{0,1,2\}$ and is generated according to the distribution implied by the following two equations:
\begin{align*}
    \frac{\Prob(Y=1 \mid X=x)}{\Prob(Y=0 \mid X=x)} = \exp(2 + 2 x_1 - 1.1 x_2), \quad
    \frac{\Prob(Y=2 \mid X=x)}{\Prob(Y=0 \mid X=x)} = \exp(-2.1 - 2 x_1 + 1.2 x_3).
\end{align*}
Instead of the true conditional probability of $Y$ defined above, we set the scoring function $s$ to be the function satisfying the following three equations for all $x$: $s(x,0) + s(x,1) + s(x,2) = 1$,
\begin{align*}
    & \frac{s(x,1)}{s(x,0)} = \exp(0.02 + 2.1 x_1 - 0.91 x_2 + 0.02 x_4), \quad
    \textnormal{and}\quad
    \frac{s(x,2)}{s(x,0)} = \exp(-0.03 - 1.95 x_1 + 1.25 x_3 + 0.1 x_5).
\end{align*}

The empirical proportion that the true miscoverage is below $\alpha_\error$ is presented in Figure~\protect\ref{fig: high dim miscoverage}.
Since weighted CP was developed to achieve marginal coverage rather than training-set conditional coverage, its proportion of having a miscoverage exceeding $\alpha_\error$ is much higher than the desired level $\alpha_\conf$.
In this simulation, the optimal threshold for the source population is greater than the optimal threshold $\tau_0$ for the target population. 
Hence, inductive CP performs considerably worse than all other methods---that incorporate covariate shift---in the sense that its miscoverage exceeds $\alpha_\error$ much more often than the desired level $\alpha_\conf$, especially in large samples ($n$=4000). As the sample size grows, the performance of inductive CP becomes worse.

The two plug-in methods appear not to be APAC because the Monte-Carlo estimated actual confidence level is below 90\% even in large samples ($n$=4000) and the 95\% confidence interval does not cover the desired 95\% level. 
PredSet-RS performs much worse than other methods, including the invalid inductive CP and plug-in methods, when the sample size is not large ($n \leq$ 2000). 
However, PredSet-RS might be APAC as its confidence level appears to approach 95\% as the sample size grows. The other two methods, PredSet-1Step and PredSet-TMLE, appear to be APAC and have reasonable performance for moderate to large sample sizes ($n \geq 2000$). 

As shown in Figure~\protect\ref{fig: high dim tauhat} 
in the Supplemental Material, the distribution of the threshold selected by PredSet-RS has a much wider spread than PredSet-1Step and PredSet-TMLE.
We therefore recommend PredSet-1Step and PredSet-TMLE rather than PredSet-RS, although all these methods appear to produce APAC prediction sets.

\begin{figure}
    \centering
    \includegraphics{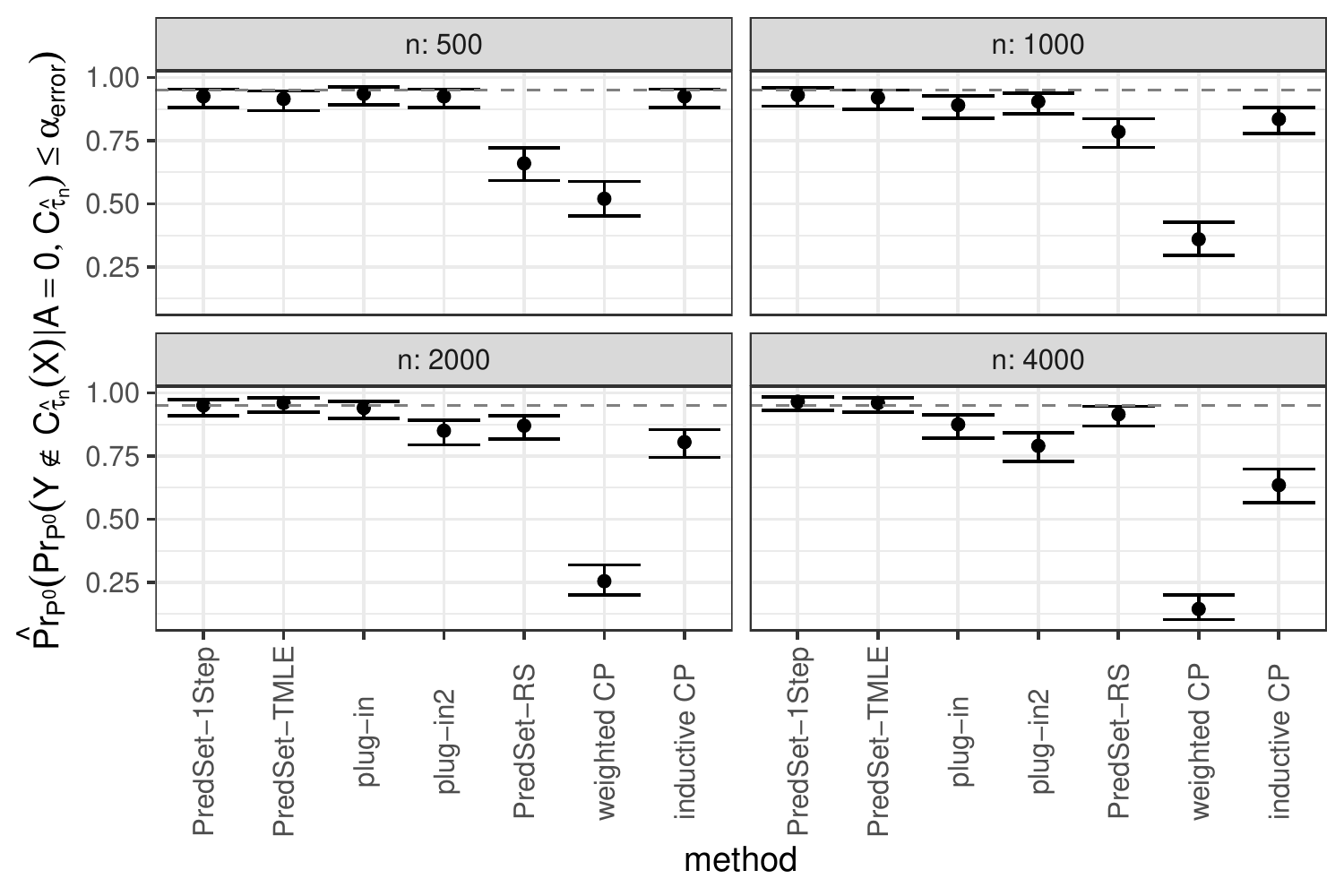}
    \caption{Empirical proportion of simulations where the estimated coverage error $\widehat\Prob_{P^0}(Y \notin C_{\hat{\tau}_n}(X) \mid A=0,C_{\hat{\tau}_n})$ does not exceed $\alpha_\error$, along with a 95\% Wilson score confidence interval, in the moderate-to-high dimensional sparse setting. The gray horizontal dashed line is the desired confidence level $1-\alpha_\conf$.}
    \label{fig: high dim miscoverage}
\end{figure}

\subsection{Low dimensional setting} \label{sec: low dim simulation}

The data-generating mechanism is similar to the previous simulation. We still generate $A$ from a $\mathrm{Bernoulli}(0.5)$ random variable. 
We generate a three-dimensional covariate from a trivariate normal distribution:
$$X \mid A=a \sim \mathrm{N} \left(
\begin{bmatrix} 0 \\ 0 \\ 0 \end{bmatrix},
\begin{bmatrix} 1 & 0.2 & -0.2 \\ 0.2 & 1 & 0.2 \\ -0.2 & 0.2 & 1 \end{bmatrix} \times \left( \frac{1}{2} \right)^{1-a}
\right).$$
The outcome $Y$ also has three labels $\{0,1,2\}$ and is generated according to the distribution implied by the following two equations:
\begin{align*}
    \frac{\Prob(Y=1 \mid X=x)}{\Prob(Y=0 \mid X=x)} &= \exp(1.4 x_1 + 1.5 x_2 -1.5 x_3 + 0.3 (1-x_1)^2 + 0.015 x_2 x_3), \\
    \frac{\Prob(Y=2 \mid X=x)}{\Prob(Y=0 \mid X=x)} &= \exp(-0.1 - 1.3 x_1 - 2.2 x_2 + 0.5 x_3 + 0.5 (1-x_2)^2 + 0.03 x_1 x_3).
\end{align*}
The scoring function is determined by the following three equations, valid for all $x$: $s(x,0) + s(x,1) + s(x,2) = 1$,
\begin{align*}
    & \frac{s(x,1)}{s(x,0)} = \exp(0.02 + 1.2 x_1 + 1.91 x_2 - 1.6 x_3), \quad \text{and} \quad \frac{s(x,2)}{s(x,0)} = \exp(-0.03 - 1.5 x_1 - 2.4 x_2 + 0.3 x_3).
\end{align*}
Unlike in the previous simulation where $g_0$ follows a logistic regression model, here neither $g_0$ nor $\mathcal{E}_{0,\tau}$ follows a parametric model that is correctly specified by an algorithm in the library of Super Learner due to interaction terms in the distribution of $Y \mid X$ and quadratic terms in the logit of $ w_0$. The only exceptions are that $\mathcal{E}_{0,\tau}$ follows a logistic regression model with an infinite slope for an extremely large or small threshold $\tau$. Thus, we do not expect our nuisance function estimators to generally converge at the parametric root-$n$ rate.

The simulation results are presented in Figures~\protect\ref{fig: low dim miscoverage} and \protect\ref{fig: low dim tauhat}
in the Supplemental Material.
The performance of the methods is similar to the moderate-to-high dimensional sparse setting.

\subsection{Low dimensional setting without covariate shift} \label{sec: low dim no cov shift simulation}

The data-generating mechanism is identical to the previous simulation, except that
$$X \mid A=a \sim \mathrm{N} \left(
\begin{bmatrix} 0 \\ 0 \\ 0 \end{bmatrix},
\begin{bmatrix} 1 & 0.2 & -0.2 \\ 0.2 & 1 & 0.2 \\ -0.2 & 0.2 & 1 \end{bmatrix}
\right).$$
In other words, covariate shift is not present.

The simulation results are presented in Figures~\protect\ref{fig: low dim noshift miscoverage}
and \protect\ref{fig: low dim noshift tauhat}
in the Supplemental Material.
Inductive CP appears to perform the best for all sample sizes. 
This is not surprising, because inductive CP has a finite sample PAC guarantee in the no-covariate-shift setting. 
Our proposed methods PredSet-1Step and PredSet-TMLE also appear to be approximately PAC when the sample size is moderate to large ($n \geq 2000$). 
The performance of our proposed methods appears to be comparable to that of inductive CP, even under no covariate shift. The performance of the other two methods---plug-in and PredSet-RS---is similar to the previous simulation.

We therefore conclude from our simulations that, when the sample size is reasonably large, our proposed methods PredSet-1Step and PredSet-TMLE empirically output approximately PAC prediction sets regardless of whether covariate shift is present or not.
Even when no covariate shift is present, in which case inductive CP has a finite sample PAC guarantee, the performance of our methods is empirically comparable with inductive CP.
Our proposed methods can be applied as a default method if the user suspects---but may be unsure---that covariate shift is present, and does not know the likelihood ratio $w_0$ of the shift.

\section{Analysis of HIV risk prediction data in South Africa} \label{sec: data analysis}

We illustrate our methods with a data set concerning HIV risk prediction in a South African cohort study.
Specifically,  we use data from a large population-based prospective cohort study in KwaZulu-Natal, South Africa which was collected and analyzed by \protect\citet{Tanser2013} to evaluate the causal effect of community coverage of antiretroviral HIV treatment on community-level HIV incidence. The study followed a total of 16,667 individuals who were HIV-uninfected at
baseline in order to observe individual HIV seroconversions over the period 2004 to 2011. In the present analysis, we aim to predict HIV seroconversion status over the follow-up period, for a target population of individuals living in a peri-urban community, using urban and rural communities as a source population.

Although the outcome is in fact available in both source and target samples, we deliberately treat the outcome in the target population as missing when constructing prediction sets, and then use the observed outcome in the target population to evaluate empirical coverage of prediction sets. 
There are 12385 and 5136 participants from source and target populations, respectively. 
All participants are treated as independent draws from their corresponding populations. 
The covariates used to predict the outcome are the followings:
(i) binned number partners in the past 12 months,
(ii) current marital status,
(iii) wealth quintile,
(iv) binned age and sex,
(v) binned community antiretroviral therapy (ART) coverage, and
(vi) binned community HIV prelevance.
For covariates that are time-varying, we use the last observed value as the covariate. All covariates are treated as categorical variables in the analysis. 
Missing data for each covariate is treated as a separate category, which is equivalent to the missing-indicator method \protect\citep{Groenwold2012}. Covariate distributions are presented in Figure~\protect\ref{fig: cov shift}
in the Supplemental Material. We also perform Fisher's exact test via a Monte Carlo approximation with 2000 runs to test the equality of covariate distributions in the two populations, and we observe evidence of shift in covariate distribution with a p-value$<0.001$.

For illustration, in this analysis, we create a severe shift on a covariate that we believe to be strongly related to the outcome. 
In the target population, we only include individuals with community ART coverage below 15\% (binned community ART coverage being 1 or 2 in Figure~\protect\ref{fig: cov shift}
in the Supplemental Material). In other words, we set the target population to be the population in the peri-urban communities with ART coverage below 15\%, this sub-population maybe of particular public health policy interest as likely to carry most of the burden of incident HIV cases. We present the analysis results for the full data analysis (target population being peri-urban communities) in Section~\protect\ref{sec: data analysis2}
in the Supplemental Material. In this subset of the data, there are 1418 participants from the target population.

We randomly select 10967 participants from the source population to train the scoring function $s$. We use Super Learner \protect\citep{VanderLaan2007}, with the same setup as in the simulations, to train a classifier of the outcome on this subsample, which is used as the scoring function $s$. We then construct prediction sets using the rest of the sample consisting of 1418 participants from each of the source and the target populations. The target PAC criterion has miscoverage level $\alpha_\error=0.05$ and confidence level $1-\alpha_\conf=0.95$. The methods we apply are a subset of the methods investigated in the simulations: PredSet-1Step, PredSet-TMLE, and inductive CP \protect\citep{papadopoulos2002inductive,Vovk2013,park2021pac}, which ignores covariate shift.
The tuning parameters of these methods, such as the number of folds and the algorithm to estimate nuisance functions, are identical to those in the simulations.

The empirical coverage of the above methods in the sample from the target population is presented in Table~\protect\ref{tab: data analysis result}. The empirical coverage of both PredSet-1Step and PredSet-TMLE is close to the target coverage level $1-\alpha_\error=95\%$, with the 95\% confidence interval containing $1-\alpha_\error$. 
In constrast, the empirical coverage of inductive CP is lower than the target coverage level. 
Thus, properly accounting for covariate shift as in PredSet-1Step and PredSet-TMLE is crucial for achieving 
the PAC property in the ``covariate shifted'' target population in this subset of the data.

\begin{table}
    \centering
    \caption{Empirical coverage of prediction sets, 95\% Wilson score confidence interval for coverage, and selected thresholds in the synthetic sample from the target population in the South Africa HIV trial data. The target coverage is at least $1-\alpha_\error=95\%$, with probability 95\% over the training data.}
    \label{tab: data analysis result}
    \begin{tabular}{l|r|r|r}
        Method & Empirical coverage & Coverage CI & Selected threshold $\hat{\tau}_n$ \\
        \hline\hline
        PredSet-1Step & 95.98\% & 94.83\%--96.89\% & 0.095 \\
        PredSet-TMLE & 95.42\% & 94.20\%--96.39\% & 0.100 \\
        Inductive Conformal Prediction & 91.89\% & 90.35\%--93.20\% & 0.195
    \end{tabular}
\end{table}

\section{Conclusion}

There has been extensive literature on (i) constructing prediction sets based on fitted machine learning models, and (ii) supervised learning under covariate shift. In this work, we study the intersection of these two problems in the challenging setting where the covariate shift needs to be estimated.
We propose a distribution-free method, PredSet-1Step, to construct asymptotically probably approximately correct (APAC) prediction sets under unknown covariate shift. PredSet-1Step may also be used to construct asymptotically risk controlling prediction sets (ARCPS) with a slight modification.
Our method is flexible, taking as input an arbitrary given scoring function, produced by essentially any statistical or machine learning method. 

We use semiparametric efficiency theory when constructing prediction sets to obtain root-$n$ convergence of the true miscoverage corresponding to the selected prediction sets, even if the estimators of the nuisance functions may converge slower than root-$n$.
Our theoretical analysis of PredSet-1Step relies on a novel result on the convergence of Wald confidence intervals based on general asymptotically linear estimators, which is a technical tool of independent interest.
We illustrate that our method has good coverage in a number of experiments and by analyzing a data set concerning HIV risk prediction in a South African cohort.
In experiments without covariate shift, PredSet-1Step performs similarly to inductive CP, which has finite-sample PAC properties. 
Thus, PredSet-1Step may be used in the common scenario if the user suspects---but may not be certain---that covariate shift is present, and does not know the form of the shift.

One interesting open question is the asymptotic behavior of our selected threshold compared to the true optimal threshold. Our simulation results (Figures~\protect\ref{fig: high dim tauhat}, \protect\ref{fig: low dim tauhat} and \protect\ref{fig: low dim noshift tauhat} in the Supplemental Material) suggest that our selected threshold might converge in probability to the true optimal threshold. Our selected threshold also appears to have a vanishing negative bias that ensures the desired confidence level. Theoretical analysis is in need to confirm these conjectures.

\section*{Acknowledgements}
We thank Arun Kumar Kuchibotla, Jing Lei, Lihua Lei, and Yachong Yang for helpful comments.
This work was supported in part by the NSF DMS 2046874 (CAREER) award, NIH grants R01AI27271, R01CA222147, R01AG065276, R01GM139926, and Analytics at Wharton.

{\small 
\setlength{\bibsep}{0.2pt plus 0.3ex}
\bibliographystyle{chicago}
\bibliography{ref}
}

\newpage

\setcounter{page}{1}
\renewcommand*{\theHsection}{S.\the\value{section}}
\setcounter{section}{0}
\renewcommand{\thesection}{S\arabic{section}}%
\setcounter{table}{0}
\renewcommand{\thetable}{S\arabic{table}}%
\setcounter{figure}{0}
\renewcommand{\thefigure}{S\arabic{figure}}%
\setcounter{algorithm}{0}
\renewcommand{\thealgorithm}{S\arabic{algorithm}}%
\setcounter{equation}{0}
\renewcommand{\theequation}{S\arabic{equation}}%
\setcounter{condition}{0}
\renewcommand{\thecondition}{S\arabic{condition}}%
\setcounter{lemma}{0}
\renewcommand{\thelemma}{S\arabic{lemma}}%
\setcounter{theorem}{0}
\renewcommand{\thetheorem}{S\arabic{theorem}}%
\setcounter{corollary}{0}
\renewcommand{\thecorollary}{S\arabic{corollary}}%

\begin{center}
    \LARGE Supplementary Material for ``Prediction Sets\\Adaptive to Unknown Covariate Shift''
\end{center}

This Supplementary Material is organized as follows. In Section~\ref{sec: rejection sampling method}, we describe another procedure, PredSet-RS, to construct APAC prediction sets, as well as its theoretical properties. We describe PredSet-TMLE and present its theoretical results in Section~\ref{sec: TMLE}. We describe how to modify PredSet-1Step and PredSet-TMLE to construct asymptotically risk-controlling prediction sets in Section~\ref{sec: ARCPS}. We discuss double robustness of PredSet-1Step in Section~\ref{sec: one step DR}. In Section~\ref{sec: E uniform convergence example}, we illustrate that the convergence in supremum required by Condition~\ref{cond: sufficient nuisance rate} may be plausible with an example of nonparametric regression. In Section~\ref{sec: data analysis2}, we present analysis results of the full HIV risk prediction data in South Africa. We present proof of our theoretical results in Section~\ref{sec: proof}. We finally discuss the difference between PAC guarantee (conditional validity) and marginal coverage guarantee (marginal validity), which is typically provided in conformal prediction methods, in Section~\ref{sec: discuss PAC}. We discuss the connection between covariate shfit and causal inference in Section~\ref{section: causal and covariate shift}. We review the literature on confidence interval coverage based on efficient estimators involving nuisance function estimation in Section~\ref{sec: CI coverage lit review}. We put supplemental figures at the end of this Supplemental Material.

We first introduce a few more functions for ease of presentation in the rest of this Supplemental Material.
For each function $w: \mathcal{X} \mapsto [0,\infty)$, let $\Pi_P(w):= \expect_{P}[w(X) \mid A=1] = \int w(x) P_{X \mid 1}(\intd x)$. We note that $\Pi_P(w_P)=\int 1 \intd P_{X \mid 0} =1$. 
For each $\tau \in \bar{\real}$, with $o:=(a,x,y)$, we define the following functions:
\begin{align}
    D^\Gcomp_\tau(P,\mathcal{E},g,\gamma,\pi): o &\mapsto \frac{a}{\gamma_P} \mathscr{W}(g,\gamma)(x) \left\{ Z_\tau(x,y) - \mathcal{E}_{P,\tau}(x) \right\}
    + \frac{1-a}{1-\gamma_P} [ \mathcal{E}_{P,\tau}(x) - \Psi^\Gcomp_\tau(P) ], \label{dgdef}\\
    \tilde{D}(\mathcal{E},g,\gamma,\pi): o &\mapsto \mathcal{E}(x) \left\{ -\frac{a}{\gamma} \frac{\mathscr{W}(g,\gamma)(x)}{\pi} + \frac{1-a}{1-\gamma} \right\} \label{tilded}\\
    D^\weighted_\tau(P,\mathcal{E},g,\gamma,\pi): o &\mapsto \frac{a}{\gamma_P} \left\{ \frac{\mathscr{W}(g,\gamma)(x)}{\Pi_P(\mathscr{W}(g,\gamma))} Z_\tau(x,y) - \Psi^\weighted_\tau(P) \right\} 
    \nonumber\\
    &+ \Psi^\weighted_\tau(P) \frac{a-\gamma_P}{\gamma_P (1-\gamma_P)} + \tilde{D}(\mathcal{E},g,\gamma,\pi).\nonumber
\end{align}
Though seemingly different, $D^\Gcomp_\tau(P,\mathcal{E}_{P,\tau},g_P,\gamma_P,1)$ and $D^\weighted_\tau(P,\mathcal{E}_{P,\tau},g_P,\gamma_P,1)$ are identical to $D_\tau(P,g_P,\gamma_P)$.
In our definitions, the arguments $\mathcal{E}$ and $\pi$ of $D^\Gcomp_\tau$ are unused. We keep these unused arguments in our notation to keep the arguments of $D^\Gcomp_\tau$ and $D^\weighted_\tau$ consistent.
We also refer to $D^\Gcomp_\tau(P,\mathcal{E}_{P,\tau},g_P,\gamma_P,1)$ and $D^\weighted_\tau(P,\mathcal{E}_{P,\tau},g_P,\gamma_P,1)$ as $D_\tau(P,\mathcal{E}_{P,\tau},g_P,\gamma_P,1)$ when we need not distinguish their mathematical expressions, where we have redefined
\begin{align}
        D_\tau(P,\mathcal{E}_{P,\tau},g_P,\gamma_P,1): o &\mapsto \frac{a}{\gamma_P} \mathscr{W}(g_P,\gamma_P)(x) \left\{ Z_\tau(x,y) - \mathcal{E}_{P,\tau}(x) \right\} + \frac{1-a}{1-\gamma_P} [ \mathcal{E}_{P,\tau}(x) - \Psi_\tau(P) ]
\label{dtau2}
\end{align}
with a slight abuse of notations. With these definitions, $D_\tau$ in \eqref{dtau} is essentially $D^\Gcomp_\tau$ in \eqref{dgdef}.

We introduce two versions of gradient functions $D^\Gcomp_\tau$ and $D^\weighted_\tau$ because they are used in estimation of $\Psi_\tau(P^0)$ with the two identifying functionals $\Psi^\Gcomp_\tau$ and $\Psi^\weighted_\tau$ respectively. For the main method we propose in Section~\ref{sec: CV one-step}, namely PredSet-1Step, $D^\Gcomp_\tau$ is used and thus we use a simplified notation in the main text for conciseness; for another method we propose in Section~\ref{sec: rejection sampling method}, namely PredSet-RS, $D^\weighted_\tau$ is used.

\section{PredSet-RS} \label{sec: rejection sampling method}

In this section, we describe another procedure, PredSet-RS, to construct APAC prediction sets based on rejection sampling along with its main theoretical properties. This procedure is an extension of the method in \protect\citetsupp{park2021pac}.
We first split the data into training and testing data with observation index sets $I_\train$ and $I_\test$, respectively. Throughout this paper, we assume that the training data set size is
\begin{equation}\label{xi}
|I_\train| = n \xi + \bigO(n^{-1})
\end{equation} 
for a constant $\xi \in (0,1)$, which holds if the size of both training and testing data sets are approximately $n/2$. Since we treat the scoring rule $s$ as fixed, these data sets are independent of the original training data where $s$ was learned.

We first estimate the likelihood ratio $w_0$ based on $I_\train$. Then we generate a sample of independent $(X,Y)$-s, with a distribution close to the target population, by applying rejection sampling with the estimated likelihood ratio as the weight to the test data from the source population. 
Next, for each candidate $\tau \in \mathcal{T}_n$, we use the corresponding Bernoulli sample of $Z_\tau$ from the data obtained via rejection sampling to estimate $\Psi_\tau(P^0)$. In particular, we apply a one-step correction to the sample proportion to account for the estimation of $w_0$ when estimating $\Psi_\tau(P^0)$ and then calculate a Wald CUB. We finally use the approximate CUB to select thresholds similarly to Section~\ref{sec: efficient estimation method}. We next describe this procedure in more detail.

\subsection{Rejection sampling from an approximation to the target population}
\label{rs_approx}

Let $\hat{g}_n^\train$, $\hat{\mathcal{E}}_{n,\tau}^\train$ and $\hat{\gamma}_n^\train$ be estimators of $g_0$, $\mathcal{E}_{0,\tau}$ and $\gamma_0$, respectively, obtained using the training data $I_\train$. Let $\breve{P}^n$ be an oracle distribution estimator (i.e., a distribution with some components that are not based on empirically observed quantities, but rather based on population quantities) with the following components:
\begin{enumerate}
    \item marginal distribution of $A$ being $P^0_{A}$;
    \item conditional distribution of $X \mid A=1$ being $P^0_{X \mid 1}$;
    \item conditional distribution of $Y \mid X=x,A=1$ being $P^0_{Y \mid x}$;
    \item likelihood ratio (with normalizing constant):
    \begin{equation}\label{pip0}
    \mathscr{W}(\hat{g}_n^\train,\hat{\gamma}_n^\train)/\Pi_{P^0}(\mathscr{W}(\hat{g}_n^\train,\hat{\gamma}_n^\train)).
    \end{equation}
\end{enumerate}
We aim at drawing i.i.d. samples from $\breve{P}^n$, based on which we may construct a CUB for $\Psi_\tau(\breve{P}^n)$ as an approximate CUB for $\Psi_\tau(P^0)$. We make the following assumption so that rejection sampling may be conducted.

\begin{condition}[Known bound on likelihood ratio estimator] \label{cond: known weight bound}
There is a known constant $\hat{B} \in [1,\infty)$ such that, for some non-negative sequence $(\breve{q}_n)_{n\ge 1}$ tending to zero, $\sup_{x \in \mathcal{X}} \mathscr{W}(\hat{g}_n^\train,\hat{\gamma}_n^\train)(x) \leq \hat{B}$ with probability $1-\breve{q}_n$.
\end{condition}

The known constant $\hat{B}$ in Condition~\ref{cond: known weight bound} may differ from the constant $B$ in Condition~\ref{cond: bounded weight}. Since $\sup_{x \in \mathcal{X}} w_0(x) \geq \expect_{P^0}[w_0(X) \mid A=1] = 1$, we have assumed without loss of generality that $\hat{B} \geq 1$.
In practice, to specify $\hat{B}$, it is possible for the user to investigate all $g_n^\train(X_i)$, for $i \in I_\test$, and choose $\hat{B}$ to be $\max_{i \in I_\test} \mathscr{W}(g_n^\train,\gamma_n^\train)(X_i)$ or a number greater by, for example, 30\%. Another possible option is to prespecify a lower bound $\hat{\delta} \in (0,1)$ for $g_n^\train$ and choose $\hat{B}$ to be $\frac{1-\hat{\delta}}{\hat{\delta}} \frac{\gamma_n^\train}{1-\gamma_n^\train}$ or a number greater by, for example, 30\%. These are heuristic \textit{ad hoc} approaches to specify $\hat{B}$ that lack strong theoretical support. As discussed after we introduced Condition~\ref{cond: consistent bounded} in Section~\ref{sec: CV one-step}, for suitably chosen $\hat{B}$, we can often expect $\breve{q}_n$ to vanish at an exponential rate.

We next describe the rejection sampling procedure \protect\citepsupp{vonNeumann1951}. We generate exogenous random variables 
\begin{equation}\label{zetai}
\zeta_i \overset{i.i.d.}{\sim} \mathrm{Unif}(0,1),
\end{equation}
for $i \in I_\test$, and we output the sample 
$$S_n:=\{O_i: A_i=1, \zeta_i \leq \mathscr{W}(\hat{g}_n^\train,\hat{\gamma}_n^\train)(X_i)/\hat{B}, i \in I_\test \}.$$
We use $J_n$ to denote the set of the indices of observations in $S_n$.

Since $\mathscr{W}(\hat{g}_n^\train,\hat{\gamma}_n^\train)$ might not be a valid likelihood ratio in the sense that $\Pi_{P^0}(\mathscr{W}(\hat{g}_n^\train,\hat{\gamma}_n^\train))$ might not equal unity, our rejection sampling procedure is different from the ordinary one where a valid likelihood ratio is known. However, we still have the following result that is similar to the properties of the usual rejection sampling.

\begin{theorem}[Properties of rejection sampling] \label{thm: rejection sampling property}
Conditional on the training data and the event $\sup_{x \in \mathcal{X}} \mathscr{W}(\hat{g}_n^\train,\hat{\gamma}_n^\train)(x) \leq \hat{B}$, where $\hat{g}_n^\train$ and $\hat{\gamma}_n^\train$ are estimators $g_0$ and $\gamma_0$ 
from the beginning of Section \ref{rs_approx}
respectively, $\mathscr{W}$ is defined in \eqref{mw}, and $\hat{B}$ is from \eqref{cond: known weight bound}, 
the accepted covariate-outcome pairs $\{(X_i,Y_i): i \in J_n\}$ are an i.i.d. sample drawn from $(X,Y) \mid A=0$ under the approximation  $\breve{P}^n$  to the target population defined at the beginning of Section \ref{rs_approx}. 
In addition, for all $i \in I_\test$, the acceptance probability equals (with $\zeta_i$ defined in \eqref{zetai}, and the normalizing factor $\Pi_{P^0}(\mathscr{W}(\hat{g}_n^\train,\hat{\gamma}_n^\train))$ from \eqref{pip0})
$$\Prob_{P^0}(A_i=1, \zeta_i \leq \mathscr{W}(\hat{g}_n^\train,\hat{\gamma}_n^\train)(X_i)/\hat{B})=\gamma_0 \Pi_{P^0}(\mathscr{W}(\hat{g}_n^\train,\hat{\gamma}_n^\train))/\hat{B}.$$
\end{theorem}

The proof of Theorem~\ref{thm: rejection sampling property} can be found in Section~\ref{sec: proof rejection sampling property} in the Supplemental Material. As in ordinary rejection sampling, the acceptance probability is inversely proportional to the known bound $\hat{B}$ on the likelihood ratio. Therefore, if $\hat{B}$ is large, for example, due to severe covariate shift, rejection sampling might output a small sample, which may lead to significant inefficiency and inaccuracy in small to moderate samples for PredSet-RS. In contrast, PredSet-1Step might not suffer  as much from a severe covariate shift. This heuristic is supported by our later 
simulation results.

\subsection{One-step correction and standard error}

Since $\{(X_i,Y_i): i \in J_n\}$ are an i.i.d. sample drawn from $(X,Y) \mid A=0$ under $\breve{P}^n$, it is not difficult to see that $\sum_{i \in J_n} Z_\tau(X_i,Y_i)$ is distributed as $\mathrm{Binom}(|J_n|,\Psi_\tau(\breve{P}^n))$ conditional on $\mathscr{W}(\hat{g}_n^\train,\hat{\gamma}_n^\train)$ and $|J_n|$.
It might be tempting to use $\sum_{i \in J_n} Z_\tau(X_i,Y_i)/|J_n|$ as an estimator of $\Psi_\tau(P^0)$, and subsequently to use existing methods to construct binomial proportion CUBs, such as the Clopper-Pearson (CP) CUB \protect\citepsupp{Clopper1934} or the Wilson score CUB \protect\citepsupp{Wilson1927}, to construct a CUB for $\Psi_\tau(\breve{P}^n)$, which serves as an approximate CUB for $\Psi_\tau(P^0)$. 

However, this na\"ive approach does not account for the estimation error of the likelihood ratio $w_0$ used in rejection sampling, typically of a rate slower than $n^{-1/2}$ under a nonparametric model. In addition, the standard error used in ordinary binomial proportion CUBs is also invalid with an estimated likelihood ratio, even if a correction is applied to obtain root-$n$-consistency. These issues may invalidate the APAC criterion.

We next describe a one-step correction to $\sum_{i \in J_n} Z_\tau(X_i,Y_i)/|J_n|$, and its asymptotic properties. 
Recall that $\hat{\mathcal{E}}_{n,\tau}^\train$ is an estimator of $\mathcal{E}_{0,\tau}$ obtained using training data $I_\train$. Let
\begin{equation}
    \hat{\pi}_n := \frac{\sum_{i \in I_\test} A_i \mathscr{W}(\hat{g}_n^\train,\hat{\gamma}_n^\train)(X_i)}{\sum_{i \in I_\test} A_i} \label{eq: definition of pi_n}
\end{equation}
be an estimator of $\Pi_{P^0}(\mathscr{W}(\hat{g}_n^\train,\hat{\gamma}_n^\train))$ and
\begin{align}
    \breve{\psi}_{n,\tau} &:= \frac{\sum_{i \in J_n} Z_\tau(X_i,Y_i)}{|J_n|} + \frac{1}{|I_\test|} \sum_{i \in I_\test} \tilde{D}(\hat{\mathcal{E}}_{n,\tau}^\train,\hat{g}_n^\train,\hat{\gamma}_n^\train,\hat{\pi}_n)(O_i) \label{eq: RS one-step estimator} \\
    &= \frac{\sum_{i \in J_n} Z_\tau(X_i,Y_i)}{|J_n|} + \frac{1}{|I_\test|} \sum_{i \in I_\test} \hat{\mathcal{E}}_{n,\tau}^\train(X_i) \left[ -\frac{A_i}{\hat{\gamma}_n^\train} \frac{\mathscr{W}(\hat{g}_n^\train,\hat{\gamma}_n^\train)(X_i)}{\hat{\pi}_n} + \frac{1-A_i}{1-\hat{\gamma}_n^\train} \right]. \nonumber
\end{align}
We will show that this corrected estimator $\breve{\psi}_{n,\tau}$ based on the sample proportion $\sum_{i \in J_n} Z_\tau(X_i,Y_i)/|J_n|$ is asymptotically normal under certain conditions. We next present an additional condition that is similar to Condition~\ref{cond: sufficient nuisance rate}, and a theorem on the theoretical properties of $\breve{\psi}_{n,\tau}$.

\begin{condition}[Sufficient rates for nuisance estimators] \label{cond: sufficient nuisance rate2}
The following conditions hold.
\begin{align*}
    & \expect_{P^0} \sup_{\tau \in \mathcal{T}_n} \| \hat{\mathcal{E}}_{n,\tau}^\train-\mathcal{E}_{0,\tau} \|_{P^0_{X \mid 0},2} = \smallo(1),\qquad
    \expect_{P^0} \left\| \mathscr{W}(\hat{g}_n^\train,\hat{\gamma}_n^\train) - \mathscr{W}(g_0,\gamma_0) \right\|_{P^0_{X \mid 0},2} = \smallo(1), \\
    & \expect_{P^0} \sup_{\tau \in \mathcal{T}_n} \int \left| \left\{ \frac{\mathscr{W}(\hat{g}_n^\train,\hat{\gamma}_n^\train)(x)}{\Pi_{P^0}(\mathscr{W}(\hat{g}_n^\train,\hat{\gamma}_n^\train))} - \mathscr{W}(g_0,\gamma_0)(x) \right\} \left\{ \hat{\mathcal{E}}_{n,\tau}^\train(x) - \mathcal{E}_{0,\tau}(x) \right\} \right| P^0_{X \mid 1}(\intd x) = \smallo(n^{-1/2}).
\end{align*}
\end{condition}

The difference between Condition~\ref{cond: sufficient nuisance rate} and Condition~\ref{cond: sufficient nuisance rate2} is that Condition~\ref{cond: sufficient nuisance rate2} mainly requires a convergence rate on the normalized likelihood ratio estimator in the product term while Condition~\ref{cond: sufficient nuisance rate} mainly requires a convergence rate on the odds ratio estimator. Since the likelihood ratio estimator and the odds ratio estimator differ by a factor that converges to the truth at root-$n$ rate (see Section~\ref{sec: reparameterize weight}), we do not expect these two conditions to be substantially different in practice.
Strictly speaking, the required convergence of the unnormalized estimator $\mathscr{W}(\hat{g}_n^\train,\hat{\gamma}_n^\train)$ to $\mathscr{W}(g_0,\gamma_0)$ in Condition~\ref{cond: sufficient nuisance rate2} may be relaxed to convergence to a constant multiple of $\mathscr{W}(g_0,\gamma_0)$.
We do not take this route, since it is unclear how this relaxed condition may arise in practice.

\begin{theorem}[Asymptotic normality of one-step corrected estimator $\breve{\psi}_{n,\tau}$] \label{thm: rejection sampling one-step correction}
Under Conditions~\ref{cond: positivity of P(A)}--\ref{cond: bounded weight}, \ref{cond: known weight bound} and \ref{cond: sufficient nuisance rate2}, with the coverage error $\Psi_\tau(P^0)$, 
$\gamma_0,g_0$ from \eqref{g0gamma0},
$\hat{B}$ from \eqref{cond: known weight bound}, 
$\zeta_i$ from \eqref{zetai}, 
$w_0$ from Condition \ref{cond: target dominated by source},
$\mathcal{E}_{0,\tau}$ from \eqref{q0tau},
and
$\tilde D$ from \eqref{tilded},
defining
\begin{align*}
    \Gamma_{n,\tau} &:= \frac{1}{|I_\train|} \sum_{i \in I_\train} \frac{A_i - \gamma_0}{\gamma_0 (1-\gamma_0)} \Psi_\tau(P^0) \\
    &\qquad+ \frac{1}{|I_\test|} \sum_{i \in I_\test} \Bigg\{ \hat{B} \frac{A_i}{\gamma_0} \ind(\zeta_i \leq w_0(X_i)/\hat{B}) [Z_\tau(X_i,Y_i)-\Psi_\tau(P^0)] \\
    &\qquad\qquad+ \frac{A_i [w_0(X_i) - 1]}{\gamma_0} \Psi_\tau(P^0) + \tilde{D}(\mathcal{E}_{0,\tau},g_0,\gamma_0,1)(O_i) \Bigg\},
\end{align*}
for the one-step corrected estimator $\breve{\psi}_{n,\tau}$ based on rejection sampling from \eqref{eq: RS one-step estimator}, it holds that
$$\sup_{\tau \in \mathcal{T}_n} \left| \breve{\psi}_{n,\tau} - \Psi_\tau(P^0) - \Gamma_{n,\tau} \right| = \smallo_p(n^{-1/2}).$$
\end{theorem}

The proof of Theorem~\ref{thm: rejection sampling one-step correction} can be found in Section~\ref{sec: proof one-step correction rejection sampling} in the Supplemental Material.

Since, for each $\tau \in \mathcal{T}_n$, $\Gamma_{n,\tau}$ is the sum of two independent sample means with mean zero, we have that both $\sqrt{n} \Gamma_{n,\tau}$ 
and $\sqrt{n}(\breve{\psi}_{n,\tau}-\Psi_\tau(P^0))$
converge in distribution to $\mathrm{N}(0,\varsigma_0^2)$, where
\begin{align}
    \begin{split}
        \varsigma_{0,\tau}^2 &:= \xi^{-1} \expect_{P^0} \left[ \frac{(A - \gamma_0)^2}{\gamma_0^2 (1-\gamma_0)^2} \Psi_\tau(P^0)^2 \right] \\
        &\quad+ (1-\xi)^{-1} \expect_{P^0} \Bigg[ \Bigg\{ \hat{B} \frac{A}{\gamma_0} \ind(\zeta \leq w_0(X)/\hat{B}) [Z_\tau(X,Y)-\Psi_\tau(P^0)] \\
        &\qquad\qquad+ \frac{A [w_0(X) - 1]}{\gamma_0} \Psi_\tau(P^0) + \tilde{D}(\mathcal{E}_{0,\tau},g_0,\gamma_0,1)(O) \Bigg\}^2 \Bigg].
    \end{split} \label{eq: RS var}
\end{align}
A consistent estimator of the asymptotic variance of $\breve{\psi}_{n,\tau}$ is
\begin{align}
\begin{split}
    \hat{\varsigma}_{n,\tau}^2 &:= \frac{n}{|I_\train|} \frac{1}{|I_\train|} \sum_{i \in I_\train} \frac{(A_i-\hat{\gamma}_n^\train)^2}{(\hat{\gamma}_n^\train)^2 (1-\hat{\gamma}_n^\train)^2} \breve{\psi}_{n,\tau}^2 \\
    &\qquad+ \frac{n}{|I_\test|} \frac{1}{|I_\test|} \sum_{i \in I_\test} \Bigg\{ \hat{B} \frac{A_i}{\hat{\gamma}_n^\train} \ind(\zeta_i \leq \mathscr{W}(\hat{g}_n^\train,\hat{\gamma}_n^\train)(X_i)/\hat{B}) [Z_\tau(X_i,Y_i)-\breve{\psi}_{n,\tau}] \\
    &\qquad\qquad+ \frac{A_i [\mathcal{W}(\hat{g}_n^\train,\hat{\gamma}_n^\train)(X_i) - 1]}{\hat{\gamma}_n^\train} \breve{\psi}_{n,\tau} + \tilde{D}(\hat{\mathcal{E}}_{n,\tau}^\train,\hat{g}_n^\train,\hat{\gamma}_n^\train,\hat{\pi}_n)(O_i) \Bigg\}^2.
\end{split}\label{varsigmantau}
\end{align}
We may then use the Wald CUB $\breve{\psi}_{n,\tau} + z_{\alpha_\conf} \hat{\varsigma}_{n,\tau}/\sqrt{n}$ as an asymptotically valid CUB for $\Psi_\tau(P^0)$. Finally, similarly to PredSet-1Step, we let 
\begin{equation}\label{trs}
\hat{\tau}^\rejectsample_n := \max \{ \tau \in \mathcal{T}_n: \breve{\psi}_{n,\tau'} + z_{\alpha_\conf} \hat{\varsigma}_{n,\tau'}/\sqrt{n} < \alpha_\error \text{ for all } \tau' \in \mathcal{T}_n \text{ such that } \tau' \leq \tau \}.
\end{equation}
be the selected threshold based on rejection sampling. This step is also illustrated in Figure~\ref{fig: illustrate}. We propose to use $C_{\hat{\tau}^\rejectsample_n}$ as the prediction set.

We have the following result on the coverage of the Wald CUB and the APAC property of $C_{\hat{\tau}^\rejectsample_n}$. Recall $\mathcal{T}^\epsilon$ defined in Section~\ref{sec: CV one-step CI and tau hat}.

\begin{theorem} \label{thm: convergence rate of rejection sample CI}
Under Conditions~\ref{cond: positivity of P(A)}--\ref{cond: bounded weight}, \ref{cond: known weight bound} and \ref{cond: sufficient nuisance rate2}, for any fixed $\epsilon>0$, with probability tending to one, for the one-step corrected estimator $\breve{\psi}_{n,\tau}$ of the coverage error $\Psi_\tau(P^0)$ based on rejection sampling from \eqref{eq: RS one-step estimator} and the variance   estimator $\hat{\varsigma}_{n,\tau}^2$ from \eqref{varsigmantau}, 
it holds that (with $\mathcal{T}^\epsilon$ from \eqref{teps}),
$$\sup_{\tau \in \mathcal{T}^\epsilon} \left| \Prob_{P^0}(\Psi_\tau(P^0) \leq \breve{\psi}_{n,\tau} + z_{\alpha_\conf} \hat{\varsigma}_{n,\tau}/\sqrt{n}) - (1-\alpha_\conf) \right| \lesssim \breve{\Delta}_{n,\epsilon},$$
where with $\mathscr{W}$ from \eqref{mw},
$\hat{g}_n^\train, \hat{\gamma}_n^\train, \hat{\mathcal{E}}_{n,\tau}^\train$ from the beginning of Section~\ref{rs_approx},
$g_0, \gamma_0$ from \eqref{g0gamma0},
$\mathcal{E}_{0,\tau}$ from \eqref{q0tau},
the marginal distribution $P^0_{X \mid 1}$ of $X$ in the source population from Condition \ref{cond: target dominated by source}, probability $1-\breve{q}_n$ of $\hat{B}$ bounding the estimated likelihood ratio in Condition~\ref{cond: known weight bound},
\begin{align}
    \begin{split}
        \breve{\Delta}_{n,\epsilon} &:= n^{1/4} \epsilon^{-1/4} \sup_{\tau \in \mathcal{T}_n} \left\{ \expect_{P^0} \left| \int \left( \frac{\mathscr{W}(\hat{g}_n^\train,\hat{\gamma}_n^\train)(x)}{\hat{\pi}_n} - \mathscr{W}(g_0,\gamma_0)(x) \right)\cdot (\hat{\mathcal{E}}_{n,\tau}^\train(x) - \mathcal{E}_{0,\tau}(x)) P^0_{X \mid 1}(\intd x) \right| \right\}^{1/2} \\
        &\quad+ \breve{q}_n
    \end{split} \label{eq: Deltan for RS CI coverage}
\end{align}
converges to zero. In addition, under Condition~\ref{cond: constant Q for extreme tau}, for all $\tau$ such that $\Psi_\tau(P^0)=0$, it holds that $\Prob_{P^0}(\Psi_\tau(P^0) \leq \breve{\psi}_{n,\tau} + z_{\alpha_\conf} \hat{\varsigma}_{n,\tau}/\sqrt{n}) = 1$ with probability tending to one. 

Moreover, under Condition~\ref{cond: positive variance},
it holds that
\begin{equation} \label{eq: rejection sampling guarantee}
    \Prob_{P^0}(\Psi_{\hat{\tau}^\rejectsample_n}(P^0) \leq \alpha_\error) \geq 1-\alpha_\conf-\const \breve{\Delta}_{n,\epsilon}.
\end{equation}
In other words, the prediction set with threshold $\hat{\tau}^\rejectsample_n$ is APAC.
\end{theorem}

The proof of Theorem~\ref{thm: convergence rate of rejection sample CI} is similar to Theorem~\ref{thm: convergence rate of one-step Wald CI} and Corollary~\ref{corollary: CV one-step APAC}. A sketch can be found in Section~\ref{sec: proof one-step correction rejection sampling} in the Supplemental Material.

A natural question is whether PredSet-1Step or PredSet-RS is preferred. Since $\hat{\pi}_n$ is consistent for unity, we expect the error of the confidence levels $\breve{\Delta}_{n,\epsilon}$ and $\Delta_{n,\epsilon}$ to be of comparable order. However, we have the following result on the superior accuracy of PredSet-1Step compared to PredSet-RS. 
\begin{theorem} \label{thm: rejection sampling has larger variance}
For any $\tau \in \mathcal{T}^0$, 
for the asymptotic variance $\sigma_{0,\tau}^2:= \expect_{P^0}[D_\tau(P^0,\mathcal{E}_{0,\tau},g_0,\gamma_0,1)(O)^2]$ of the 
one-step corrected estimator $\hat{\psi}_{n,\tau}$ from \eqref{psintau} and the asymptotic variance $\varsigma_{0,\tau}^2$ from \eqref{eq: RS var} of the rejection sampling-based estimator $\breve{\psi}_{n,\tau}$ from \eqref{eq: RS one-step estimator},
it holds that $\sigma_{0,\tau}^2 < \varsigma_{0,\tau}^2$.
\end{theorem}

This result shows that $\breve{\psi}_{n,\tau}$ is a less accurate estimator of $\Psi_\tau(P^0)$ than $\hat{\psi}_{n,\tau}$, and thus has a wider two-sided Wald CI. Therefore, we expect the distribution of $\hat{\tau}^\rejectsample_n$ to have a wider spread than that of $\hat{\tau}^\onestep_n$. Thus, PredSet-RS provides an overly conservative threshold more often than PredSet-1Step. The proof of Theorem~\ref{thm: rejection sampling has larger variance} can be found in Section~\ref{sec: proof one-step correction rejection sampling} in the Supplemental Material.

\begin{remark} \label{rmk: combine with Bonferroni}
In our procedure, it is possible to reverse the role of training and test data and obtain two approximate CUBs. 
In general, we may split the data into $V$ folds, split $\alpha_\conf$ into $\alpha_\conf=\sum_{v=1}^V \alpha_\conf^v$, treat the data in and out of each fold as testing and training data respectively, and obtain $V$ approximate $(1-\alpha_\conf^v)$-level CUBs. We may then set the combined CUB to be the minimum of these CUBs, and the argument for Bonferroni correction \protect\citepsupp{Bonferroni1936,Dunn1961,Bland1995} implies that the confidence level of the combined CUB is at least $1-\alpha_\conf- \const \breve{\Delta}_{\epsilon,n}$. However, it is well known that Bonferroni correction is conservative \protect\citepsupp{Bland1995,Bender1999,Moran2003}, especially if $V$ is large, and hence we do not advocate for this approach. Other simple methods---for example, based on treating these $V$ CUBs as independent---might not apply, because the CUBs are constructed from the entire data and are thus dependent.
\end{remark}

\subsection{Tuning parameters of PredSet-RS in simulations}

In all simulations in Section~\ref{sec: simulation}, we split the data such that $|I_\train|=|I_\test|=n/2$. Nuisance functions $g_0$ and $\mathcal{E}_{0,\tau}$ are estimated with the identical method as in PredSet-1Step and PredSet-TMLE.

In the moderate-to-high dimensional setting, the maximum of the true likelihood ratio equals 4, and we set $\hat{B}$ in Condition~\ref{cond: known weight bound} to be 8 for PredSet-RS. In the low dimensional settings (both with and without covariate shift), the true likelihood ratio is bounded by $2^{3/2} \approx 2.8$, and we set $\hat{B}$ to be a larger number 5.5. We chose these bounds to account for the estimation error in the likelihood ratio and the user's \textit{a priori} ignorance of a tight bound on the estimated likelihood ratio.

\section{Alternative to one-step correction: targeted minimum-loss based estimation (TMLE)} \label{sec: TMLE}

In this section, we present an alternative method PredSet-TMLE to PredSet-1Step as mentioned in Remark~\ref{rmk: one-step vs TMLE}. When constructing asymptotically efficient estimators of $\Psi_\tau(P^0)$, we may use cross-validated targeted minimum-loss based estimators (CV-TMLE) \protect\citepsupp{VanderLaan2006,VanderLaan2018}. After splitting the data into $V$ folds as in Section~\ref{sec: efficient estimation method}, a CV-TMLE is constructed as in the following Algorithm~\ref{alg: CV-TMLE}.

\begin{algorithm}
\caption{CV-TMLE of coverage error $\Psi_\tau(P^0)$ used in PredSet-TMLE} \label{alg: CV-TMLE}
\begin{algorithmic}[1]
\State Obtain initial estimators $\hat{g}_n^{-v}$ and $\hat{\mathcal{E}}_{n,\tau}^{-v}$ of nuisance functions $g_0$ and $\mathcal{E}_{0,\tau}$ using the same procedure as Lines~1--4 in Algorithm~\ref{alg: CV one step}.
\For{$\tau \in \mathcal{T}_n$ and $v \in [V]$} \Comment{(Obtain a TMLE for fold $v$ based on sample splitting)}
    \State If $Z_\tau(X,Y)$ is not all zero or all one, perform logistic regression with outcome $Z_\tau(X,Y)$, offset $\logit\{ \hat{\mathcal{E}}_{n,\tau}^{-v}(X) \}$, clever covariate $\mathscr{W}(\hat{g}_n^{-v},\hat{\gamma}_n^v)(X)$, and no intercept; using data with $A=1$ in fold $v$. Set $\tilde{\mathcal{E}}_{n,\tau}^v$ to be the fitted mean model. Otherwise, when $\hat{\mathcal{E}}_{n,\tau}^{-v}$ is constant zero or one, we set $\tilde{\mathcal{E}}_{n,\tau}^v := \mathcal{E}_{n,\tau}^{-v}$. \label{TMLE step: target reg}
    \State Let $\tilde{P}_{\tau}^{n,v}$ be a distribution with the following components: (i) marginal distribution of $A$ being $P_{A}^{n,v}$, (ii) conditional distribution of $X \mid A=0$ being $P_{X \mid 0}^{n,v}$, (iii) distribution of $Z_\tau \mid X,A=1$ defined by $\tilde{\mathcal{E}}_{n,\tau}^v$, and (iv) likelihood ratio $\mathscr{W}(\hat{g}_n^{-v},\hat{\gamma}_n^{-v})$.
    \State Set
        $$\tilde{\psi}_{n,\tau}^v := \Psi^\Gcomp(\tilde{P}_{\tau}^{n,v}) = \frac{\sum_{i \in I_v} (1-A_i) \tilde{\mathcal{E}}_{n,\tau}^v(X_i) }{\sum_{i \in I_v} (1-A_i)}.$$
\EndFor
\For{$\tau \in \mathcal{T}_n$}
    \State Obtain the CV-TMLE $\tilde{\psi}_{n,\tau} := \frac{1}{n} \sum_{v=1}^V |I_v| \tilde{\psi}_{n,\tau}^v$.
\EndFor
\end{algorithmic}
\end{algorithm}

We assume that Condition~\ref{cond: constant Q for extreme tau} holds for $\hat{\mathcal{E}}_{n,\tau}^{-v}$.
It is then not hard to check that $\tilde{\psi}_{n,\tau}^v$ always lies in the interval $[0,1]$, the range known  to contain $\Psi_\tau(P^0)$, and thus so does $\tilde{\psi}_{n,\tau}$. This is different from the cross-fit one-step corrected estimator $\psi_{n,\tau}$ in \eqref{psi-n-v}, which may fall out of this known range, and hence $\tilde{\psi}_{n,\tau}$ may be preferable to $\hat{\psi}_{n,\tau}$.

However, the logistic regression procedure in Line~\ref{TMLE step: target reg} of Algorithm~\ref{alg: CV-TMLE} may fail or be numerically unstable if, for some $i \in I_v$, $\hat{\mathcal{E}}_{n,\tau}^{-v}(X_i)$ is equal or close to zero or one. This may happen when the threshold $\tau$ is too small or too large. In such cases, we may replace this logistic regression step with an ordinary least squares fit with the same setup, except that the offset is $\hat{\mathcal{E}}_{n,\tau}^{-v}(X)$ instead. The resulting targeted conditional coverage error estimator might fall out of the interval $[0,1]$, but the asymptotic behavior of the corresponding CV-TMLE remains the same. In our implementation, we use ordinary least-squares in Step~\ref{TMLE step: target reg} whenever a numerical issue, namely a warning or error, occurs in our implementation when running a logistic regression. 

We estimate the asymptotic variance similarly to the way it is done in Section~\ref{sec: CV one-step CI and tau hat}. Let
$$(\tilde{\sigma}_{n,\tau}^v)^2 := \frac{1}{|I_v|} \sum_{i \in I_v} D^\Gcomp(\tilde{P}_{\tau}^{n,v},\tilde{\mathcal{E}}_{n,\tau}^v,\hat{g}_n^{-v},\hat{\gamma}_n^v)(O_i)^2$$
and $\tilde{\sigma}_{n,\tau}^2 := \frac{1}{n} \sum_{v=1}^V |I_v| (\tilde{\sigma}_{n,\tau}^v)^2$. We propose to use $\tilde{\sigma}_{n,\tau}/\sqrt{n}$ as the standard error when constructing Wald CUBs. The subsequent procedure to select a threshold is almost identical to that in Section~\ref{sec: CV one-step CI and tau hat}. Similarly to \eqref{tos}, let $\hat{\tau}^\tmle_n := \max \{\tau \in \mathcal{T}_n: \tilde{\psi}_{n,\tau'} + z_{\alpha_\conf} \tilde{\sigma}_{n,\tau'}/\sqrt{n} < \alpha_\error \text{ for all } \tau' \in \mathcal{T}_n \text{ such that } \tau' \leq \tau\}$ be the selected threshold.

We next present theoretical results for the above alternative procedure, which are almost identical to those in Section~\ref{sec: efficient estimation method}. They can be proved similarly. We provide the proof in Section~\ref{sec: proof}.

\begin{theorem}[Asymptotic efficiency of CV-TMLE] \label{thm: CV-TMLE efficiency}
Theorem~\ref{thm: CV one-step efficiency} holds with $\hat{\psi}_{n,\tau}$ replaced by $\tilde{\psi}_{n,\tau}$.
\end{theorem}

\begin{theorem}[Convergence rate of Wald-CUB coverage based on CV-TMLE] \label{thm: convergence rate of TMLE Wald CI}
Theorem~\ref{thm: convergence rate of one-step Wald CI} holds with $(\hat{\psi}_{n,\tau},\hat{\sigma}_{n,\tau})$ replaced by $(\tilde{\psi}_{n,\tau},\tilde{\sigma}_{n,\tau})$.
\end{theorem}

Consequently, Corollary~\ref{corollary: CV one-step APAC} also holds with $\hat{\tau}^\onestep_n$ replaced by $\hat{\tau}^\tmle_n$.

\section{Procedures to construct asymptotically risk-controlling prediction sets (ARCPS), and their theoretical analyses} \label{sec: ARCPS}

In this section, we present algorithms to construct ARCPS, and theoretical results that are similar to those in Sections~\ref{sec: pathwise differentiability} and \ref{sec: efficient estimation method} as well as Section~\ref{sec: TMLE}. Since the results are strikingly similar, the presentation is abbreviated with the main differences highlighted.

We first review the definition of \textit{risk-controlling  prediction set} (RCPS) \protect\citepsupp{Bates2021}, which is more general than PAC, as well as its asymptotic extension.
Given a loss function $\ell$ taking an estimated prediction set and a new observation $(X,Y)$ as inputs, a prediction set $\hat{C}$ is an RCPS with confidence level $1-\alpha_\conf$, if, for a target upper bound  $\alpha_\error$ on the risk of $\hat{C}$,
$$\Prob_{P^0}(\expect_{\bar{P}^0}[\ell(\hat{C},X,Y) \mid A=0,\hat{C}] \leq \alpha_\error) \geq 1-\alpha_\conf.$$
The RSPC criterion reduces to the PAC criterion when $\ell(C,X,Y) = \ind(Y \notin C(X))$, where $\ind(\cdot)$ is the indicator function.
Since PAC prediction sets are a special case of RCPS, Lemma~\ref{lemma: trivial finite sample prediction set} shows that RCPS in the target population under unknown covariate shift is not desirable.
We thus consider asymptotic RCPS (ARCPS).
We say that a prediction set $\hat{C}_n$ is an ARCPS if
$$\Prob_{P^0}(\expect_{\bar{P}^0}[\ell(\hat{C},X,Y) \mid A=0,\hat{C}] \leq \alpha_\error) \geq 1-\alpha_\conf-\smallo(1).$$

We next describe the modified procedure based on PredSet-1Step to construct ARCPS.
To abbreviate the presentation, we use the following notational conventions. For a prediction set with threshold $\tau$, we use $Z_\tau: \mathcal{X} \times \mathcal{Y} \rightarrow \real$ to denote the associated loss function that defines the corresponding risk $\expect_{P^0}[Z_\tau(X,Y) \mid A=0]$. For example, when classifying medical images, the user might take $Z_\tau(x,y)=c_y \ind(y \notin C_\tau(x))$, where $c_y$ is a label-specific miscoverage cost \protect\citepsupp{Bates2021} (e.g., $c_y$ may be large for the label corresponding to a severe condition and small for the label corresponding to no condition). For a distribution $P$, we use $\mathcal{E}_{P,\tau}$ to denote the conditional risk function under $P$, $\mathcal{E}_{P,\tau}: x \mapsto \expect_P[Z_\tau(X,Y) \mid X=x,A=1]$. We assume the following additional condition.

\begin{condition}[Finite second moment] \label{cond: finite variance}
It holds that $\sup_{\tau \in \mathcal{T}_n} \expect_{P^0} [D_\tau(P^0,\mathcal{E}_{0,\tau},g_0,\gamma_0,1)(O)^2] < \infty$.
\end{condition}

Under this additional condition, Theorem~\ref{thm: differentiability of Psi} holds. We do not need to make this assumption for APAC prediction sets because it is automatically satisfied for that problem due to using a binary loss.

With the above notations, the algorithm to construct cross-fit one-step corrected estimators is identical to Algorithm~\ref{alg: CV one step}. To estimate conditional risk $\mathcal{E}_{0,\tau}$, we use estimators appropriate for its structure. For example, if $Z_\tau$ is unbounded, we may estimate $\mathcal{E}_{0,\tau}$ via regression techniques such as least-squares, rather than classification techniques. We also present the algorithm to construct CV-TMLE in Algorithm~\ref{alg: CV-TMLE ARCPS} below. Once an asymptotically efficient estimator of the risk $\Psi_\tau(P^0)$ is available, we estimate the asymptotic variance, construct Wald CUB and select a threshold as described in Section~\ref{sec: CV one-step CI and tau hat} or \ref{sec: TMLE}.

\begin{algorithm}
\caption{CV-TMLE of risk $\Psi_\tau(P^0)$ used in PredSet-TMLE for asymptotic risk contrlling prediction sets} \label{alg: CV-TMLE ARCPS}
\begin{algorithmic}[1]
\State Obtain initial estimators $\hat{g}_n^{-v}$ and $\hat{\mathcal{E}}_{n,\tau}^{-v}$ of nuisance functions $g_0$ and $\mathcal{E}_{0,\tau}$ using the same procedure as Lines~1--4 in Algorithm~\ref{alg: CV one step}.
\For{$\tau \in \mathcal{T}_n$ and $v \in [V]$} \Comment{(Obtain a TMLE for fold $v$ based on sample splitting)}
    \State Compute an ordinary least squares fit with outcome $Z_\tau(X,Y)$, offset $\hat{\mathcal{E}}_{n,\tau}^{-v}(X)$, clever covariate $\frac{1}{\hat{\gamma}_n^v} \mathscr{W}(\hat{g}_n^{-v},\hat{\gamma}_n^v)(X)$, and no intercept; using data with $A=1$ in fold $v$. Set $\tilde{\mathcal{E}}_{n,\tau}^v$ to be the fitted mean model.
    \State Let $\tilde{P}_{\tau}^{n,v}$ be a distribution with the following components: (i) marginal distribution of $A$ being $P_{A}^{n,v}$, (ii) conditional distribution of $X \mid A=0$ being $P_{X \mid 0}^{n,v}$, (iii) distribution of $Z_\tau \mid X,A=1$ with mean given by $\tilde{\mathcal{E}}_{n,\tau}^v$, and (iv) likelihood ratio $\mathscr{W}(\hat{g}_n^{-v},\hat{\gamma}_n^v)$.
    \State Set
        $$\tilde{\psi}_{n,\tau}^v := \Psi^\Gcomp(\tilde{P}_{\tau}^{n,v}) = \frac{\sum_{i \in I_v} (1-A_i) \tilde{\mathcal{E}}_{n,\tau}^v(X_i) }{\sum_{i \in I_v} (1-A_i)}.$$
\EndFor
\For{$\tau \in \mathcal{T}_n$}
    \State Obtain the CV-TMLE $\tilde{\psi}_{n,\tau} := \frac{1}{n} \sum_{v=1}^V |I_v| \tilde{\psi}_{n,\tau}^v$.
\EndFor
\end{algorithmic}
\end{algorithm}

The theoretical results are also very similar to our previous ones. In particular, Theorems~\ref{thm: CV one-step efficiency} and \ref{thm: CV-TMLE efficiency} hold for both above procedures. For other results, we will need to make the following additional assumption, which is slightly stronger than Condition~\ref{cond: finite variance}.
\begin{condition}[Finite third moment] \label{cond: finite 3rd moment}
It holds that $\sup_{\tau \in \mathcal{T}_n} \expect_{P^0} |D_\tau(P^0,\mathcal{E}_{0,\tau},g_0,\gamma_0,1)(O)|^3 < \infty$.
\end{condition}

Again, we do not need to make this assumption for APAC prediction sets because it is automatically satisfied for that problem. Conditions~\ref{cond: finite variance} and \ref{cond: finite 3rd moment} are satisfied if the loss function $Z_\tau$ is bounded for all $\tau \in \bar{\real}$. Under this additional condition, Theorems~\ref{thm: convergence rate of one-step Wald CI} and \ref{thm: convergence rate of TMLE Wald CI} hold. Consequently, Corollary~\ref{corollary: CV one-step APAC} also holds. The proofs of these results are strikingly similar and thus omitted.

We do not study a method corresponding to PredSet-RS because it has worse performance than PredSet-1Step and PredSet-TMLE in the binary loss case, and we expect similar results to hold for general losses.

\section{Double robustness of PredSet-1Step in special cases} \label{sec: one step DR}

As we mentioned in Remark~\ref{rmk: mixed bias}, PredSet-1Step is not doubly robust in general under nonparametric models. We present two special cases where PredSet-1Step is doubly robust in this section.

We assume that there exists functions $\mathcal{E}_{\infty,\tau}$ and $g_\infty$ such that Condition~\ref{cond: sufficient nuisance rate} holds with $(\mathcal{E}_{0,\tau},g_0)$ replaced by $(\mathcal{E}_{\infty,\tau},g_\infty)$. We consider the scenario in which $\mathcal{E}_{\infty,\tau}=\mathcal{E}_{0,\tau}$ or $g_\infty=g_0$, that is, one nuisance function is consistently estimated, but not necessarily both are. 
The key result that makes double robustness possible is that
\begin{equation}
    \expect_{P^0}[D_\tau(P^0,\mathcal{E}_{\infty,\tau},g_\infty,\gamma_0,1)]=0 \label{eq: DR IF}
\end{equation}
if $\mathcal{E}_{\infty,\tau}=\mathcal{E}_{0,\tau}$ \emph{or} $g_\infty=g_0$. In other words, the gradient is doubly robust. This equality can be readily checked by using \eqref{eq: identification}.

We next discuss the double robustness of PredSet-1Step, possibly with slight modifications, in the two special cases separately.

\subsection{Known propensity score or conditional coverage error} \label{sec: one step DR known}

The first case is when $\mathcal{E}_{0,\tau}$ or $g_0$ is known exactly, namely the following condition.
\begin{condition}[Known propensity score or conditional coverage error] \label{cond: known nuisance}
    It holds that $\hat{\mathcal{E}}_{n,\tau}^{-v}=\mathcal{E}_{0,\tau}$ or $\hat{g}_n^{-v}=g_0$.
\end{condition}

This case may occur in practice if the distinction of the two populations is determined by the study design, in which case $g_0$ is known exactly, or if a large, effectively infinite, amount of auxiliary unlabeled data from both populations are available, in which case the error in estimating $g_0$ is negligible. 
For example, it may be expensive to measure outcome $Y$ but inexpensive to measure the covariate $X$. We may adopt a two-stage design: in the first stage, we measure $X$ for all individuals in the study; in the second stage, we randomly select a subset for whom we measure $Y$; thus $g_0$ is known. 
The other case, when the conditional coverage error rate $\mathcal{E}_{0,\tau}$ known exactly, is rarer in practice; nevertheless, it may still be of theoretical interest.

Knowing the propensity score $g_0$ differs substantially from the conventional setting in the literature where the likelihood ratio $w_0$ is known exactly \protect\citepsupp{Tibshirani2019,Lei2021,park2021pac}. 
By \eqref{eq: reparameterize weight}, to derive $g_0$ from $w_0$ or vice versa, the true proportion $\gamma_0$ of source population data must be known. 
This quantity might be determined by the study design, but might also be unknown. Thus, the special case we consider is distinct from the setting previously studied in the literature. The distinction between knowing $g_0$ and knowing $w_0$ is substantial for PredSet-1Step, which is based on semiparametric efficiency theory.

When one nuisance function is known, this function does not need to be estimated, and is used directly in Algorithm~\ref{alg: CV one step}. In this case, PredSet-1Step enjoys similar properties as in Theorems~\ref{thm: CV one-step efficiency}, \ref{thm: convergence rate of one-step Wald CI} and Corollary~\ref{cond: positive variance}, even if the other nuisance function estimator is inconsistent. The only differences are the following.
\begin{enumerate}
    \item $g_0$ is replaced by $g_\infty$ when $\mathcal{E}_{0,\tau}$ is known exactly;
    \item $\mathcal{E}_{0,\tau}$ is replaced by $\mathcal{E}_{\infty,\tau}$ when $g_0$ is known exactly;
    \item $\Delta_{n,\epsilon}$ is replaced by the following term with a rate that is no slower than $\Delta_{n,\epsilon}$:
    $$\epsilon^{-1/2} \bigO \left( \sup_{v \in [V], \tau \in \mathcal{T}_n} \left\{ \expect_{P^0} \left\| \frac{1-\hat{g}_n^{-v}}{\hat{g}_n^{-v}} - \frac{1-g_\infty}{g_\infty} \right\|_{P^0_{X \mid 1},2} + \expect_{P^0} \| \hat{\mathcal{E}}_{n,\tau}^{-v} - \mathcal{E}_{\infty,\tau} \|_{P^0_{X \mid 1},2} \right\} \right)+q_n,$$
    where we recall $\epsilon$ from \eqref{teps} and Conditions~\ref{cond: constant Q for extreme tau}--\ref{cond: positive variance}.
\end{enumerate}
These results can be proved using a similar argument as for Theorem~\ref{thm: CV one-step efficiency} and Corollary~\ref{corollary: CV one-step APAC}.

\subsection{Hadamard differentiable nuisance functions} \label{sec: one step DR parametric}

In the second special case, we assume that the following condition holds.
\begin{condition}[Hadamard differentiable nuisance functions] \label{cond: Hadamard differentiable nuisance}
    Both nuisance functions $\mathcal{E}_{\infty,\tau}$ and $g_\infty$ are Hadamard differentiable \protect\citepsupp[see, e.g., Chapters~20 and 23 in][]{vanderVaart1998}.
\end{condition}

In this case, we can often find estimators $\hat{g}_n$ of $g_\infty$ and $\hat{\mathcal{E}}_{n,\tau}$ of $\mathcal{E}_{\infty,\tau}$ that are both asymptotically linear, in the sense that, for a square-integrable function $\IF^g$ and a square-integrable collection of functions $\IF^\mathcal{E}_\tau$ ($\tau \in \bar{\real}$),
\begin{align*}
    \hat{g}_n(x)-g_0(x) &= \frac{1}{n} \sum_{i=1}^n \IF^g(O_i,x) + \bigO_p(n^{-1}), \\
    \hat{\mathcal{E}}_{n,\tau}(x)-\mathcal{E}_{\infty,\tau}(x) &= \frac{1}{n} \sum_{i=1}^n \IF^\mathcal{E}_\tau(O_i,x) + \bigO_p(n^{-1}),
\end{align*}
for all $x \in \mathcal{X}$ and $\tau \in \bar{\real}$. We assume that both $\bigO_p(n^{-1})$ terms are uniform over $x \in \mathcal{X}$ and $\tau \in \bar{\real}$. This case often occurs when both functions are estimated in parametric models, in which case the estimators $\hat{\mathcal{E}}_{n,\tau}$ and $\hat{g}_n$ can often be obtained via M-estimation or Z-estimation. In other words, this condition holds if both $g_\infty$ and $\mathcal{E}_{\infty,\tau}$ fall in parametric models.
Further, we require that at least one of $g_0$ and $\mathcal{E}_{0,\tau}$ is correctly specified. 
This set of assumptions might hold less commonly than Condition~\ref{cond: sufficient nuisance rate}, which requires consistent estimation of both nuisance functions, but not necessarily in parametric models. Nevertheless, this scenario may be of theoretical interest.

In this case, cross-fitting is not required, because no flexible supervised machine learning method is needed to estimate the nuisance functions. Thus, in the above notations, we already implicitly assume no sample splitting. A method similar to PredSet-1Step would still provide APAC prediction sets. 
First, the estimator from Algorithm~\ref{alg: CV one step} can be simplified to
$$\hat{\psi}_{n,\tau} := \frac{\sum_{i=1}^n (1-A_i) \hat{\mathcal{E}}_{n,\tau}(X_i)}{\sum_{i=1}^n (1-A_i)} + \frac{1}{n} \sum_{i=1}^n \frac{A_i}{\gamma_n} \mathscr{W}(\hat{g}_n,\hat{\gamma}_n) [Z_\tau(X_i,Y_i) - \hat{\mathcal{E}}_{n,\tau}(X_i)]$$
where $\gamma_n := \Prob_{P^{n,v}}(A=1) = \frac{1}{n} \sum_{i=1}^n A_i$. This estimator is still an asymptotically linear estimator of $\Psi_\tau(P^0)$. The asymptotic variance of $\psi_{n,\tau}$, however, may differ from that implied by the gradient under nonparametric models because $\hat{\psi}_{n,\tau}$ may have a different influence function:
\begin{align}
    \begin{split}
    &\hat{\psi}_{n,\tau} - \Psi_\tau(P^0) \\
    &= \frac{1}{n} \sum_{i=1}^n \Bigg\{ \frac{A_i}{1-\gamma_0} \frac{1-g_\infty(X_i)}{g_\infty(X_i)} [Z_\tau(X_i,Y_i) - \mathcal{E}_{\infty,\tau}(X_i)] + \frac{1-A_i}{1-\gamma_0} [\mathcal{E}_{\infty,\tau}-\Psi_\tau(P^0)] \\
    &\qquad+ \frac{\gamma_0}{1-\gamma_0} P^0_{X \mid 1} \left\{ \frac{1}{g_\infty(\cdot)^2} [\mathcal{E}_{\infty,\tau}(\cdot)-\mathcal{E}_{0,\tau}(\cdot)] \IF^g(O_i,\cdot) \right\} \\
    &\qquad- \frac{\gamma_0}{1-\gamma_0} P^0_{X \mid 1} \left\{ \left( \frac{1-g_\infty(\cdot)}{g_\infty(\cdot)} - \frac{1-g_0(\cdot)}{g_0(\cdot)} \right) \IF^\mathcal{E}_\tau(O_i,\cdot) \right\} \Bigg\} + \bigO_p(n^{-1}).
    \end{split} \label{eq: DR one step parametric RAL}
\end{align}
Because of Hadamard differentiability, we may estimate the standard error and construct CUB for $\Psi_\tau(P^0)$ via the nonparametric bootstrap even if we do not know whether $\mathcal{E}_{\infty,\tau}=\mathcal{E}_{0,\tau}$ or $g_\infty=g_0$ \protect\citepsupp[see, e.g.,][for more details on the nonparametric bootstrap]{Efron1994,vandervaart1996,vanderVaart1998,Hall2013}. With CUBs computed, a threshold can be selected similarly as described in Section~\ref{sec: CV one-step CI and tau hat}. The convergence rate $\Delta_{n,\epsilon}$ in the nonparametric case can be improved to $n^{-1/2}$ (first-order accuracy), and might be further improved to $n^{-1}$ (second-order accuracy) for advanced inferential techniques based on the nonparametric bootstrap \protect\citepsupp{Hall2013}.

\section{Example of uniform convergence in Condition~\ref{cond: sufficient nuisance rate}} \label{sec: E uniform convergence example}

In this section, we study a simple example to show that the convergence of $\mathcal{E}_{n,\tau}^{-v}$ uniformly over $\tau \in \mathcal{T}_n$ may be plausible and may be verified with existing techniques for nonparametric regression. In particular, we study kernel estimators \protect\citepsupp{Nadaraya1964,Watson1964}.

Suppose that the covariate $X$ lies in a compact subset $\mathcal{X}$ of the real line and the outcome $Y$ is generated as $Y=f(X)+\epsilon$ for an unknown function $f$ and exogenous continuous noise $\epsilon$ with bounded Lebesgue density $\pi$. Consider the scoring function $s(x,y)=-|y-\tilde{f}(x)|$ where $\tilde{f}$ is an estimate of $f$ obtained from a held-out sample. 
Suppose that $\mathcal{E}_{0,\tau}(x)$ is estimated with a kernel estimator $\mathcal{E}_{n,\tau}^{-v}(x)=\sum_{i \notin I_v,A_i=1} K_h(X_i-x) \ind(|Y_i-\tilde{f}(X_i)| > \tau)/\sum_{i \notin I_v,A_i=1} K_h(X_i-x)$ for a kernel $K_h$ with bandwidth $h$. 
Letting $\Delta := f-\tilde{f}$, we have
$$\mathcal{E}_{0,\tau}(x) = \Prob_{P^0}(|\Delta(x) + \epsilon| > -\tau) = \begin{cases}
1 - \int_{\tau-\Delta(x)}^{-\tau-\Delta(x)} \pi(\varepsilon) \intd \varepsilon & \text{ if } \tau \leq 0 \\
1 & \text{ if } \tau > 0
\end{cases}$$
and, for some $\const$ that does not depend on $\tau,x_1,x_2$,
\begin{align*}
    |\mathcal{E}_{0,\tau}(x_1) - \mathcal{E}_{0,\tau}(x_2)| &= \begin{cases}
    \left| \int_{\tau - \Delta(x_1)}^{\tau-\Delta(x_2)} \pi(\varepsilon) \intd \varepsilon - \int_{-\tau - \Delta(x_1)}^{-\tau-\Delta(x_2)} \pi(\varepsilon) \intd \varepsilon \right| & \text{ if } \tau \leq 0 \\
    0 & \text{ if } \tau > 0
    \end{cases} \\
    &\leq \const |\Delta(x_1) - \Delta(x_2)|.
\end{align*}
If $\Delta$ is Lipschitz continuous---which holds if both $f$ and $\tilde{f}$ are Lipschitz continuous in $\mathcal{X}$---then so is $\mathcal{E}_{0,\tau}$. 
Suppose that $K_h$ is taken to be the kernel $K_h(x)=h^{-1} \ind(|x|/h \leq 1)$. 
Then by the fact that $\mathcal{E}_{0,\tau}(x) \in [0,1]$, Condition~\ref{cond: bounded weight}, and Theorem~5.2 of \protect\citetsupp{Gyofi2002}, we have that $\expect_{P^0} \| \mathcal{E}_{n,\tau}^{-v} - \mathcal{E}_{0,\tau} \|_{P^0_{X \mid 0},2} \leq \const (h+(nh)^{-1/2})$ for a constant $\const$ independent of $\tau$. Therefore, choosing $h \propto n^{-1/3}$ yields that
$$\expect_{P^0} \sup_{v \in [V], \tau \in \bar{\real}} \| \hat{\mathcal{E}}_{n,\tau}^{-v}-\mathcal{E}_{0,\tau} \|_{P^0_{X \mid 0},2} = \bigO(n^{-1/3}) = \smallo(n^{-1/4}).$$
We have thus established a convergence rate at $\smallo(n^{-1/4})$ of $ \| \hat{\mathcal{E}}_{n,\tau}^{-v}-\mathcal{E}_{0,\tau} \|_{P^0_{X \mid 0},2}$ uniformly over $\tau \in \bar{\real}$.

The propensity score function can also be estimated at rate $\smallo(n^{-1/4})$ under appropriate smoothness conditions on $g_0$. Thus, in this example, Condition~\ref{cond: sufficient nuisance rate} holds.

The analysis of the above example is very similar to that of a kernel estimator $\mathcal{E}_{n,\tau}^{-v}$ at a single $\tau$. 
Thus, we expect that in similar cases, the uniformity over $\tau \in \mathcal{T}_n$ in Condition~\ref{cond: sufficient nuisance rate} may not stringent.

\section{Analysis results of full HIV risk prediction data in South Africa} \label{sec: data analysis2}

The analysis method of the full data is almost identical to that described in Section~\ref{sec: data analysis}. A major difference between these two analyses is that all participants in the target population are included in this full data analysis. 
The numbers of participants in the data splits also differ. In the full data analysis, we randomly select 7249 participants from the source population to train the scoring function $s$; the rest of the data consisting of 5136 participants from each population is used to select the threshold.

The empirical coverage of the prediction set construction methods in the sample from the target population is presented in Table~\ref{tab: data analysis result2}. The coverage of all methods are below the target coverage level $1-\alpha_\error=95\%$, but the coverage of our methods PredSet-1Step and PredSet-TMLE is closer to 95\%. The improvement in the full data is not as significant as in the subset analyzed in Section~\ref{sec: data analysis}. One reason for this decrease in improvement could be that the two moderately and severely shifted covariates, namely wealth quintile and community HIV prevalence, are not strongly causally related to the outcome compared to other covariates without significant shift (for example, community ART coverage), and thus the optimal thresholds in the two populations do not differ by much. 
One reason for our methods not providing at least 95\% coverage could be that the crucial covariate shift assumption (Condition~\ref{cond: same Y|X}) fails to hold exactly in this data set. Our methods are likely to provide better coverage if more covariates that are predictive of the outcome are available in the data.

\begin{table}
    \centering
    \caption{Empirical coverage of prediction sets, 95\% Wilson score confidence interval for coverage, and selected thresholds in the full sample from the target population in the South Africa HIV trial data. The target coverage is at least $1-\alpha_\error=95\%$, with probability 95\% over the training data.}
    \label{tab: data analysis result2}
    \begin{tabular}{l|r|r|r}
        Method & Empirical coverage & Coverage CI & Selected threshold $\hat{\tau}$ \\
        \hline\hline
        PredSet-1Step & 94.08\% & 93.40\%--94.69\% & 0.190 \\
        PredSet-TMLE & 93.81\% & 93.12\%--94.44\% & 0.200 \\
        Inductive Conformal Prediction & 93.03\% & 92.30\%--93.59\% & 0.227
    \end{tabular}
\end{table}

\section{Proofs} \label{sec: proof}

\subsection{Negative result on prediction sets with finite-sample coverage} \label{sec: proof trivial finite sample prediction set}

We first prove a negative result similar to Lemma~\ref{lemma: trivial finite sample prediction set} for prediction sets with a finite-sample marginal coverage guarantee.

\begin{lemma} \label{lemma: trivial finite sample prediction set marginal}
    Consider the same setting as in Lemma~\ref{lemma: trivial finite sample prediction set}. Suppose that a (possibly randomized) prediction set $\hat{C}$ has finite-sample marginal coverage guarantee in the target population, that is,
    \begin{equation}
        \Prob_{\bar{P^0}}(Y \notin \hat{C}(X) \mid A=0) \leq \alpha \label{eq: finite sample marginal coverage}
    \end{equation}
    for some given $\alpha \in (0,1)$ and any $\bar{P}^0 \in \bar{\modelspace}^*$. Then, for any $\bar{P^0} \in \bar{\modelspace}^*$ and a.e. $y \in \mathcal{Y}$ with respect to the Lebesgue measure,
    $$\Prob_{\bar{P}^0}(y \notin \hat{C}(X) \mid A=0) \leq \alpha.$$
\end{lemma}
\begin{proof}
    Let $\bar{\modelspace} \supseteq \bar{\modelspace}^*$ be the space of distributions for the full data point $\bar{O}$ with the distribution of $(X,Y) \mid A=a$ dominated by the Lebesgue measure for each $a \in \{0,1\}$.
    These distributions may or may not satisfy the covariate shift assumption (namely Conditions~\ref{cond: positivity of P(A)}--\ref{cond: target dominated by source}). For any $x \in \mathcal{X}$, we will sometimes write $\hat{C}(x)$ as $C(x;O_1,\ldots,O_n)$ where $C$ is the prediction set construction algorithm that takes the training data $(O_1,\ldots,O_n)$ and a future observed covariate $x$ as inputs and outputs a prediction set. 
    This notation is helpful to clarify the dependence of $\hat{C}$ on the observed training data $(O_1,\ldots,O_n)$. We use $\bar{O}_{n+1}$ to denote the full data point from a future draw.
    
    Define the randomized test $\eta(\bar{O}_1,\ldots,\bar{O}_{n+1})$ as follows: if $A_{n+1}=0$, set $\eta(\bar{O}_1,\ldots,\bar{O}_{n+1}) = \ind(Y_{n+1} \notin C(X_{n+1}; O_1,\ldots,O_n))$; if $A_{n+1}=1$, set $\eta(\bar{O}_1,\ldots,\bar{O}_{n+1})$ to be one with probability $\alpha$ and zero otherwise. We note that although $\eta$ is a function of $n+1$ full data points $(\bar{O}_1,\ldots,\bar{O}_{n+1})$, it only relies on $n$ observed training data points $(O_1,\ldots,O_n)$ and one future full data point $\bar{O}_{n+1}$. By \eqref{eq: finite sample marginal coverage} and the definition of the test $\eta$, we have that
    $$\Prob_{\bar{P}^0}(\eta(\bar{O}_1,\ldots,\bar{O}_{n+1})=1) \leq \alpha \qquad \text{for any distribution $\bar{P}^0 \in \bar{\modelspace}^*$.}$$
    Therefore, $\eta$ can be viewed as a test with level $\alpha$ for the null hypothesis $Y \independent A \mid X$. By Theorem~2 and Remark~4 in \protect\citetsupp{Shah2020}, which essentially state that the power of $\eta$ against the alternative hypothesis is at most $\alpha$, we have that
    \begin{equation}
        \Prob_{\bar{Q}}(\eta(\bar{O}_1,\ldots,\bar{O}_{n+1})=1) \leq \alpha \qquad \text{for any distribution $\bar{Q} \in \bar{\modelspace}$}. \label{eq: no power}
    \end{equation}
    For any $x \in \mathcal{X}$, let $D_x \subseteq \mathcal{Y}$ be any Lebesgue measurable set with nonzero finite measure and $\mathscr{U}_x$ be the uniform distribution on $D_x$.
    We may take $\bar{Q}$ to be a distribution such that (i) the distribution of $A$ is an arbitrary Bernoulli distribution with success probability in $(0,1)$, (ii) the distribution of $X \mid A=a$ ($a \in \{0,1\}$) satisfies the dominance condition \ref{cond: target dominated by source} and is arbitrary in all other aspects, and (iii) the distribution of $Y \mid X=x,A=1$ is arbitrary and the distribution of $Y \mid X=x,A=0$ is $\mathscr{U}_x$. 
    We take $\bar{P}^0$ to be the distribution that is identical to $\bar{Q}$, except that the distribution of $Y \mid X=x,A=0$ is identical to $Y \mid X=x,A=1$ rather than $\mathscr{U}_x$ under $\bar{P}^0$. Note that $\bar{P}^0 \in \bar{\modelspace}^*$. Since $\hat{C}$ is trained only on observed data $(O_1,\ldots,O_n)$ and $\bar{P}^0$ and $\bar{Q}$ imply the same distribution of the observed data point $O=(A,X,AY)$, by \eqref{eq: no power}, we have that
    \begin{align*}
        &\Prob_{\bar{Q}}(Y_{n+1} \notin C(X_{n+1};O_1,\ldots,O_n) \mid A_{n+1}=0) \\
        &= \int_{\mathcal{Y}} \Prob_{\bar{P}^0}(y \notin C(X_{n+1};O_1,\ldots,O_n) \mid A_{n+1}=0) \mathscr{U}_{X_{n+1}}(\intd y)\leq \alpha,
    \end{align*}
    where the probability is over training data and possible exogenous randomness in $C$. Since $D_x$ is arbitrary, it follows that the integrand is bounded by $\alpha$, namely
    $$\Prob_{\bar{P}^0}(y \notin C(X_{n+1};O_1,\ldots,O_n) \mid A_{n+1}=0) \leq \alpha$$
    for a.e. $y \in \mathcal{Y}$.
    The desired result follows by replacing the notations $X_{n+1}$ and $A_{n+1}$ with $X$ and $A$, respectively, and noting that $C(x;O_1,\ldots,O_n)=\hat{C}(x)$ by definition.
\end{proof}

A similar result holds when $Y$ is discrete by Remark~4 in \protect\citetsupp{Shah2020}. We note that the above argument does not apply to the case where the covariate shift is known, in which case finite-smaple coverage guarantee is known to be achievable \protect\citepsupp[see, e.g.,][]{Tibshirani2019,park2021pac,Lei2021}. The reason is that a key result we rely on, Theorem~2 and Remark~4 in \protect\citetsupp{Shah2020}, requires no restrictions on the joint distribution of $(A,X,Y)$ except for the common dominating measure, but knowledge about the likelihood ratio of the covariate shift restricts the set $\bar{\modelspace}$ of possible joint distributions in a nontrivial way.

We now show that PAC guarantee implies a marginal coverage guarantee. A similar result holds when no covariate shift is present and the proof is almost identical.
\begin{lemma} \label{lemma: PAC implies marginal}
    Suppose that $\hat{C}$ is PAC:
    $$\Prob_{\bar{P^0}} \left( \Prob_{\bar{P^0}}(Y \notin \hat{C}(X) \mid A=0,\hat{C}) \leq \alpha_\error \right) \geq 1-\alpha_\conf.$$
    Then, $\hat{C}$ satisfies a marginal coverage guarantee:
    $$\Prob_{\bar{P^0}}(Y \notin \hat{C}(X) \mid A=0) \leq \alpha_\error+\alpha_\conf.$$
\end{lemma}
\begin{proof}
    The proof is simple by considering the two cases where $\hat{C}$ is approximately correct or not. Since $(A,X,Y)$ is independent of $\hat{C}$,
    \begin{align*}
        & \Prob_{\bar{P^0}}(Y \notin \hat{C}(X) \mid A=0) \\
        &= \Prob_{\bar{P^0}} \left( Y \notin \hat{C}(X) \mid A=0,\Prob_{\bar{P^0}}(Y \notin \hat{C}(X) \mid A=0,\hat{C}) \leq \alpha_\error \right) \\
        &\qquad\times \Prob_{\bar{P}^0} \left( \Prob_{\bar{P^0}}(Y \notin \hat{C}(X) \mid A=0,\hat{C}) \leq \alpha_\error \right) \hspace{1in} (\text{$\hat{C}$ is approximately correct}) \\
        &\quad+ \Prob_{\bar{P^0}} \left( Y \notin \hat{C}(X) \mid A=0, \Prob_{\bar{P^0}}(Y \notin \hat{C}(X) \mid A=0,\hat{C}) > \alpha_\error \right) \\
        &\qquad\times \Prob_{\bar{P}^0} \left(\Prob_{\bar{P^0}}(Y \notin \hat{C}(X) \mid A=0,\hat{C}) > \alpha_\error \right) \hspace{1in} (\text{$\hat{C}$ is not approximately correct}) \\
        &\leq \alpha_\error \times 1 + 1 \times \alpha_\conf = \alpha_\error + \alpha_\conf.
    \end{align*}
\end{proof}

Lemma~\ref{lemma: trivial finite sample prediction set} follows from Lemmas~\ref{lemma: trivial finite sample prediction set marginal}--\ref{lemma: PAC implies marginal}.

\subsection{Pathwise differentiability (Theorem~\ref{thm: differentiability of Psi})} \label{sec: proof differentiability}

We first show that $D^\Gcomp_\tau(P,\mathcal{E}_{P,\tau},g_P,\gamma_P,1)$ and $D^\weighted_\tau(P,\mathcal{E}_{P,\tau},g_P,\gamma_P,1)$ are identical. Using tedious but elementary algebra, we can show that
\begin{align*}
    D^\weighted_\tau(P,\mathcal{E},g,\gamma,\pi)(o) &= \frac{a}{\gamma_P} \frac{\mathscr{W}(g,\gamma)(x)}{\Pi_P(\mathscr{W}(g,\gamma))} \left\{ Z_\tau(x,y) - \mathcal{E}(x) \right\}
    + \frac{1-a}{1-\gamma_P} [ \mathcal{E}(x) - \Psi^\weighted_\tau(P) ] \\
    &+ \mathcal{E}(x) \left\{ a \mathscr{W}(g,\gamma) \left( \frac{1}{\gamma_P \Pi_P(\mathscr{W}(g,\gamma))} - \frac{1}{\gamma \pi} \right) + (1-a) \left( \frac{1}{1-\gamma} - \frac{1}{1-\gamma_P} \right) \right\},
\end{align*}
and it follows that $D^\Gcomp_\tau(P,\mathcal{E}_{P,\tau},g_P,\gamma_P,1) \equiv D^\weighted_\tau(P,\mathcal{E}_{P,\tau},g_P,\gamma_P,1)$.

We next prove pathwise differentiability results. Throughout the rest of the Supplemental Material, for a distribution $P$ and a $P$-measurable function $f$, we define $Pf := \int f(o) P(\intd o) = \expect_P[f(O)]$. We also use $P^n$, $n \ge 1$, to denote the empirical distribution.

To prove pathwise differentiability of the parameter $\Psi_\tau$, we consider the Hilbert space $L_0^2(P^0):=\{f: P^0 f=0, P^0 f^2 < \infty\}$ with the covariance inner product $\langle f,g \rangle \mapsto P^0 f g = \expect_{P^0}[f(O) g(O)] = \int f(o) g(o) P^0(\intd o)$. Here, we define $f(o)=f(y,x,a)$ to be $f(0,x,a)$ when $a=0$, since $Y$ is missing when $A=0$. The $L^2(P^0)$-closure of the set of all bounded functions in $L_0^2(P^0)$ is $L_0^2(P^0)$, and thus it suffices to consider bounded functions to prove pathwise differentiability \protect\citepsupp[see, e.g., pages~68--69 in][]{Tsiatis2006}.

Let $H$ be any bounded function in $L_0^2(P^0)$.  Define $H_A: a \mapsto \expect_{P^0}[H(O) \mid A=a]$, $H_X: (x \mid a) \mapsto \expect_{P^0}[H(O) \mid X=x,A=a] - H_A(a)$, and $H_Y: (y \mid x,a) \mapsto H(o) - H_X(x \mid a) - H_A(a)$. Since $Y$ is missing when $A=0$, we have that $H_Y(y \mid x,0)=0$.
These functions are orthogonal projections of $H$ onto certain orthogonal subspaces of $L_0^2(P^0)$, which aids calculation of expectations. Indeed, consider the following subspaces of $L_0^2(P^0)$ and their corresponding uncentered versions:
\begin{itemize}
    \item $L_0^2(\emptyset):=\{o \mapsto 0\}$, $L^2(\emptyset) := \{o \mapsto c: c \in \real\}$;
    \item $L_0^2(A):= \{f \in L_0^2(P^0): f \text{ only depends on } a\}$, $L^2(A):= \{f \in L^2(P^0): f \text{ only depends on } a\}$;
    \item $L_0^2(X \mid A):= \{ f \in L_0^2(P^0): f \text{ only depends on } (x,a), \, \expect_{P^0}[f(X,A) \mid A=a]=0,\, \forall\, a \in \{0,1\} \}$, $L^2(X,A):= \{ f \in L^2(P^0): f \text{ only depends on } (x,a)\}$;
    \item $L_0^2(Y \mid X,A):= \{ f \in L_0^2(P^0): \expect_{P^0}[f(O) \mid X=x,A=a] = 0,\,\text{ for } P^0-\text{a.s. } x \in \mathcal{X}, \forall\, a \in \{0,1\} \}$. Here $\expect_{P^0}[f(O) \mid X=x,A=a]$ takes different integral forms for $a=0$ and $a=1$:
    $$\expect_{P^0}[f(O) \mid X=x,A=a] =
    \begin{cases}
    \int f(y,x,a) P^0_{Y \mid x}(\intd y) & a=1, \\
    f(0,x,a) & a=0.
    \end{cases}$$
\end{itemize}
It is not difficult to verify that $H_A \in L_0^2(A)$, $H_X \in L_0^2(X \mid A)$ and $H_Y \in L_0^2(Y \mid X,A)$. Moreover, the subspaces $L_0^2(\emptyset)$, $L_0^2(A)$, $L_0^2(X \mid A)$ and $L_0^2(Y \mid X,A)$ are mutually orthogonal. The following result summarizes some of these results, with some additional claims we will use. Since its proof is direct, we omit it.
\begin{lemma}[Orthogonality properties]\label{orth} We have that
\begin{enumerate}
    \item $L_0^2(A)$ is orthogonal to $L^2(\emptyset)$;
    \item For any $f \in L_0^2(X \mid A)$ and $g \in L^2(A)$, it holds that
    $\expect_{P^0}[f(X,A) g(A) \mid A=a]=0,\, \forall\, a \in \{0,1\}$
    \item For any $f \in L_0^2(Y \mid X,A)$ and $g \in L^2(X,A)$, it holds that $\expect_{P^0}[f(O) g(X,A) \mid X=x,A=a] = 0$, for  $P^0$-a.s. $x \in \mathcal{X}, a \in \{0,1\}$
\end{enumerate}
\end{lemma}

For all $\epsilon \in \R$ sufficiently close to zero, define $P^\epsilon$ via its Radon-Nikodym derivative 
$$\frac{\intd P^\epsilon}{\intd P^0}: o \mapsto (1+\epsilon H_A(a)) (1+\epsilon H_X(x \mid a)) (1+\epsilon H_Y(y \mid x,a)).$$
It is not difficult to verify that $H$ is the score function of $P^\epsilon$ for $\epsilon$ at $\epsilon=0$ \protect\citepsupp[see e.g., Definition~1.1.1, Equations~1.1.2 and 1.1.5 in Chapter~1,][]{Pfanzagl1985}. In addition, the following holds
\begin{align*}
    & \frac{\intd P^\epsilon_{Y \mid x}}{\intd P^0_{Y \mid x}}(y)=1+\epsilon H_Y(y \mid x,1), \quad \frac{\intd P^\epsilon_{X \mid a}}{\intd P^0_{X \mid a}}(x)=1+\epsilon H_X(x \mid a), \quad \frac{\intd P^\epsilon_{A}}{\intd P^0_A} = 1+ \epsilon H_A(a).
\end{align*}

For a general functional $\Phi: P \mapsto \Phi(P) \in \real$, to prove that its canonical gradient at $P^0$ under a nonparametric model is $\IF \in L_0^2(P^0)$, it suffices to show that \protect\citepsupp[see e.g., Definition~4.1.1 and the associated discussions in Chapter~4,][]{Pfanzagl1985}
\begin{equation} \label{eq: general pathwise differentiability}
    \left. \frac{\intd \Phi(P^\epsilon)}{\intd \epsilon} \right|_{\epsilon=0} = P^0 (\IF \cdot H) = \expect_{P^0}[\IF(O) H(O)].
\end{equation}
We refer, for example, to \protect\citetsupp{Pfanzagl1985} and \protect\citetsupp{Pfanzagl1990}, for a more in-depth introduction to pathwise differentiability.

\begin{proof}[Proof of Theorem~\ref{thm: differentiability of Psi}]
We focus on the pathwise differentiability of $\Psi^\Gcomp_\tau$ in this proof. Since $\Psi^\weighted_\tau(P)=\Psi^\Gcomp_\tau(P)$ and $D^\weighted_\tau(P,\mathcal{E}_{P,\tau},g_P,\gamma_P,1)=D^\Gcomp_\tau(P,\mathcal{E}_{P,\tau},g_P,\gamma_P,1)$ for any distribution $P$, 
the results follow immediately once we prove the result for $\Psi^\Gcomp_\tau$. Recalling from \eqref{qdef}
that 
$\mathcal{E}_{0,\tau}(x) = \expect_{P^0}[Z_\tau \mid X=x,A=1]$
and
$\Psi^\Gcomp_\tau(P^0)=\expect_{P^0}[\mathcal{E}_{0,\tau}(X) \mid A=0] = \iint Z_\tau(x,y) P^0_{Y \mid x}(\intd y) P^0_{X \mid 0}(\intd x)$, it holds that
$$\Psi^\Gcomp_\tau(P^\epsilon) = \iint (1+\epsilon H_X(x \mid 0)) Z_\tau(x,y) (1+\epsilon H_Y(y \mid x,1)) P^0_{Y \mid x}(\intd y) P^0_{X \mid 0}(\intd x).$$
By the chain rule,
we have that 
\begin{align*}
     &\left. \frac{\intd \Psi^\Gcomp_\tau(P^\epsilon)}{\intd \epsilon} \right|_{\epsilon=0}= \iint Z_\tau(x,y) H_Y(y \mid x,1) P^0_{Y \mid x}(\intd y) P^0_{X \mid 0}(\intd x) \\
    &\quad+ \iint H_X(x \mid 0) Z_\tau(x,y) P^0_{Y \mid x}(\intd y) P^0_{X \mid 0}(\intd x) \\
    &= \iint (Z_\tau(x,y) - \mathcal{E}_{0,\tau}(x)) H_Y(y \mid x,1) P^0_{Y \mid x}(\intd y) P^0_{X \mid 0}(\intd x) \\
    &\quad+ \int H_X(x \mid 0) [ \mathcal{E}_{0,\tau}(x) - \Psi_\tau(P^0) ] P^0_{X \mid 0}(\intd x),
\end{align*}
where we centered functions that are multiplied by the projections $H_Y$ and $H_X$ of $H$ onto orthogonal subspaces. 
In the first term, we subtracted $\mathcal{E}_{0,\tau}(x)$ from $Z_\tau(x,y)$, which does not change the value of the integral due to Part 3 of Lemma \ref{orth}, because $\mathcal{E}_{0,\tau} \in L^2(X,A)$ and $H_Y \in L_0^2(Y \mid X,A)$. 
In the second term, we subtracted the constant $\Psi_\tau(P^0)$ from $\mathcal{E}_{0,\tau}$ and this does not change the value of the integral due to Part 2 of Lemma \ref{orth}, because the function $o \mapsto \Psi_\tau(P^0)$ lies in $L^2(\emptyset) \subseteqq L^2(A)$ and $H_X$ lies in $L_0^2(X \mid A)$.

We next rewrite the above two terms as integrals with respect to the measure $P^0$ (rather than the specific components of $P^0$) by multiplying $\ind(a=1)$ or $\ind(a=0)$ with an inverse probability weight $1/P^0(A=1)$ or $1/P^0(A=0)$, respectively. Then, with $\gamma_0$ in \eqref{g0gamma0} and $w_0$ in Condition~\ref{cond: target dominated by source}, the above equals
\begin{align*}
    &\iiint \frac{\ind(a=1)}{\gamma_0} w_0(x) [Z_\tau(x,y) - \mathcal{E}_{0,\tau}(x)] H_Y(y \mid x,1) P^0_{Y \mid x}(\intd y) P^0_{X \mid a}(\intd x) P^0_{A}(\intd a) \\
    &\quad+ \iint \frac{\ind(a=0)}{1-\gamma_0} [ \mathcal{E}_{0,\tau}(x) - \Psi^\Gcomp_\tau(P^0) ] H_X(x \mid a) P^0_{X \mid a}(\intd x) P^0_{A}(\intd a) \\
    &= \int \frac{\ind(a=1)}{\gamma_0} w_0(x) [Z_\tau(x,y) - \mathcal{E}_{0,\tau}(x)] H_Y(y \mid x,1) P^0(\intd o) \\
    &\quad+ \int \frac{\ind(a=0)}{1-\gamma_0} [ \mathcal{E}_{0,\tau}(x) - \Psi^\Gcomp_\tau(P^0) ] H_X(x \mid a) P^0(\intd o).
\end{align*}

We then use the binary nature of $A$, substitute $w_0$ with $g_0$ and $\gamma_0$ according to \eqref{eq: reparameterize weight}, 
and recall the definition of $D_\tau$ from \eqref{dtau}
to show that the above equals
\begin{align*}
    &\int \frac{a}{\gamma_0} \mathscr{W}(g_0,\gamma_0)(x) [Z_\tau(x,y) - \mathcal{E}_{0,\tau}(x)] H_Y(y \mid x,a) P^0(\intd o) \\
    &\quad+ \int \frac{1-a}{1-\gamma_0} [ \mathcal{E}_{0,\tau}(x) - \Psi^\Gcomp_\tau(P^0) ] H_X(x \mid a) P^0(\intd o) \\
    &= \int \frac{a}{\gamma_0} \mathscr{W}(g_0,\gamma_0)(x) [Z_\tau(x,y) - \mathcal{E}_{0,\tau}(x)] H(o) P^0(\intd o) \\
    &\quad+ \int \frac{1-a}{1-\gamma_0} [ \mathcal{E}_{0,\tau}(x) - \Psi^\Gcomp_\tau(P^0) ] H(o) P^0(\intd o) \\
    &= P^0 [D^\Gcomp_\tau(P^0,\mathcal{E}_{0,\tau},g_0,\gamma_0,1) H].
\end{align*}
Here, we used the orthogonality properties of $L_0^2(A)$, $L_0^2(X \mid A)$ and $L_0^2(Y \mid X,A)$ to replace $H_Y$ and $H_X$ by $H$. Indeed, the function $o \mapsto \frac{a}{\gamma_0} \mathscr{W}(g_0,\gamma_0)(x) [Z_\tau(x,y) - \mathcal{E}_{0,\tau}(x)]$ lies in $L_0^2(Y \mid X,A)$ and thus is orthogonal to $H_X$ and $H_A$; the function $o \mapsto \frac{1-a}{1-\gamma_0} [ \mathcal{E}_{0,\tau}(x) - \Psi^\Gcomp_\tau(P^0) ]$ lies in $L_0^2(X \mid A)$, and is thus orthogonal to $H_A$ and $H_Y$. In conclusion, we have shown that
$$\left. \frac{\intd \Psi^\Gcomp_\tau(P^\epsilon)}{\intd \epsilon} \right|_{\epsilon=0} = P^0 [D^\Gcomp_\tau(P^0,\mathcal{E}_{0,\tau},g_0,\gamma_0,1) H].$$
It is not difficult to check that $P^0 D^\Gcomp_\tau(P^0,\mathcal{E}_{0,\tau},g_0,\gamma_0,1) =0$ and $P^0 D^\Gcomp_\tau(P^0,\mathcal{E}_{0,\tau},g_0,\gamma_0,1)^2 < \infty$ and thus $D^\Gcomp_\tau(P^0,\mathcal{E}_{0,\tau},g_0,\gamma_0,1) \in L_0^2(P^0)$. The desired pathwise differentiability result follows.
\end{proof}

\subsection{Asymptotic efficiency of cross-fit one-step corrected estimators (Theorem~\ref{thm: CV one-step efficiency})} \label{sec: proof efficiency}

We will rely on empirical process theory when proving properties of estimators and CUBs in this and the next few subsections. We refer the readers to, for example, \protect\citetsupp{vandervaart1996}, \protect\citetsupp{Kosorok2008} and \protect\citetsupp{Gine2016} for in-depth introductions to this field. We list some useful results and specific references below:
\begin{itemize}
    \item Suppose that a function class $\funclass$ has a $P$-square-integrable envelope function $F$ (i.e., $\sup_{f \in \funclass} |f(o)|$ $\leq F(o)$ for all $o$, and $P( F^2) <\infty$). 
    If $\funclass$ is a VC-subgraph class (also referred to as ``$\funclass$ is VC-subgraph''), i.e., the subgraphs of the functions in $\funclass$ are a family of sets with a finite VC dimension, then $\funclass$ has bounded uniform entropy integral \protect\citepsupp[BUEI, see page 162, Section 9.1.2 of][]{Kosorok2008} by Theorem~9.3 in \protect\citetsupp{Kosorok2008}. 
    This further implies that $\funclass$ is $P$-Donsker \protect\citepsupp[see the definition on page~128 of][]{Kosorok2008} by Theorem~8.19 in \protect\citetsupp{Kosorok2008}. We will not concern ourselves with the measurability issues that can arise when taking suprema of stochastic processes over uncountable classes, and note that this issue can be resolved by replacing expectations with outer expectations. See, for example, Section~2.3 in \protect\citetsupp{vandervaart1996} or Section~2.2 in \protect\citetsupp{Kosorok2008}, for more details.
    \item Examples of, and results on, preservation (also termed \textit{permanence}) of VC-subgraph, BUEI, and Donkser classes can be found---for example---in Chapter~9 in \protect\citetsupp{Kosorok2008} and Chapter~2 in \protect\citetsupp{vandervaart1996}. These tools make it convenient to show that many commonly encountered function classes---including those we use in the proofs---are VC-subgraph, or BUEI or $P^0$-Donsker.
    \item If a function class $\funclass$ is BUEI with envelope function $F$, then $\expect_{P^0} \sqrt{n} \sup_{f \in \funclass} |(P^n-P^0) f| \lesssim \expect_{P^0}[(P^n F^2)^{1/2}] \leq \|F\|_{P^0,2}$, where the constant hidden in $\lesssim$ involves the uniform entropy integral. Here, the second inequality follows from Jensen's inequality. The first inequality follows from a slight modification of, for example, the proof of Theorem~2.5.2 (particularly the fourth displayed equation on page~128) in \protect\citetsupp{vandervaart1996}, or the proof of Theorem~8.19 (particularly the first displayed equation on page~151) in \protect\citetsupp{Kosorok2008}. The modification is to replace $\funclass_{\delta_n}$ with $\funclass$. This modification is unessential because restricting the $L^2(P^0)$-norm of the function class does not affect the arguments.
\end{itemize}

We will also rely on the following Delta-method for influence functions.
\begin{lemma}[Delta-method for influence functions]\label{iflemma} Suppose that $\phi_n = \phi_0 + P^n \IF + \smallo_p(n^{-1/2})$ where $\phi_0$ is a fixed quantity, $\phi_n$ is a random variable, $\IF$ is a fixed square-integrable function with mean zero, and both $\phi_0$ and $\phi_n$ may be vectors. Let $f$ be a (possibly vector-valued) function that is continuously differentiable at $\phi_0$ with derivative $D f(\phi_0)$. Then $f(\phi_n) = f(\phi_0) + P^n D f(\phi_0) \IF + \smallo_p(n^{-1/2})$. 
If $f$ is twice differentiable at $\phi_0$ with Hessian $D^2 f(\phi_0)$, then $f(\phi_n) = f(\phi_0) + P^n D f(\phi_0) \IF + \bigO_p(n^{-1})$.
\end{lemma}
Indeed, $f(\phi_n) = f(\phi_0) + Df(\phi_0) (\phi_n-\phi_0) + \smallo(\| \phi_n-\phi_0 \|) = f(\phi_0) + P^n D f(\phi_0) \IF + \smallo_p(n^{-1/2})$ by a first-order Taylor expansion of $f(\phi_n)$ around $\phi_0$ and because $P^n \IF=\bigO_p(n^{-1/2})$. If $f$ is twice differentiable at $\phi_0$ with Hessian $D^2 f(\phi_0)$, then $f(\phi_n) = f(\phi_0) + Df(\phi_0) (\phi_n-\phi_0) + \frac{1}{2} (\phi_n-\phi_0)^\top D^2 f(\phi_0) (\phi_n-\phi_0) + \smallo(\|\phi_n-\phi_0\|^2) = f(\phi_0) + P^n D f(\phi_0) \IF + \bigO_p(n^{-1})$ by a second-order Taylor expansion and because $P^n \IF=\bigO_p(n^{-1/2})$. This proves the lemma.

For any $\tau \in \bar{\real}$, any distribution $P$, any $g: \mathcal{X} \mapsto (0,1]$ and $\mathcal{E}: \mathcal{X} \mapsto [0,1]$, we define the following remainder term, with $\Psi^\Gcomp_\tau(P),\mathcal{E}_{P,\tau}$ from \eqref{qdef}, 
$\gamma_P$ from \eqref{gpgammap},
$D_\tau$ from \eqref{dtau},
$\gamma_0,w_0$ from \eqref{g0gamma0}, 
$\mathscr{W}$ from \eqref{mw},
the marginal distributions of $X$ in the target population and source population---$P^0_{X \mid 0}$ and $P^0_{X \mid 1}$, respectively---and
$\mathcal{E}_{0,\tau}$ from \eqref{q0tau}:
\begin{align}
    R^\Gcomp_\tau(P,P^0,\mathcal{E},g,\pi) &:= \Psi^\Gcomp_\tau(P) - \Psi^\Gcomp_\tau(P^0) + P^0 D^\Gcomp_\tau(P,\mathcal{E},g,\gamma_P,\pi) \label{rt}\\
    \begin{split}
    &= -\frac{\gamma_0}{1-\gamma_P} P^0_{X \mid 1} \left\{ \left( \frac{1-g}{g} - \frac{1-g_0}{g} \right) (\mathcal{E}-\mathcal{E}_0) \right\} \\
    &\quad- \frac{\gamma_P-\gamma_0}{1-\gamma_P} \{ \Psi_\tau^\Gcomp(P) - \Psi_\tau^\Gcomp(P) \}.
    \end{split} \nonumber
\end{align}
Intuitively, this is the remainder from the linear approximation of $\Psi^\Gcomp_\tau(P^0)$ at $P$ (see \eqref{eq: intuitive expansion} and Figure~\ref{fig: one step}).

\begin{proof}[Proof of Theorem~\ref{thm: CV one-step efficiency}]
We first study the behavior of the estimator $\hat{\psi}_{n,\tau}^v$ \eqref{psi-n-v} for a fixed fold $v \in [V]$.

For each fold $v$, we first prove a consistency result of the plug-in estimator $\Psi^\Gcomp_\tau(\hat{P}_{\tau}^{n,v})$. By the definition of $\Psi^\Gcomp_\tau(\hat{P}_{\tau}^{n,v})$ and consistency of $\hat{\mathcal{E}}_{n,\tau}^{-v}$ in Condition~\ref{cond: sufficient nuisance rate}, it holds that
\begin{align*}
    \Psi^\Gcomp_\tau(\hat{P}_{\tau}^{n,v}) - \Psi^\Gcomp_\tau(P^0) &= P^{n,v}_{X \mid 0} \hat{\mathcal{E}}_{n,\tau}^{-v} - P^0_{X \mid 0} \mathcal{E}_{0,\tau} \\
    &= (P^{n,v}_{X \mid 0} - P^0_{X \mid 0}) \hat{\mathcal{E}}_{n,\tau}^{-v} + P^0_{X \mid 0} \hat{\mathcal{E}}_{n,\tau}^{-v} - \mathcal{E}_{0,\tau}\} \\
    &= \bigO_p \left( n^{-1/2} + \expect_{P^0} \sup_{v \in [V], \tau \in \mathcal{T}_n} \| \hat{\mathcal{E}}_{n,\tau}^{-v}-\mathcal{E}_{0,\tau} \|_{P^0_{X \mid 0},2} \right) \\
    &= \bigO_p \left( \expect_{P^0} \sup_{v \in [V], \tau \in \mathcal{T}_n} \| \hat{\mathcal{E}}_{n,\tau}^{-v}-\mathcal{E}_{0,\tau} \|_{P^0_{X \mid 0},2} \right) = \smallo_p(1),
\end{align*}
where the last line follows from the fact that we focus on nonparametric models and the convergence rate of $\hat{\mathcal{E}}_{n,\tau}^{-v}$ cannot be faster than $n^{-1/2}$.

By the definition of $R^\Gcomp_\tau$ and $D^\Gcomp_\tau$, as well as the definition of $\hat{P}_{\tau}^{n,v}$ from Section \ref{sec: CV one-step}, we have
\begin{align*}
    & \Psi^\Gcomp_\tau(\hat{P}_{\tau}^{n,v}) - \Psi^\Gcomp_\tau(P^0) \\
    &= -P^0 D^\Gcomp_\tau(\hat{P}_{\tau}^{n,v},\hat{\mathcal{E}}_{n,\tau}^{-v},\hat{g}_n^{-v},\hat{\gamma}_n^v,1) + R^\Gcomp_\tau(\hat{P}_{\tau}^{n,v},P^0,\hat{\mathcal{E}}_{n,\tau}^{-v},\hat{g}_n^{-v},\hat{\gamma}_n^v,1) \\
    &= (P^{n,v} - P^0) D_\tau(P^0,\mathcal{E}_{0,\tau},g_0,\gamma_0,1) - P^{n,v} D^\Gcomp_\tau(\hat{P}_{\tau}^{n,v},\hat{\mathcal{E}}_{n,\tau}^{-v},\hat{g}_n^{-v},\hat{\gamma}_n^v,1) \\
    &\quad + (P^{n,v} - P^0) [D^\Gcomp_\tau(\hat{P}_{\tau}^{n,v},\hat{\mathcal{E}}_{n,\tau}^{-v},\hat{g}_n^{-v},\hat{\gamma}_n^v,1) - D_\tau(P^0,\mathcal{E}_{0,\tau},g_0,\gamma_0,1)] \\
    &\quad+ R^\Gcomp_\tau(\hat{P}_{\tau}^{n,v},P^0,\hat{\mathcal{E}}_{n,\tau}^{-v},\hat{g}_n^{-v},\hat{\gamma}_n^v,1).
\end{align*}
Recall that $|I_v|\approx n/V$; thus $\hat{\gamma}_n^v$ is a root-$n$-consistent estimator of $\gamma_0$ and bounded away from zero (e.g., greater than $\gamma_0/2$) and from unity (e.g., smaller than $(1+\gamma_0)/2$) with probability tending to one. Therefore, $R^\Gcomp_\tau(\hat{P}_{\tau}^{n,v},P^0,\hat{\mathcal{E}}_{n,\tau}^{-v},\hat{g}_n^{-v},\hat{\gamma}_n^v,1) = \smallo_p(n^{-1/2})$ under Condition~\ref{cond: sufficient nuisance rate}.

We next show that, under Condition~\ref{cond: consistent bounded} and \ref{cond: sufficient nuisance rate}, 
\begin{equation}\label{dres}
(P^{n,v} - P^0) \left[D^\Gcomp_\tau(\hat{P}_{\tau}^{n,v},\hat{\mathcal{E}}_{n,\tau}^{-v},\hat{g}_n^{-v},\hat{\gamma}_n^v,1) - D_\tau(P^0,\mathcal{E}_{0,\tau},g_0,\gamma_0,1)\right] = \smallo_p(n^{-1/2}).
\end{equation}
Since $(\hat{\mathcal{E}}_{n,\tau}^{-v},\hat{g}_n^{-v})$ is independent of $P^{n,v}$, we first condition on $(\hat{\mathcal{E}}_{n,\tau}^{-v},\hat{g}_n^{-v})$. As mentioned above, with probability tending to one, $\hat{\gamma}_n^v$ is greater than $\gamma_0/2>0$ and smaller than $(1+\gamma_0)/2 < 1$. Thus, as per definition \eqref{eq: reparameterize weight}, $\mathscr{W}(\hat{g}_n^{-v},\hat{\gamma}_n^v)$ is bounded above by a constant. We will show at the end of this proof that, for any sequence $\{\delta_n\}_{n\ge 1}$ such that $\delta_n \rightarrow 0$ and $\sqrt{n} \delta_n \rightarrow \infty$ as $n\to\infty$,
with probability tending to one, the function class
\begin{align*}
    \funclass_{\delta_n} &:= \Big\{ o \mapsto \frac{a}{\gamma} \mathscr{W}(\hat{g}_n^{-v},\gamma)(x) \{Z_\tau(x,y)-\hat{\mathcal{E}}_{n,\tau}^{-v}(x)\} + \frac{1-a}{1-\gamma} [\hat{\mathcal{E}}_{n,\tau}^{-v}(x) - \Psi^\Gcomp_\tau(\tilde{P}_{\tau}^{n,v})] \\
    &\hspace{0.5in}- D_\tau(P^0,\mathcal{E}_{0,\tau},g_0,\gamma_0,1): \gamma \in [\gamma_0-\delta_n,\gamma_0+\delta_n] \Big\} \\
    &= \Bigg\{ o \mapsto \left[ \frac{a}{1-\gamma} \frac{1-\hat{g}_n^{-v}(x)}{\hat{g}_n^{-v}(x)} - \frac{a}{1-\gamma_0} \frac{1-g_0(x)}{g_0(x)} \right] Z_\tau(x,y) \\
    &\hspace{0.5in}- \left[ \frac{a}{1-\gamma} \frac{1-\hat{g}_n^{-v}(x)}{\hat{g}_n^{-v}(x)} \hat{\mathcal{E}}_{n,\tau}^{-v}(x) - \frac{a}{1-\gamma_0} \frac{1-g_0(x)}{g_0(x)} \mathcal{E}_{0,\tau}(x) \right] \\
    &\hspace{0.5in}+ \left[ \frac{1-a}{1-\gamma} \hat{\mathcal{E}}_{n,\tau}^{-v}(x) - \frac{1-a}{1-\gamma_0} \mathcal{E}_{0,\tau}(x) \right] \\
    &\hspace{0.5in}- \left[ \frac{1-a}{1-\gamma} \Psi^\Gcomp_\tau(\hat{P}_{\tau}^{n,v}) - \frac{1-a}{1-\gamma_0} \Psi_\tau(P^0) \right]: \gamma \in [\gamma_0-\delta_n,\gamma_0+\delta_n] \Bigg\}
\end{align*}
is a BUEI class containing $D^\Gcomp_\tau(\tilde{P}_{\tau}^{n,v},\mathcal{E}_{n,\tau}^{-v},g_n^{-v},\gamma_n^v,1) - D_\tau(P^0,\mathcal{E}_{0,\tau},g_0,\gamma_0,1)$. With $B$ in Conditions~\ref{cond: bounded weight} and \ref{cond: consistent bounded}, an envelope of $\funclass_{\delta_n}$ is
\begin{align*}
    o &\mapsto \frac{a}{1-\gamma_0-\delta_n} \left| \frac{1-\hat{g}_n^{-v}}{\hat{g}_n^{-v}}(x) - \frac{1-g_0(x)}{g_0(x)} \right| + B \left| \frac{a}{\gamma_0-\delta_n} - \frac{a}{\gamma_0} \right| \\
    &\quad+ \frac{a}{1-\gamma_0-\delta_n} \left| \frac{1-\hat{g}_n^{-v}}{\hat{g}_n^{-v}}(x) - \frac{1-g_0(x)}{g_0(x)} \right| + \frac{a}{\gamma_0-\delta_n} B |\hat{\mathcal{E}}_{n,\tau}^{-v}(x) - \mathcal{E}_{0,\tau}(x)| + \left| \frac{a}{\gamma_0-\delta_n} - \frac{a}{\gamma_0} \right| B \\
    &\quad+ \frac{1-a}{1-\gamma_0-\delta_n} |\hat{\mathcal{E}}_{n,\tau}^{-v}(x) - \mathcal{E}_{0,\tau}(x)| + \left| \frac{1-a}{1-\gamma_0-\delta_n} - \frac{1-a}{1-\gamma_0} \right| \\
    &\quad+\frac{1-a}{1-\gamma_0-\delta_n} |\Psi^\Gcomp_\tau(\hat{P}_{\tau}^{n,v}) - \Psi_\tau(P^0)| + \left| \frac{1-a}{1-\gamma_0-\delta_n} - \frac{1-a}{1-\gamma_0} \right|.
\end{align*}
The $L^2(P^0)$-norm of this envelope is of order
\begin{equation}\label{ubd0}
A_n:=\expect_{P^0} \left\{ \| \hat{\mathcal{E}}_{n,\tau}^{-v} - \mathcal{E}_{0,\tau} \|_{P^0_{X \mid 1},2} + \left\| \frac{1-\hat{g}_n^{-v}}{\hat{g}_n^{-v}} - \frac{1-g_0}{g_0} \right\|_{P^0_{X \mid 1},2} \right\} + \delta_n.
\end{equation}
Therefore, by the proof of Theorem~2.5.2 in \protect\citetsupp{vandervaart1996} (specifically Line~17 on page~128, the fourth displayed equation on that page), conditionally on $(\hat{\mathcal{E}}_{n,\tau}^{-v},\hat{g}_n^{-v})$ and the event that $D^\Gcomp_\tau(\hat{P}_{\tau}^{n,v},\hat{\mathcal{E}}_{n,\tau}^{-v},\hat{g}_n^{-v},\hat{\gamma}_n^v,1) - D_\tau(P^0,\mathcal{E}_{0,\tau},g_0,\gamma_0,1)$ falls in $\funclass_{\delta_n}$, we have that 
$$\expect_{P^0} \sqrt{n} \left| (P^{n,v} - P^0) [D^\Gcomp_\tau(\hat{P}_{\tau}^{n,v},\hat{\mathcal{E}}_{n,\tau}^{-v},\hat{g}_n^{-v},\hat{\gamma}_n^v,1) - D_\tau(P^0,\mathcal{E}_{0,\tau},g_0,\gamma_0,1)] \right|$$
is upper bounded, up to a multiplicative constant, by the $L^2(P^0)$-norm of an envelope of $\funclass_{\delta_n}$, because the first term in the expression in the cited bound, the uniform entropy integral, is bounded as $\funclass_{\delta_n}$ is of BUEI. By \eqref{ubd0},
\begin{align*}
    & \expect_{P^0} \sqrt{n} \left| (P^{n,v} - P^0) [D^\Gcomp_\tau(\hat{P}_{\tau}^{n,v},\hat{\mathcal{E}}_{n,\tau}^{-v},\hat{g}_n^{-v},\hat{\gamma}_n^v,1) - D_\tau(P^0,\mathcal{E}_{0,\tau},g_0,\gamma_0,1)] \right|
    \lesssim A_n,
\end{align*}
where we recall the definition of $\lesssim$ in Section~\ref{sec: identification}. Since $A_n$ tends to zero unconditionally under Condition~\ref{cond: sufficient nuisance rate} and $D^\Gcomp_\tau(\hat{P}_{\tau}^{n,v},\hat{\mathcal{E}}_{n,\tau}^{-v},\hat{g}_n^{-v},\hat{\gamma}_n^v,1) - D_\tau(P^0,\mathcal{E}_{0,\tau},g_0,\gamma_0,1)$ falls in $\funclass_{\delta_n}$ with probability tending to one, the claim \eqref{dres} holds by Lemma~6.1 in \protect\cite{Chernozhukov2018debiasedML}.

By the definition of $\hat{\psi}_{n,\tau}^v$ in \eqref{psi-n-v}, the above results taken together imply that
$$\hat{\psi}_{n,\tau}^v - \Psi_\tau(P^0) = (P^{n,v} - P^0) D_\tau(P^0,\mathcal{E}_{0,\tau},g_0,\gamma_0,1) + \smallo_p(n^{-1/2}).$$
By the definition of $\hat{\psi}_{n,\tau}$, taking a sum of the above equations over $v\in[V]$ weighted by fold sizes, it thus follows that
$$\hat{\psi}_{n,\tau} - \Psi_\tau(P^0) = (P^n - P^0) D_\tau(P^0,\mathcal{E}_{0,\tau},g_0,\gamma_0,1) + \smallo_p(n^{-1/2}).$$
Therefore, $\hat{\psi}_{n,\tau}$ has the canonical gradient $D_\tau(P^0,\mathcal{E}_{0,\tau},g_0,\gamma_0,1)$ as its influence function and thus is an asymptotically efficient estimator of $\Psi_\tau(P^0)$ \protect\citepsupp[see, e.g., Chapter~11 in][]{Pfanzagl1990}.
The claimed uniform convergence follows because the above arguments apply uniformly over $\tau \in \mathcal{T}_n$ so that the $\smallo_p(n^{-1/2})$ term is also uniform over $\tau \in \mathcal{T}_n$ by Condition~\ref{cond: sufficient nuisance rate}.

Now, as mentioned before, we prove the claim that $\funclass_{\delta_n}$ is VC-subgraph containing 
$$D^\Gcomp_\tau(\hat{P}_{\tau}^{n,v},\hat{\mathcal{E}}_{n,\tau}^{-v},\hat{g}_n^{-v},\hat{\gamma}_n^v,1) - D_\tau(P^0,\mathcal{E}_{0,\tau},g_0,\gamma_0,1).$$
We condition on the event that $(1-\hat{g}_n^{-v})/\hat{g}_n^{-v}$ is upper bounded by a constant so that $\Psi^\Gcomp(\hat{P}_{\tau}^{n,v})$ is well defined, which has probability tending to one. By Lemma~9.6 in \protect\citetsupp{Kosorok2008}, the function classes $\{o \mapsto \gamma: \gamma \in [\gamma_0/2,(1+\gamma_0)/2]\}$ and $\{o \mapsto 1-\gamma: \gamma \in [\gamma_0/2,(1+\gamma_0)/2]\}$ are VC-subgraph. By Part~(viii) of Lemma~9.9 in \protect\citetsupp{Kosorok2008}, taking the monotone function $\phi$ to be $x \mapsto 1/x$, the function classes $\{o \mapsto \gamma: 1/\gamma \in [\gamma_0/2,(1+\gamma_0)/2]\}$ and $\{o \mapsto 1/(1-\gamma): \gamma \in [\gamma_0/2,(1+\gamma_0)/2]\}$ are VC-subgraph. Further, by Part~(vi) of Lemma~9.9 in \protect\citetsupp{Kosorok2008}, taking the fixed function $g$ to be $o \mapsto a \frac{1-\hat{g}_n^{-v}(x)}{\hat{g}_n^{-v}(x)} \{Z_\tau - \hat{\mathcal{E}}_{n,\tau}^{-v}(x)\}$ and $o \mapsto (1-a) [\hat{\mathcal{E}}_{n,\tau}^{-v} - \Psi^\Gcomp(\hat{P}_{\tau}^{n,v})]$ respectively, we can show that the function classes $\{o \mapsto \frac{a}{\gamma} \mathscr{W}(\hat{g}_n^{-v},\gamma)(x) [Z_\tau - \hat{\mathcal{E}}_{n,\tau}^{-v}(x)]: \gamma \in [\gamma_0/2,(1+\gamma_0)/2]\}$ and $\{o \mapsto \frac{1-a}{1-\gamma} [\hat{\mathcal{E}}_{n,\tau}^{-v} - \Psi^\Gcomp(\hat{P}_{\tau}^{n,v})]: \gamma \in [\gamma_0/2,(1+\gamma_0)/2]\}$ are VC-subgraph. The function class with a single element $\{- D_\tau(P^0,\mathcal{E}_{0,\tau},g_0,\gamma_0,1)\}$ is also VC-subgraph. By Part~(iii) of Lemma~9.14 in \protect\citetsupp{Kosorok2008}, we have that the function class
\begin{align*}
    \Big\{ o &\mapsto \frac{a}{\gamma} \mathscr{W}(\hat{g}_n^{-v},\gamma)(x) \{Z_\tau(x,y)-\hat{\mathcal{E}}_{n,\tau}^{-v}(x)\} + \frac{1-a}{1-\gamma} [\hat{\mathcal{E}}_{n,\tau}^{-v}(x) - \Psi^\Gcomp(\hat{P}_{\tau}^{n,v})] \\
    &\hspace{0.5in}- D_\tau(P^0,\mathcal{E}_{0,\tau},g_0,\gamma_0,1): \gamma \in [\gamma_0/2,(1+\gamma_0)/2] \Big\}
\end{align*}
is VC-subgraph and thus BUEI.
As $\delta_n$ tends to zero, $\funclass_{\delta_n}$ is eventually a subclass of the above function class and thus is eventually BUEI. Since $\hat{\gamma}_n^v=\gamma_0+\bigO_p(n^{-1/2})$, we conclude that, with probability tending to one, $\funclass_{\delta_n}$ contains $D^\Gcomp_\tau(\hat{P}_{\tau}^{n,v},\hat{\mathcal{E}}_{n,\tau}^{-v},\hat{g}_n^{-v},\hat{\gamma}_n^v,1) - D_\tau(P^0,\mathcal{E}_{0,\tau},g_0,\gamma_0,1)$, as claimed.
\end{proof}

\subsection{Coverage of Wald CUB based on cross-fit one-step estimator \& APAC guarantee (Theorems~\ref{thm: general CI coverage}--\ref{thm: threshold selection based on CUB})} \label{sec: proof Wald CI coverage}

The proof of Theorem~\ref{thm: general CI coverage} relies on the following lemma. 

\begin{lemma} \label{lemma: bound for tail prob diff}
Let $A$ and $B$ be two random variables, $t$ be any scalar and $\eta$ be any positive number. It holds that $|\Prob(A > t) - \Prob(B > t)| \leq \Prob(B \in [t-\eta,t+\eta)) + \Prob(|A-B|>\eta)$.
\end{lemma}
\begin{proof}[Proof of Lemma~\ref{lemma: bound for tail prob diff}]
We have that
\begin{align*}
    \Prob(A > t) &= \Prob(A > t,B \geq t-\eta) + \Prob(A > t, B < t-\eta) \\
    &\leq \Prob(B \geq t-\eta) + \Prob(A - B > t - B, t-B > \eta) \\
    &\leq \Prob(B \geq t-\eta) + \Prob(A-B > \eta) 
    \leq \Prob(B \geq t-\eta) + \Prob(|A-B| > \eta).
\end{align*}
Note that (i) $A$ and $B$ are arbitrary and their roles can be switched, and (ii) $t$ is arbitrary. If we switch the roles of $A$ and $B$ and replace $t$ by $t+\eta$ in the above inequality, we obtain that $\Prob(B > t+\eta) \leq \Prob(A \geq t) + \Prob(|A-B| > \eta)$. Therefore, both of the following inequalities hold:
\begin{align*}
    \Prob(A > t) - \Prob(B > t) &\leq \Prob(B \in [t-\eta,t)) + \Prob(|A-B| > \eta) \\
    &\leq \Prob(B \in [t-\eta,t+\eta)) + \Prob(|A-B| > \eta), \\
    \Prob(A > t) - \Prob(B > t) &\geq -\Prob(B \in [t,t+\eta)) - \Prob(|A-B| \geq \eta) \\
    &\geq -\Prob(B \in [t-\eta,t+\eta)) - \Prob(|A-B| > \eta).
\end{align*}
The desired inequality follows.
\end{proof}

\begin{proof}[Proof of Theorem~\ref{thm: general CI coverage}]
Let $\mathcal{Z}$ be a random variable distributed as $N(0,\sigma_0^2 )$. By the triangle inequality,
\begin{align*}
    & |\Prob_{P^0}(\phi_0 < \hat{\phi}_n + z_\alpha \hat{\sigma}_n/\sqrt{n}) - (1-\alpha)| 
    = |\Prob_{P^0}(n^{1/2} (\hat{\phi}_n - \phi_0) > -z_\alpha \hat{\sigma}_n) - (1-\alpha)| \\
    &\leq |\Prob_{P^0}(n^{1/2} (\hat{\phi}_n - \phi_0) > -z_\alpha \hat{\sigma}_n) - \Prob_{P^0}(n^{1/2} P^n \IF > -z_\alpha \hat{\sigma}_n)| \\
    &\quad+ |\Prob_{P^0}(n^{1/2} P^n \IF > -z_\alpha \hat{\sigma}_n) - \Prob_{P^0}(\mathcal{Z} > -z_\alpha \hat{\sigma}_n)|
    + |\Prob_{P^0}(\mathcal{Z} > -z_\alpha \hat{\sigma}_n) - (1-\alpha)|.
\end{align*}
We refer to the last three terms above as Terms~1--3 and study them separately.

\noindent\textbf{Term~2}: By the Beery-Esseen Theorem, Term~2 is bounded by $\const \rho_0 n^{-1/2}/\sigma_0^3$ for a universal constant $\const$.

\noindent\textbf{Term~3}: Let $\eta >0$ be any fixed number. Since $\mathcal{Z} \sim N(0,\sigma_0^2)$, we have that
$$\sup_{\sigma \in [\sigma_0-\eta,\sigma_0+\eta], \sigma \neq \sigma_0} \frac{|\Prob_{P^0}(\mathcal{Z} > -z_\alpha \sigma) - \Prob_{P^0}(\mathcal{Z} > -z_\alpha \sigma_0)|}{|\sigma - \sigma_0|} \leq \const,$$
where the constant $\const$ can be taken as any number larger than the maximum density of the asymptotic normal distribution, namely $\exp\{-1/(2 \sigma_0^2)\}/(\sqrt{2 \pi} \sigma_0)$, which is further bounded by $1/(\sqrt{2 \pi} \sigma_0)$; and thus this constant $\const$ is decreasing in $\sigma_0$, of order $\sigma_0^{-1}$. 
Since $\hat{\sigma}^2_n$ is consistent for $\sigma^2_0 > 0$, $|\hat{\sigma}_n - \sigma_0| \leq \eta$ with probability tending to one. Hence, conditional on this event and on $\hat{\sigma}_n$, Term~2 is bounded by $\const |\hat{\sigma}_n-\sigma_0|/\sigma_0$ for a universal constant $\const$. (We have explicitly stated the order $\sigma_0^{-1}$ and thus the constant is universal.) We marginalize over $\hat{\sigma}_n$ and thus find that Term~3 is bounded by $\const \expect_{P^0} [\ind(|\hat{\sigma}_n - \sigma_0| \leq \eta) |\hat{\sigma}_n-\sigma_0|]/\sigma_0 + \Prob_{P^0}(|\hat{\sigma}_n-\sigma_0| > \eta)$.

\noindent\textbf{Term~1}: We apply Lemma~\ref{lemma: bound for tail prob diff} with $A=n^{1/2}(\hat{\phi}_n-\phi_0)$, $B=n^{1/2} P^n \IF$ and $t=-z_\alpha \hat{\sigma}_n$ and see that Term~1 is bounded by
$$\Prob_{P^0}\left(n^{1/2} P^n \IF \in [-z_\alpha \hat{\sigma}_n-\eta,-z_\alpha \hat{\sigma}_n+\eta)\right) + \Prob_{P^0}(n^{1/2} |\hat{\phi}_n-\phi_0-P^n \IF| > \eta).$$
Since $n^{1/2} P^n \IF$ converges in distribution to a nondegenerate normal distribution, and $\hat{\sigma}^2_n$ is consistent for $\sigma^2_0 > 0$, by the Berry-Esseen Theorem and a similar argument as for Term~3, we have that
\begin{align*}
    \Prob_{P^0}\left(n^{1/2} P^n \IF \in [-z_\alpha \hat{\sigma}_n-\eta,-z_\alpha \hat{\sigma}_n+\eta)\right) \leq \frac{\const}{\sigma_0} \eta + \const \frac{\rho_0}{\sigma_0^3} n^{-1/2}
\end{align*}
for a universal constant $\const$. 
Thus Term~1 is bounded by
\begin{align*}
    & \frac{\const}{\sigma_0} \eta + \const \frac{\rho_0}{\sigma_0^3} n^{-1/2} + \Prob_{P^0}(n^{1/2} |\hat{\phi}_n-\phi_0-P^n \IF| > \eta)
    \leq \frac{\const}{\sigma_0} \eta + \frac{\expect_{P^0} n^{1/2} |\hat{\phi}_n-\phi_0-P^n \IF|}{\eta} + \const \frac{\rho_0}{\sigma_0^3} n^{-1/2}.
\end{align*}
We take $\eta$ to be $\const^{-1/2} \sigma_0^{1/2} n^{1/4} \left\{ \expect_{P^0} \left| \hat{\phi}_n-\phi_0-P^n \IF \right| \right\}^{1/2}$ and have that Term~1 is bounded by
$$\const \frac{n^{1/4}}{\sigma_0^{1/2}} \left\{ \expect_{P^0} \left| \hat{\phi}_n-\phi_0-P^n \IF \right| \right\}^{1/2} + \const \frac{\rho_0}{\sigma_0^3} n^{-1/2}.$$
for a universal constant $\const$.

Summing up the three bounds for the three terms, we have proved Theorem~\ref{thm: general CI coverage}.
\end{proof}

We next apply Theorem~\ref{thm: general CI coverage} to estimators $\hat{\psi}_{n,\tau}$ and $\hat{\sigma}_{n,\tau}$ to prove Theorem~\ref{thm: convergence rate of one-step Wald CI}.

\begin{lemma} \label{lemma: CV one-step empirical process and remainder bound}
Under the conditions of Theorem~\ref{thm: CV one-step efficiency}, it holds that
\begin{align*}
    & \sqrt{n} \expect_{P^0} \left| \hat{\psi}_{n,\tau}-\Psi_\tau(P^0)-P^n D_\tau(P^0,\mathcal{E}_{0,\tau},g_0,\gamma_0,1) \right| \\
    &\lesssim \max_{v \in [V]} \Big\{ \expect_{P^0} \| \hat{\mathcal{E}}_{n,\tau}^{-v} - \mathcal{E}_{0,\tau} \|_{P^0_{X \mid 1},2} + \expect_{P^0} \left\| \frac{1-\hat{g}_n^{-v}}{\hat{g}_n^{-v}} - \frac{1-g_0}{g_0} \right\|_{P^0_{X \mid 1},2} \\
    &\quad+ \sqrt{n} \expect_{P^0} \left| P^0_{X \mid 1} \left\{ \frac{1-\hat{g}_n^{-v}}{\hat{g}_n^{-v}} - \frac{1-g_0}{g_0} \right\} (\hat{\mathcal{E}}_{n,\tau}^{-v} - \mathcal{E}_{0,\tau}) \right| \Big\} + \bigO(n^{-1/2}).
\end{align*}
\end{lemma}
\begin{proof}[Proof of Lemma~\ref{lemma: CV one-step empirical process and remainder bound}]
From the proof of Theorem~\ref{thm: CV one-step efficiency}, we see that
\begin{align*}
    & \hat{\psi}_{n,\tau}^v - \Psi_\tau(P^0) - P^n D_\tau(P^0,\mathcal{E}_{0,\tau},g_0,\gamma_0,1) \\
    &= (P^{n,v} - P^0) [D^\Gcomp_\tau(\hat{P}_{\tau}^{n,v},\hat{\mathcal{E}}_{n,\tau}^{-v},\hat{g}_n^{-v},\hat{\gamma}_n^v,1) - D_\tau(P^0,\mathcal{E}_{0,\tau},g_0,\gamma_0,1)] \\
    &\quad+ R^\Gcomp_\tau(\hat{P}_{\tau}^{n,v},P^0,\hat{\mathcal{E}}_{n,\tau}^{-v},\hat{g}_n^{-v},\hat{\gamma}_n^v,1).
\end{align*}
For the first term on the right-hand side, since $\delta_n$ can converge to zero at any rate slower than $n^{-1/2}$ in equation \eqref{ubd0} in the proof of Theorem~\ref{thm: CV one-step efficiency}, we have that
\begin{align*}
    & \expect_{P^0} \sqrt{n} \left| (P^{n,v} - P^0)[ D^\Gcomp_\tau(\hat{P}_{\tau}^{n,v},\hat{\mathcal{E}}_{n,\tau}^{-v},\hat{g}_n^{-v}.\hat{\gamma}_n^v,1) - D_\tau(P^0,\mathcal{E}_{0,\tau},g_0,\gamma_0,1)] \right| \\
    &\lesssim \expect_{P^0} \left\{ \left\| \hat{\mathcal{E}}_{n,\tau}^{-v} - \mathcal{E}_{0,\tau} \right\|_{P^0_{X \mid 1},2}+ \left\| \frac{1-\hat{g}_n^{-v}}{\hat{g}_n^{-v}} - \frac{1-g_0}{g_0} \right\|_{P^0_{X \mid 1},2} \right\} + \bigO(n^{-1/2}).
\end{align*}
By the definition of $R^\Gcomp_\tau$ in \eqref{rt}, we see that, under Condition~\ref{cond: bounded weight},
\begin{align*}
    & \expect_{P^0} \left| R^\Gcomp_\tau(\hat{P}_{\tau}^{n,v},P^0,\hat{\mathcal{E}}_{n,\tau}^{-v},\hat{g}_n^{-v},\hat{\gamma}_n^v,1) \right| \\
    &\lesssim \expect_{P^0} \left| P^0_{X \mid 1} \left\{ \frac{1-\hat{g}_n^{-v}}{\hat{g}_n^{-v}} - \frac{1-g_0}{g_0} \right\} (\hat{\mathcal{E}}_{n,\tau}^{-v} - \mathcal{E}_{0,\tau}) \right|.
\end{align*}
The desired result follows from the above two bounds.
\end{proof}

Further, we bound the convergence rate of the estimator $\sigma_{n,\tau}$ of the standard error.

\begin{lemma} \label{lemma: CV one-step var convergence bound}
Recall the probability $q_n$ of having a bounded nuisance estimator in Condition~\ref{cond: consistent bounded}. Under the conditions of Theorem~\ref{thm: CV one-step efficiency}, with probability $1-q_n-\bigO(\exp(-n))$ tending to one, it holds that $|\hat{\sigma}_{n,\tau} - \sigma_{0,\tau}| \lesssim \max_{v \in[V]} \{ \| \hat{\mathcal{E}}_{n,\tau}^{-v} - \mathcal{E}_{0,\tau} \|_{P^0_{X \mid 1},2} + \| (1-\hat{g}_n^{-v})/\hat{g}_n^{-v} - (1-g_0)/g_0 \|_{P^0_{X \mid 1},2} \} + \bigO_p(n^{-1/2})$.
\end{lemma}
It follows that for any $\eta>0$,
\begin{align*}
    &\expect_{P^0}[\ind(|\hat{\sigma}_{n,\tau}-\sigma_{0,\tau}| \leq \eta) |\hat{\sigma}_{n,\tau}-\sigma_{0,\tau}|] \\
    &\lesssim \expect_{P^0} \max_{v \in[V]} \{ \| \hat{\mathcal{E}}_{n,\tau}^{-v} - \mathcal{E}_{0,\tau} \|_{P^0_{X \mid 1},2} + \| (1-\hat{g}_n^{-v})/\hat{g}_n^{-v} - (1-g_0)/g_0 \|_{P^0_{X \mid 1},2} \} + \bigO(n^{-1/2}) \\
    &\leq \expect_{P^0} \sum_{v \in [V]} \{ \| \hat{\mathcal{E}}_{n,\tau}^{-v} - \mathcal{E}_{0,\tau} \|_{P^0_{X \mid 1},2} + \| (1-\hat{g}_n^{-v})/\hat{g}_n^{-v} - (1-g_0)/g_0 \|_{P^0_{X \mid 1},2} \} + \bigO(n^{-1/2}) \\
    &\lesssim \max_{v \in [V]} \{ \expect_{P^0} \| \hat{\mathcal{E}}_{n,\tau}^{-v} - \mathcal{E}_{0,\tau} \|_{P^0_{X \mid 1},2} + \expect_{P^0} \| (1-\hat{g}_n^{-v})/\hat{g}_n^{-v} - (1-g_0)/g_0 \|_{P^0_{X \mid 1},2} \} + \bigO(n^{-1/2}).
\end{align*}
In addition, with $q_n$ in Condition~\ref{cond: consistent bounded}
\begin{align*}
    &\Prob_{P^0}(|\hat{\sigma}_{n,\tau}-\sigma_{0,\tau}| > \eta) \\
    &\leq \Prob_{P^0} \left( \const \max_{v \in[V]} \{ \| \hat{\mathcal{E}}_{n,\tau}^{-v} - \mathcal{E}_{0,\tau} \|_{P^0_{X \mid 1},2} + \| (1-\hat{g}_n^{-v})/\hat{g}_n^{-v} - (1-g_0)/g_0 \|_{P^0_{X \mid 1},2} \} + \bigO_p(n^{-1/2}) > \eta \right)\\
    &\quad+ q_n + \bigO(\exp(-n)) \\
    &\leq \Prob_{P^0} \left( \const \max_{v \in [V]} \{ \| \hat{\mathcal{E}}_{n,\tau}^{-v} - \mathcal{E}_{0,\tau} \|_{P^0_{X \mid 1},2} +  \| (1-\hat{g}_n^{-v})/\hat{g}_n^{-v} - (1-g_0)/g_0 \|_{P^0_{X \mid 1},2} \} > \eta/2 \right) \\
    &\quad+ \Prob_{P^0} \left( \bigO_p(n^{-1/2}) > \eta/2 \right) + q_n + \bigO(\exp(-n)) \\
    &\leq \const \max_{v \in [V]} \{ \expect_{P^0} \| \hat{\mathcal{E}}_{n,\tau}^{-v} - \mathcal{E}_{0,\tau} \|_{P^0_{X \mid 1},2} + \expect_{P^0} \| (1-\hat{g}_n^{-v})/\hat{g}_n^{-v} - (1-g_0)/g_0 \|_{P^0_{X \mid 1},2} \} + q_n + \bigO(n^{-1/2}).
\end{align*}
 
\begin{proof}[Proof of Lemma~\ref{lemma: CV one-step var convergence bound}]
We can write
\begin{align*}
    & (\hat{\sigma}_{n,\tau}^v)^2 - \sigma_{0,\tau}^2
    = P^{n,v} D^\Gcomp_\tau(\hat{P}_{\tau}^{n,v},\hat{\mathcal{E}}_{n,\tau}^{-v},\hat{g}_n^{-v},\hat{\gamma}_n^v,1)^2 - P^0 D_\tau(P^0,\mathcal{E}_{0,\tau},g_0,\gamma_0,1)^2 \\
    &= (P^{n,v} - P^0) D^\Gcomp_\tau(\hat{P}_{\tau}^{n,v},\hat{\mathcal{E}}_{n,\tau}^{-v},\hat{g}_n^{-v},\hat{\gamma}_n^v,1)^2 \\
    &\quad+ P^0 \left\{ D^\Gcomp_\tau(\hat{P}_{\tau}^{n,v},\hat{\mathcal{E}}_{n,\tau}^{-v},\hat{g}_n^{-v},\hat{\gamma}_n^v,1) - D_\tau(P^0,\mathcal{E}_{0,\tau},g_0,\gamma_0,1) \right\} \\
    &\quad\quad\times \left\{ D^\Gcomp_\tau(\hat{P}_{\tau}^{n,v},\hat{\mathcal{E}}_{n,\tau}^{-v},\hat{g}_n^{-v},\hat{\gamma}_n^v,1) + D_\tau(P^0,\mathcal{E}_{0,\tau},g_0,\gamma_0,1) \right\}.
\end{align*}
Since $(\hat{\mathcal{E}}_{n,\tau}^{-v},\hat{g}_n^{-v},\hat{\gamma}_n^v)$ is independent of $P^{n,-v}$, we first condition on $(\hat{\mathcal{E}}_{n,\tau}^{-v},\hat{g}_n^{-v},\hat{\gamma}_n^v)$. 
Further, we condition on the event that $\hat{g}_n^{-v}$ and $\hat{\gamma}_n^v$ are bounded away from zero. The first event has probability $1-q_n$ tending to one. In addition, since $|I_v| \hat{\gamma}_n^v$ is distributed as $\mathrm{Binom}(|I_v|, \gamma_0)$, by Theorem~4 in \protect\citetsupp{Chung2006}, we have that, for any fixed constant $C \in (0,\gamma_0)$, $\Prob_{P^0}(\hat{\gamma}_n^v \leq C) = \Prob_{P^0}(|I_v| \hat{\gamma}_n^v \leq |I_v| \gamma_0 - |I_v| (\gamma_0-C)) \leq \exp \{ -|I_v|^2 (\gamma_0-C)^2/(2 |I_v| \gamma_0) \} = \bigO(\exp(-n))$. Therefore, the event that both $\hat{g}_n^{-v}$ and $\gamma_n^v$ are bounded away from zero has probability at least $1-q_n-\bigO(\exp(-n))$.

Conditional on this event, similarly to the proof of Theorem~\ref{thm: CV one-step efficiency}, we can show that
$$D^\Gcomp_\tau(\hat{P}_{\tau}^{n,v},\hat{\mathcal{E}}_{n,\tau}^{-v},\hat{g}_n^{-v},\hat{\gamma}_n^v,1)^2$$
falls in a fixed BUEI (and thus $P^0$-Donsker) class 
and hence the first term on the right-hand side is $\bigO_p(n^{-1/2})$. Since $D^\Gcomp_\tau(\hat{P}_{\tau}^{n,v},\hat{\mathcal{E}}_{n,\tau}^{-v},\hat{g}_n^{-v},\hat{\gamma}_n^v,1) + D_\tau(P^0,\mathcal{E}_{0,\tau},g_0,\gamma_0,1)$ is bounded above by an absolute constant with probability tending to one, under Condition~\ref{cond: bounded weight}, we have that the second term on the right-hand side is bounded by
$$\const \left\{ \| \hat{\mathcal{E}}_{n,\tau}^{-v} - \mathcal{E}_{0,\tau} \|_{P^0_{X \mid 1},2} + \| (1-\hat{g}_n^{-v})/\hat{g}_n^{-v} - (1-g_0)/g_0 \|_{P^0_{X \mid 1},2}  \right\} + \bigO_p(n^{-1/2}).$$
Therefore,
$$|(\hat{\sigma}_{n,\tau}^v)^2 - \sigma_{0,\tau}^2| \lesssim \| \hat{\mathcal{E}}_{n,\tau}^{-v} - \mathcal{E}_{0,\tau} \|_{P^0_{X \mid 1},2} + \| (1-\hat{g}_n^{-v})/\hat{g}_n^{-v} - (1-g_0)/g_0 \|_{P^0_{X \mid 1},2} + \bigO_p(n^{-1/2}).$$
By the definition of $\hat{\sigma}^2_{n,\tau}$, we have that
$$|\hat{\sigma}_{n,\tau}^2 - \sigma_{0,\tau}^2| \lesssim \max_{v \in [V]} \left\{ \| \hat{\mathcal{E}}_{n,\tau}^{-v} - \mathcal{E}_{0,\tau} \|_{P^0_{X \mid 1},2} + \| (1-\hat{g}_n^{-v})/\hat{g}_n^{-v} - (1-g_0)/g_0 \|_{P^0_{X \mid 1},2} \right\} + \bigO_p(n^{-1/2}).$$
The desired result follows by noting that, by a first-order Taylor expansion of the square-root function around $\sigma_{0,\tau}^2$,
$$\hat{\sigma}_{n,\tau} - \sigma_{0,\tau} = \frac{1}{\sigma_{0,\tau}} (\hat{\sigma}_{n,\tau}^2 - \sigma_{0,\tau}^2) + \smallo(|\hat{\sigma}_{n,\tau}^2 - \sigma_{0,\tau}^2|).$$
\end{proof}

This leads to the proof of Theorem~\ref{thm: convergence rate of one-step Wald CI}.
\begin{proof}[Proof of Theorem~\ref{thm: convergence rate of one-step Wald CI}]
The result for $\tau \in \mathcal{T}^\epsilon$ follows by applying the results of Lemmas~\ref{lemma: CV one-step empirical process and remainder bound}--\ref{lemma: CV one-step var convergence bound} to Theorem~\ref{thm: general CI coverage} with $\hat{\phi}_n = \hat{\psi}_{n,\tau}$, $\phi_0 = \Psi_\tau(P^0)$, $\hat{\sigma}_n = \hat{\sigma}_{n,\tau}$ and $\sigma_{0,\tau}
= \expect_{P^0}[P^0 D_\tau(P^0,\mathcal{E}_{0,\tau},g_0,\gamma_0,1)^2]^{1/2}$. Indeed, the required third-moment condition
$$\sup_{\tau \in \bar{\real}} P^0 |D_\tau(P^0,\mathcal{E}_{0,\tau},g_0,\gamma_0,1)|^3 < \infty$$
holds by Condition~\ref{cond: bounded weight}. All arguments hold uniformly for all $\tau \in \mathcal{T}^\epsilon$. We note that the product bias term of order
$$n^{1/4} \sup_{v \in [V], \tau \in \mathcal{T}_n} \left\{ \expect_{P^0} \left| \int \left( \frac{1-\hat{g}_n^{-v}(x)}{\hat{g}_n^{-v}(x)} - \frac{1-g_0(x)}{g_0(x)} \right) \cdot (\hat{\mathcal{E}}_{n,\tau}^{-v}(x) - \mathcal{E}_{0,\tau}(x)) P^0_{X \mid 1}(\intd x) \right| \right\}^{1/2}$$
dominates the other terms of order
$$\sup_{v \in [V], \tau \in \mathcal{T}_n} \left\{ \expect_{P^0} \left\| \frac{1-\hat{g}_n^{-v}}{\hat{g}_n^{-v}} - \frac{1-g_0}{g_0} \right\|_{P^0_{X \mid 1},2} + \expect_{P^0} \| \hat{\mathcal{E}}_{n,\tau}^{-v} - \mathcal{E}_{\infty,\tau} \|_{P^0_{X \mid 1},2} \right\}$$
because the nuisance functions cannot converge at a rate faster than the parametric rate $n^{-1/2}$ under a nonparametric model.

For $\tau \in \mathcal{T}^-$, under Condition~\ref{cond: constant Q for extreme tau}, it is not difficult to check that $\hat{\psi}_{n,\tau}=\Psi_\tau(P^0)$ and $\hat{\sigma}_{n,\tau}=0$, and hence the desired result follows.
\end{proof}

We next prove Theorem~\ref{thm: threshold selection based on CUB}, another building block of Corollary~\ref{corollary: CV one-step APAC}.

\begin{proof}[Proof of Theorem~\ref{thm: threshold selection based on CUB}]
We use a similar argument to the proof of Theorem~1 in \protect\citetsupp{Bates2021}. Recall the definition of $\tau^\dagger_n$ in Condition~\ref{cond: positive variance}, which is well defined because $\mathcal{T}_n$ is finite. Suppose that the event $\Psi_{\hat{\tau}_n}(P^0) > \alpha_\error$ occurs. This event implies that $\{\tau \in \mathcal{T}_n: \Psi_\tau(P^0) > \alpha_\error\} \neq \emptyset$ and thus $\tau^\dagger_n < \infty$. By monotonicity of $\tau \mapsto \Psi_\tau(P^0)$, we have that $\hat{\tau}_n \geq \tau^\dagger_n$. By the definition of $\hat{\tau}_n$ in \eqref{eq: general taun},
$$\lambda_n(\tau^\dagger_n) < \alpha_\error < \Psi_{\tau^\dagger_n}(P^0).$$
In other words, the CUB $\lambda_n(\tau^\dagger_n)$ does not contain the true coverage error $\Psi_{\tau^\dagger_n}(P^0)$. The probability of this event equals
\begin{align*}
    & 1-\Prob_{P^0} \left( \lambda_n(\tau^\dagger_n) \geq \Psi_{\tau^\dagger_n}(P^0) \right) 
    \leq 1 - \inf_{\tau \in \mathcal{T}_n} \Prob_{P^0} \left( \lambda_n(\tau) \geq \Psi_\tau(P^0) \right) \\
    &= \alpha_\conf + (1-\alpha_\conf) - \inf_{\tau \in \mathcal{T}_n} \Prob_{P^0} \left( \lambda_n(\tau) \geq \Psi_\tau(P^0) \right) \\
    &\leq \alpha_\conf + \sup_{\tau \in \mathcal{T}_n} \left| \Prob_{P^0}(\lambda_n(\tau) \geq \Psi_\tau(P^0)) - (1-\alpha_\conf) \right|.
\end{align*}
We have thus shown that
\begin{align*}
    \Prob_{P^0}(\Psi_{\hat{\tau}_n}(P^0) > \alpha_\error) &\leq 1 - \inf_{\tau \in \mathcal{T}_n} \Prob_{P^0} \left( \lambda_n(\tau) \geq \Psi_\tau(P^0) \right) \\
    &\leq \alpha_\conf + \sup_{\tau \in \mathcal{T}_n} \left| \Prob_{P^0}(\lambda_n(\tau) \geq \Psi_\tau(P^0)) - (1-\alpha_\conf) \right|.
\end{align*}
The desired result follows because the complement of the event that $\Psi_{\hat{\tau}_n}(P^0) > \alpha_\error$ is the event that $\Psi_{\hat{\tau}_n}(P^0) \leq \alpha_\error$, namely that $C_{\hat{\tau}_n}$ is approximately correct.
\end{proof}

Corollary~\ref{corollary: CV one-step APAC} follows immediately from Theorems~\ref{thm: convergence rate of one-step Wald CI} and \ref{thm: threshold selection based on CUB}.

\subsection{Properties of rejection sampling with invalid likelihood ratio (Theorem~\ref{thm: rejection sampling property})} \label{sec: proof rejection sampling property}

\begin{proof}[Proof of Theorem~\ref{thm: rejection sampling property}]
Let $\hat{w}_n:= \mathscr{W}(\hat{g}_n^\train,\hat{\gamma}_n^\train)$ for short. Due to sample splitting, the test data is an i.i.d. sample from $P^0$ conditional on $\hat{w}_n$. Since $\zeta_i \sim \mathrm{Unif}(0,1)$, we have that the probability of accepting observation $i$ conditional on $X_i$ and $A_i=1$ is $\hat{w}_n(X_i)/\hat{B}$. Therefore, the acceptance probability for an observation in the test data set from the source population is
\begin{align*}
    \Prob_{P^0}(\zeta_i \leq \hat{w}_n(X_i)/\hat{B} \mid A_i=1) &= \expect_{P^0}[ \Prob_{P^0}(\zeta_i \leq \hat{w}_n(X_i)/\hat{B} \mid X_i,A_i=1) \mid A_i=1] \\
    &= \frac{\expect_{P^0}[\hat{w}_n(X_i) \mid A_i=1]}{\hat{B}} = \frac{\Pi_{P^0}(\hat{w}_n)}{\hat{B}}
\end{align*}
and thus, the acceptance probability for an observation in the test data is
$$\Prob_{P^0}(A_i=1, \zeta_i \leq \hat{w}_n(X_i)) = \Prob_{P^0}(A_i=1) \Prob_{P^0}(\zeta_i \leq \hat{w}_n(X_i) \mid A_i=1) = \frac{\gamma_0 \Pi_{P^0}(\hat{w}_n)}{\hat{B}}.$$

To show that $\{(X_i,Y_i): i \in J_n\}$ is an i.i.d. sample drawn from $(X,Y) \mid A=0$ under $\breve{P}^n$, it suffices to show that $X_i$ is distributed as $X \mid A=0$ under $\breve{P}^n$ conditional on $i \in J_n$. Let $E$ be any measurable subset of $\mathcal{X}$. For all $i \in I_\test$, we have that
\begin{align*}
    &\Prob_{P^0}(X_i \in E, \zeta_i \leq \hat{w}_n(X_i)/\hat{B} \mid A_i=1) = \int_E \Prob_{P^0}(\zeta_i \leq \hat{w}_n(X_i)/\hat{B} \mid X_i=x,A_i=1) P^0_{X \mid 1}(\intd x) \\
    &= \int_E \frac{\hat{w}_n(x)}{\hat{B}} P^0_{X \mid 1}(\intd x) = \frac{\Pi_{P^0}(\hat{w}_n)}{\hat{B}} \int_E \frac{\hat{w}_n(x)}{\Pi_{P^0}(w_n)} P^0_{X \mid 1}(\intd x)
    = \frac{\Pi_{P^0}(\hat{w}_n)}{\hat{B}} \Prob_{\breve{P}^n}(X \in E \mid A=0).
\end{align*}
Therefore,
\begin{align*}
    &\Prob_{P^0}(X_i \in E \mid i \in J_n) = \Prob_{P^0}(X_i \in E \mid \zeta_i \leq \hat{w}_n(X_i)/\hat{B}, A_i=1) \\
    &= \frac{\Prob_{P^0}(X_i \in E, \zeta_i \leq \hat{w}_n(X_i)/\hat{B} \mid A_i=1)}{\Prob_{P^0}(\zeta_i \leq \hat{w}_n(X_i)/\hat{B} \mid A_i=1)} \\
    &= \left\{ \frac{\Pi_{P^0}(\hat{w}_n)}{\hat{B}} \Prob_{\breve{P}^n}(X \in E) \right\} \Bigg/ \left\{ \frac{\Pi_{P^0}(\hat{w}_n)}{\hat{B}} \right\} 
    = \Prob_{\breve{P}^n}(X \in E \mid A=0).
\end{align*}
In other words, $X_i \mid i \in J_n$ is distributed as $X \mid A=0$ under $\breve{P}^n$. The desired result follows.
\end{proof}

\subsection{One-step correction for rejection sampling (Theorems~\ref{thm: rejection sampling one-step correction}--\ref{thm: rejection sampling has larger variance})} \label{sec: proof one-step correction rejection sampling}

Similarly to $R^\Gcomp_\tau$, we define the following remainder for the weighted formula $\Psi^\weighted_\tau$:
\begin{align}
\begin{split}
    &R^\weighted_\tau(P,P^0,\mathcal{E},g,\gamma,\pi) := \Psi^\weighted_\tau(P) - \Psi^\weighted_\tau(P^0) + P^0 D^\weighted_\tau(P,\mathcal{E},g,\gamma,\pi) \\
    &= -\frac{\gamma_0}{\gamma_P} P^0_{X \mid 1} \left[ \frac{\mathscr{W}(g,\gamma)}{\Pi_P(\mathscr{W}(g,\gamma))} - \mathscr{W}(g_0,\gamma_0) \right] (\mathcal{E}-\mathcal{E}_{0,\tau}) \\
    &\quad+ \gamma_0 \left( \frac{1}{\gamma_P} - \frac{\Pi_P(\mathscr{W}(g,\gamma))}{\gamma \pi} \right) P^0_{X \mid 1} \left[ \frac{\mathscr{W}(g,\gamma)}{\Pi_P(\mathscr{W}(g,\gamma))} - \mathscr{W}(g_0,\gamma_0) \right] \mathcal{E} \\
    &\quad- \frac{\gamma_P-\gamma_0}{1-\gamma_P} \Psi^\weighted_\tau(P) - \frac{\gamma_P-\gamma_0}{\gamma_P} P^0_{X \mid 0} \mathcal{E}_{0,\tau} \\
    &\quad+ \frac{\gamma-\gamma_0}{\gamma_0 (1-\gamma_0 )} P^0_{X \mid 0} \mathcal{E}_{0,\tau} + \frac{\gamma-\gamma_0}{\gamma_0 (1-\gamma_0 )} P^0_{X \mid 0} (\mathcal{E} - \mathcal{E}_{0,\tau})
    \ (\gamma-\gamma_0) \left( \frac{1}{\gamma (1-\gamma)} - \frac{1}{\gamma_0 (1-\gamma_0)} \right) P^0_{X \mid 0} \mathcal{E} \\
    &\quad+ \frac{\gamma_0}{\gamma} \left( 1 - \frac{\Pi_P(\mathscr{W}(g,\gamma))}{\pi} \right) P^0_{X \mid 0} \mathcal{E}_{0,\tau} 
    + \frac{\gamma_0}{\gamma} \left( 1 - \frac{\Pi_P(\mathscr{W}(g,\gamma))}{\pi} \right) P^0_{X \mid 0} (\mathcal{E}-\mathcal{E}_{0,\tau}).
\end{split} \label{eq: RS remainder}
\end{align}

We first prove a lemma on the root-$n$-consistency of a one-step corrected oracle estimator  $A_{n,\tau}:=\Psi^\weighted_\tau(\breve{P}^n) + \frac{1}{|I_\test|} \sum_{i \in I_\test} \tilde{D}(\hat{\mathcal{E}}_{n,\tau}^\train,\hat{g}_n^\train,\hat{\gamma}_n^\train,\hat{\pi}_n)(O_i)$ of $\Psi_\tau(P^0)$. We call $A_{n,\tau}$ an oracle estimator because $\breve{P}^n$ defined at the beginning of Section \ref{rs_approx} involves unknown components of $P^0$ by definition, and thus so does $A_{n,\tau}$. We study this oracle estimator because the sample proportion $\sum_{i \in J_n} Z_\tau(X_i,Y_i)/|J_n|$ from rejection sampling is centered around $\Psi^\weighted_\tau(\breve{P}^n)$ conditional on $I_\train$, and thus $\breve{\psi}_{n,\tau}$ from \eqref{eq: RS one-step estimator} is centered around $A_{n,\tau}$.

\begin{lemma} \label{lemma: root-n-consistency of one-step oracle estimator}
Under Conditions~\ref{cond: positivity of P(A)}--\ref{cond: bounded weight}, \ref{cond: known weight bound} and \ref{cond: sufficient nuisance rate2}, with $\Gamma_{n,\tau}$ from Theorem \ref{thm: rejection sampling one-step correction} and
$$\tilde{\Gamma}_{n,\tau} := \Gamma_{n,\tau} - \frac{1}{|I_\test|} \sum_{i \in I_\test} \hat{B} \frac{A_i}{\gamma_0} \ind(\zeta_i \leq w_0(X_i)/\hat{B}) [Z_\tau(X_i,Y_i)-\Psi_\tau(P^0)],$$
it holds that
$$\sup_{\tau \in \mathcal{T}_n} \left|  A_{n,\tau} - \Psi_\tau(P^0) - \tilde{\Gamma}_{n,\tau} \right| = \smallo_p(n^{-1/2}).$$
\end{lemma}

\begin{proof}[Proof of Lemma~\ref{lemma: root-n-consistency of one-step oracle estimator}]
Let $P^{n,\train}$ and $P^{n,\test}$ denote the empirical distribution in the training and test data,  with index sets $I_\train$ and $I_\test$, respectively. Condition~\ref{cond: sufficient nuisance rate2} implies that
$$\expect_{P^0} |\Pi_{P^0}(\hat{w}_n) - 1|=\expect_{P^0} |\Pi_{P^0}(\hat{w}_n) - \Pi_{P^0}(\mathscr{W}(g_0,\gamma_0))| \leq \expect_{P^0} \left\| \mathscr{W}(\hat{g}_n^\train,\hat{\gamma}_n^\train) - \mathscr{W}(g_0,\gamma_0) \right\|_{P^0_{X \mid 0},2} = \smallo(1)$$
and hence the normalized likelihood ratio estimator is also consistent:
\begin{align*}
    \expect_{P^0} \left\| \frac{\hat{w}_n}{\Pi_{P^0}(\hat{w}_n)} - \mathscr{W}(g_0,\gamma_0) \right\|_{P^0_{X \mid 0},2} &\leq \expect_{P^0} \| \hat{w}_n - \mathscr{W}(g_0,\gamma_0) \|_{P^0_{X \mid 0},2} + \expect_{P^0} \left| \frac{1}{\Pi_{P^0}(\hat{w}_n)} - 1 \right| \| \hat{w}_n \|_{P^0_{X \mid 0},2} \\
    &= \smallo(1).
\end{align*}
By the definition of $R^\weighted_\tau$ in \eqref{eq: RS remainder} and by the definition of $\breve{P}^n$ from the beginning of Section \ref{rs_approx}, we have that
\begin{align}
    \Psi^\weighted_\tau(\breve{P}^n) - \Psi_\tau(P^0) &= -P^0 D^\weighted_\tau(\breve{P}^n,\hat{\mathcal{E}}_{n,\tau}^\train,\hat{g}_n^\train,\hat{\gamma}_n^\train,\hat{\pi}_n) + R^\weighted_\tau(\breve{P}^n,P^0,\hat{\mathcal{E}}_{n,\tau}^\train,\hat{g}_n^\train,\hat{\gamma}_n^\train,\hat{\pi}_n) \nonumber \\
    &= -P^0 \tilde{D}(\hat{\mathcal{E}}_{n,\tau}^\train,\hat{g}_n^\train,\hat{\gamma}_n^\train,\hat{\pi}_n) + R^\weighted_\tau(\breve{P}^n,P^0,\hat{\mathcal{E}}_{n,\tau}^\train,\hat{g}_n^\train,\hat{\gamma}_n^\train,\hat{\pi}_n) \nonumber \\
    \begin{split}
        &= (P^{n,\test} - P^0) \tilde{D}(\mathcal{E}_{0,\tau},g_0,\gamma_0,1) - P^{n,\test} \tilde{D}(\hat{\mathcal{E}}_{n,\tau}^\train,\hat{g}_n^\train,\hat{\gamma}_n^\train,\hat{\pi}_n) \\
        &\quad+ (P^{n,\test}-P^0) [\tilde{D}(\hat{\mathcal{E}}_{n,\tau}^\train,\hat{g}_n^\train,\hat{\gamma}_n^\train,\hat{\pi}_n) - \tilde{D}(\mathcal{E}_{0,\tau},g_0,\gamma_0,1)] \\
        &\quad+ R^\weighted_\tau(\breve{P}^n,P^0,\hat{\mathcal{E}}_{n,\tau}^\train,\hat{g}_n^\train,\hat{\gamma}_n^\train,\hat{\pi}_n).
    \end{split} \label{eq: RS estimator expansion}
\end{align}
We note that $(\hat{\mathcal{E}}_{n,\tau}^\train,\hat{g}_n^\train,\hat{\gamma}_n^\train)$ is independent of $P^{n,\test}$ and is consistent for $(\mathcal{E}_{0,\tau},g_0,\gamma_0)$ under Condition~\ref{cond: sufficient nuisance rate2}. By a similar argument to that in the proof of Theorem~\ref{thm: CV one-step efficiency}, we have that $(P^{n,\test}-P^0) [\tilde{D}(\hat{\mathcal{E}}_{n,\tau}^\train,\hat{g}_n^\train,\hat{\gamma}_n^\train,\hat{\pi}_n) - \tilde{D}(\mathcal{E}_{0,\tau},g_0,\gamma_0,1)] = \smallo_p(n^{-1/2})$ since, conditional on the training data, $\tilde{D}(\hat{\mathcal{E}}_{n,\tau}^\train,\hat{g}_n^\train,\hat{\gamma}_n^\train,\hat{\pi}_n) - \tilde{D}(\mathcal{E}_{0,\tau},g_0,\gamma_0,1)$ falls into a BUEI class with a bounded envelope with probability tending to one. 
Next, $\hat{\gamma}_n^\train - \gamma_0$ and $\hat{\pi}_n - \Pi_{P^0}(\hat{w}_n)$ are both of order $\bigO_p(n^{-1/2})$, and so, under Conditions~\ref{cond: bounded weight} and \ref{cond: sufficient nuisance rate2}, the following terms are all $\smallo_p(n^{-1/2})$:
\begin{align*}
    & P^0_{X \mid 1} \left[ \frac{\mathscr{W}(\hat{g}_n^\train,\hat{\gamma}_n^\train)}{\Pi_P(\mathscr{W}(\hat{g}_n^\train,\hat{\gamma}_n^\train))} - \mathscr{W}(g_0,\gamma_0) \right] (\hat{\mathcal{E}}_{n,\tau}^\train-\mathcal{E}_{0,\tau}), \\
    &\gamma_0 \left( \frac{1}{\gamma_0} - \frac{\Pi_P(\mathscr{W}(\hat{g}_n^\train,\hat{\gamma}_n^\train))}{\hat{\gamma}_n^\train \hat{\pi}_n} \right) P^0_{X \mid 1} \left[ \frac{\mathscr{W}(\hat{g}_n^\train,\hat{\gamma}_n^\train)}{\Pi_P(\mathscr{W}(\hat{g}_n^\train,\hat{\gamma}_n^\train))} - \mathscr{W}(g_0,\gamma_0) \right] \hat{\mathcal{E}}_{n,\tau}^\train, \\
    & \frac{\hat{\gamma}_n^\train-\gamma_0}{\gamma_0 (1-\gamma_0 )} P^0_{X \mid 0} (\hat{\mathcal{E}}_{n,\tau}^\train - \mathcal{E}_{0,\tau})
    \ (\hat{\gamma}_n^\train-\gamma_0) \left( \frac{1}{\hat{\gamma}_n^\train (1-\hat{\gamma}_n^\train)} - \frac{1}{\gamma_0 (1-\gamma_0)} \right) P^0_{X \mid 0} \hat{\mathcal{E}}_{n,\tau}^\train, \\
    & \frac{\gamma_0}{\hat{\gamma}_n^\train} \left( 1 - \frac{\Pi_P(\mathscr{W}(\hat{g}_n^\train,\hat{\gamma}_n^\train))}{\hat{\pi}_n} \right) P^0_{X \mid 0} (\hat{\mathcal{E}}_{n,\tau}^\train-\mathcal{E}_{0,\tau}).
\end{align*}
Further, by the definition of $\breve{P}^n$, $\gamma_{\breve{P}^n}=\gamma_0$. Plug these into the definition of $R^\weighted$ in \eqref{eq: RS remainder} to find that 
\begin{align*}
    & R^\weighted_\tau(\breve{P}^n,P^0,\hat{\mathcal{E}}_{n,\tau}^\train,\hat{g}_n^\train,\hat{\gamma}_n^\train,\hat{\pi}_n) \\
    &= \frac{\hat{\gamma}_n^\train-\gamma_0}{\gamma_0 (1-\gamma_0)} P^0_{X \mid 0} \mathcal{E}_{0,\tau} + \frac{\gamma_0}{\hat{\gamma}_n^\train} \left( 1 - \frac{\Pi_{P^0}(\hat{w}_n)}{\hat{\pi}_n} \right) P^0_{X \mid 0} \mathcal{E}_{0,\tau} + \smallo_p(n^{-1/2}) \\
    &= \frac{1}{|I_\train|} \sum_{i \in I_\train} \frac{A_i - \gamma_0}{\gamma_0 (1-\gamma_0)} \Psi_\tau(P^0) + \frac{1}{|I_\test|} \sum_{i \in I_\test} \frac{A_i [w_0(X_i) - 1]}{\gamma_0} \Psi_\tau(P^0) + \smallo_p(n^{-1/2}),
\end{align*}
where the last step follows because $\hat{\gamma}_n^\train=\frac{1}{|I_\train|} \sum_{i \in I_\train} A_i$ and from the definition of $\hat{\pi}_n$ in \eqref{eq: definition of pi_n},
since $\hat{w}_n$ is consistent for $\mathscr{W}(g_0,\gamma_0)$, and from the Delta-method for influence functions, Lemma \ref{iflemma}:
\begin{align*}
    &\frac{\gamma_0}{\hat{\gamma}_n^\train} \left( 1 - \frac{\Pi_{P^0}(\hat{w}_n)}{\hat{\pi}_n} \right) P^0_{X \mid 0} \mathcal{E}_{0,\tau} \\
    &= \frac{\gamma_0\Psi_\tau(P^0)}{\gamma_0 + \frac{1}{|I_\train|} \sum_{i \in I_\train} (A_i - \gamma_0)} \left( 1 - \Pi_{P^0}(\hat{w}_n) \frac{\gamma_0 + \frac{1}{|I_\test|} \sum_{i \in I_\test} (A_i - \gamma_0)}{\gamma_0 \Pi_{P^0}(\hat{w}_n) + \frac{1}{|I_\test|} \sum_{i \in I_\test} [A_i \hat{w}_n(X_i) - \gamma_0 \Pi_{P^0}(\hat{w}_n)]} \right) \\
    &= \Psi_\tau(P^0) \left( 1- \frac{1}{\gamma_0} \frac{1}{|I_\train|} \sum_{i \in I_\train} (A_i - \gamma_0) + \smallo_p(n^{-1/2}) \right) \Bigg\{ 1 - \Pi_{P^0}(\hat{w}_n) \Bigg[ \left( \gamma_0 + \frac{1}{|I_\test|} \sum_{i \in I_\test} (A_i - \gamma_0) \right) \\
    &\quad\times \left( \frac{1}{\gamma_0 \Pi_{P^0}(\hat{w}_n)} - \frac{1}{\gamma_0^2 \Pi_{P^0}(\hat{w}_n)^2} \frac{1}{|I_\test|} \sum_{i \in I_\test} [A_i \hat{w}_n(X_i) - \gamma_0 \Pi_{P^0}(\hat{w}_n)] + \smallo_p(n^{-1/2}) \right) \Bigg] \Bigg\} \\
    &= \Psi_\tau(P^0) \left( 1- \frac{1}{\gamma_0} \frac{1}{|I_\train|} \sum_{i \in I_\train} (A_i - \gamma_0) + \smallo_p(n^{-1/2}) \right) \\
    &\quad\times \left\{ 1 - \Pi_{P^0}(\hat{w}_n) \left[ \frac{1}{\Pi_{P^0}(\hat{w}_n)} + \frac{1}{|I_\test|} \sum_{i \in I_\test} \frac{A_i [\Pi_{P^0}(\hat{w}_n) - \hat{w}_n(X_i)]}{\gamma_0 \Pi_{P^0}(\hat{w}_n)} + \smallo_p(n^{-1/2}) \right] \right\}.
\end{align*}
This further equals
\begin{align*}
    &\Psi_\tau(P^0) \left(1 - \frac{\sum_{i \in I_\train} (A_i - \gamma_0)}{\gamma_0 |I_\train|} + \smallo_p(n^{-1/2}) \right) \left\{ \frac{1}{|I_\test|} \sum_{i \in I_\test} \frac{A_i [\hat{w}_n(X_i) - \Pi_{P^0}(\hat{w}_n)]}{\gamma_0 \Pi_{P^0}(\hat{w}_n)} + \smallo_p(n^{-1/2}) \right\} \\
    &= \frac{1}{|I_\test|} \sum_{i \in I_\test} \frac{A_i [\hat{w}_n(X_i) - \Pi_{P^0}(\hat{w}_n)]}{\gamma_0 \Pi_{P^0}(\hat{w}_n)} \Psi_\tau(P^0) + \smallo_p(n^{-1/2}) \\
    &= \frac{1}{|I_\test|} \sum_{i \in I_\test} \frac{A_i [\hat{w}_n(X_i) - 1]}{\gamma_0 \Pi_{P^0}(\hat{w}_n)} \Psi_\tau(P^0) - \frac{\Pi_{P^0}(\hat{w}_n) - 1}{\Pi_{P^0}(\hat{w}_n)} \frac{1}{|I_\test|} \sum_{i \in I_\test} \frac{A_i}{\gamma_0} \Psi_\tau(P^0) + \smallo_p(n^{-1/2}) \\
    &= \frac{1}{|I_\test|} \sum_{i \in I_\test} \frac{A_i [\hat{w}_n(X_i) - 1]}{\gamma_0 \Pi_{P^0}(\hat{w}_n)} \Psi_\tau(P^0) + \smallo_p(n^{-1/2}).
\end{align*}
In addition, by Condition~\ref{cond: bounded weight}, we have that the following three terms are $\bigO_p(n^{-1/2})$:
\begin{align*}
    &\sup_{\tau \in \mathcal{T}_n} \left| P^{n,\test} \tilde{D}(\hat{\mathcal{E}}_{n,\tau}^\train,\hat{g}_n^\train,\hat{\gamma}_n^\train,\hat{\pi}_n) \right|, \qquad
    \sup_{\tau \in \mathcal{T}_n} \left| \frac{1}{|I_\train|} \sum_{i \in I_\train} \frac{A_i - \gamma_0}{\gamma_0 (1-\gamma_0)} \Psi_\tau(P^0) \right|, \\
    &\qquad\qquad\qquad\sup_{\tau \in \mathcal{T}_n} \left| \frac{1}{|I_\test|} \sum_{i \in I_\test} \frac{A_i [w_0(X_i) - 1]}{\gamma_0} \Psi_\tau(P^0) \right|.
\end{align*}
We plug all the above results in \eqref{eq: RS estimator expansion} and have that
\begin{align*}
    &\Psi^\weighted_\tau(\breve{P}^n) - \Psi_\tau(P^0) \\
    &= (P^{n,\test} - P^0) \tilde{D}(\mathcal{E}_{0,\tau},g_0,\gamma_0,1) - P^{n,\test} \tilde{D}(\hat{\mathcal{E}}_{n,\tau}^\train,\hat{g}_n^\train,\hat{\gamma}_n^\train,\hat{\pi}_n) \\
    &\quad+ \smallo_p(n^{-1/2}) + \frac{1}{|I_\train|} \sum_{i \in I_\train} \frac{A_i - \gamma_0}{\gamma_0 (1-\gamma_0)} \Psi_\tau(P^0) + \frac{1}{|I_\test|} \sum_{i \in I_\test} \frac{A_i [w_0(X_i) - 1]}{\gamma_0} \Psi_\tau(P^0).
\end{align*}
Adding $P^{n,\test} \tilde{D}(\hat{\mathcal{E}}_{n,\tau}^\train,\hat{g}_n^\train,\hat{\gamma}_n^\train,\hat{\pi}_n)$ to both sides and we have that
$$A_{n,\tau} - \Psi_\tau(P^0) = \tilde{\Gamma}_{n,\tau} + \smallo_p(n^{-1/2}).$$
The desired results follow by noting that the above arguments apply uniformly over $\tau \in \mathcal{T}_n$.
\end{proof}

\begin{proof}[Proof of Theorem~\ref{thm: rejection sampling one-step correction}]
We first study the sample proportion
$$\frac{\sum_{i \in J_n} Z_\tau(X_i,Y_i)}{|J_n|} = \frac{\frac{1}{|I_\test|} \sum_{i \in I_\test} \ind(A_i=1, \zeta_i \leq \hat{w}_n(X_i) \leq \hat{B}) Z_\tau(X_i,Y_i)}{\frac{1}{|I_\test|} \sum_{i \in I_\test} \ind(A_i=1, \zeta_i \leq \hat{w}_n(X_i) \leq \hat{B}) }.$$
We condition on the event $\hat{w}_n \leq \hat{B}$, which has probability tending to one, throughout this proof. We first condition on the training data and hence also on $\hat{w}_n$. The numerator and the denominator of the above expression are both asymptotically linear:
\begin{align*}
    &\frac{1}{|I_\test|} \sum_{i \in I_\test} \ind(A_i=1, \zeta_i \leq \hat{w}_n(X_i) \leq \hat{B}) Z_\tau(X_i,Y_i) 
    = \frac{\Psi_\tau^\weighted(\breve{P}^n) \gamma_0 \Pi_{P^0}(\hat{w}_n)}{\hat{B}} \\
    &+ \frac{1}{|I_\test|} \sum_{i \in I_\test} \left\{ \ind(A_i=1, \zeta_i \leq \hat{w}_n(X_i) \leq \hat{B}) Z_\tau(X_i,Y_i) - \frac{\Psi_\tau^\weighted(\breve{P}^n) \gamma_0 \Pi_{P^0}(\hat{w}_n)}{\hat{B}} \right\}, \\
    & \frac{1}{|I_\test|} \sum_{i \in I_\test} \ind(A_i=1, \zeta_i \leq \hat{w}_n(X_i) \leq \hat{B}) \\
    &= \frac{\gamma_0 \Pi_{P^0}(\hat{w}_n)}{\hat{B}} + \frac{1}{|I_\test|} \sum_{i \in I_\test} \left\{ \ind(A_i=1, \zeta_i \leq \hat{w}_n(X_i) \leq \hat{B}) - \frac{\gamma_0 \Pi_{P^0}(\hat{w}_n)}{\hat{B}} \right\}.
\end{align*}
We apply the Delta-method, Lemma \ref{iflemma}, to $f: (a,b) \mapsto a/b$ with arguments
$$\left(\frac{1}{|I_\test|} \sum_{i \in I_\test} \ind(A_i=1, \zeta_i \leq \hat{w}_n(X_i) \leq \hat{B}) Z_\tau(X_i,Y_i),\frac{1}{|I_\test|} \sum_{i \in I_\test} \ind(A_i=1, \zeta_i \leq \hat{w}_n(X_i) \leq \hat{B})\right)$$
and obtain that $\sum_{i \in J_n} Z_\tau(X_i,Y_i)/|J_n|$ equals 
\begin{align*}
    & \left( \frac{\Psi_\tau^\weighted(\breve{P}^n) \gamma_0 \Pi_{P^0}(\hat{w}_n)}{\hat{B}} \right) \Bigg/ \left( \frac{\gamma_0 \Pi_{P^0}(\hat{w}_n)}{\hat{B}} \right) \\
    &\quad+ \frac{\hat{B}}{\gamma_0 \Pi_{P^0}(\hat{w}_n)} \frac{1}{|I_\test|} \sum_{i \in I_\test} \left\{ \ind(A_i=1, \zeta_i \leq \hat{w}_n(X_i) \leq \hat{B}) Z_\tau(X_i,Y_i) - \frac{\Psi_\tau^\weighted(\breve{P}^n) \gamma_0 \Pi_{P^0}(\hat{w}_n)}{\hat{B}} \right\} \\
    &\quad- \frac{\Psi_\tau^\weighted(\breve{P}^n) \gamma_0 \Pi_{P^0}(\hat{w}_n)}{\hat{B}} \left( \frac{\hat{B}}{\gamma_0 \Pi_{P^0}(\hat{w}_n)} \right)^2 \frac{1}{|I_\test|} \sum_{i \in I_\test} \left\{ \ind(A_i=1, \zeta_i \leq w_n(X_i) \leq \hat{B}) - \frac{\gamma_0 \Pi_{P^0}(\hat{w}_n)}{\hat{B}} \right\} \\
    &\quad+ \smallo_p(n^{-1/2}) \\
    &= \Psi^\weighted_\tau(\breve{P}^n) + \frac{1}{|I_\test|} \sum_{i \in I_\test} \hat{B} \frac{A_i}{\gamma_0 \Pi_{P^0}(\hat{w}_n)} \ind(\zeta_i \leq \hat{w}_n(X_i)/\hat{B}) [Z_\tau(X_i,Y_i)-\Psi^\weighted_\tau(\breve{P}^n)] + \smallo_p(n^{-1/2}).
\end{align*}
The $\smallo_p(n^{-1/2})$ term is uniform over $\tau \in \mathcal{T}_n$. We next consider the randomness in the training data. Under Condition~\ref{cond: sufficient nuisance rate2}, $\hat{w}_n$ is consistent for $w_0=\mathscr{W}(g_0,\gamma_0)$, and so $\Pi_{P^0}(\hat{w}_n)$ is consistent for unity, and $\Psi^\weighted_\tau(\breve{P}^n)$ is consistent for $\Psi_\tau(P^0)$. Consequently,
\begin{align}
    \begin{split}
        &\frac{\sum_{i \in J_n} Z_\tau(X_i,Y_i)}{|J_n|} \\
        &= \Psi^\weighted_\tau(\breve{P}^n) + \frac{1}{|I_\test|} \sum_{i \in I_\test} \hat{B} \frac{A_i}{\gamma_0} \ind(\zeta_i \leq w_0(X_i)/\hat{B}) [Z_\tau(X_i,Y_i)-\Psi_\tau(P^0)] + \smallo_p(n^{-1/2}).
    \end{split} \label{eq: RS sample mean expansion}
\end{align}
We add $P^{n,\test} \tilde{D}(\hat{\mathcal{E}}_{n,\tau}^\train,\hat{g}_n^\train,\hat{\gamma}_n^\train,\hat{\pi}_n)$ to both sides and apply Lemma~\ref{lemma: root-n-consistency of one-step oracle estimator} to obtain the desired result.
\end{proof}

The proof of Theorem~\ref{thm: convergence rate of rejection sample CI} is very similar to that of Theorem~\ref{thm: convergence rate of one-step Wald CI} and Corollary~\ref{corollary: CV one-step APAC}. This proof is an application of Theorems~\ref{thm: general CI coverage} and \ref{thm: threshold selection based on CUB}. For $\breve{\psi}_{n,\tau}$ from \eqref{eq: RS one-step estimator}, we can obtain the following equality by \eqref{eq: RS estimator expansion} and \eqref{eq: RS sample mean expansion}: 
\begin{align*}
    \breve{\psi}_{n,\tau} &= \Psi_\tau(P^0) + \frac{1}{|I_\test|} \sum_{i \in I_\test} \hat{B} \frac{A_i}{\gamma_0 \Pi_{P^0}(\hat{w}_n)} \ind(\zeta_i \leq \hat{w}_n(X_i)/\hat{B}) [Z_\tau(X_i,Y_i)-\Psi^\weighted_\tau(\breve{P}^n)] \\
    &\quad+ (P^{n,\test}-P^0) \tilde{D}(\mathcal{E}_{0,\tau},g_0,\gamma_0,1) + (P^{n,\test}-P^0) [\tilde{D}(\hat{\mathcal{E}}_{n,\tau}^\train,\hat{g}_n^\train,\hat{\gamma}_n^\train,\hat{\pi}_n) - \tilde{D}(\mathcal{E}_{0,\tau},g_0,\gamma_0,1)] \\
    &\quad + R^\weighted_\tau(\breve{P}^n,P^0,\hat{\mathcal{E}}_{n,\tau}^\train,\hat{g}_n^\train,\hat{\gamma}_n^\train,\hat{\pi}_n).
\end{align*}
With $\Gamma_{n,\tau}$ in Theorem~\ref{thm: rejection sampling one-step correction}, by \eqref{eq: RS remainder} and since $\gamma_{\breve{P}^n}=\gamma_0$, this further equals 
\begin{align*}
    & \Psi_\tau(P^0) + \Gamma_{n,\tau} 
    - P^0_{X \mid 1} \left[ \frac{\mathscr{W}(\hat{g}_n^\train,\hat{\gamma}_n^\train)}{\Pi_P(\mathscr{W}(\hat{g}_n^\train,\hat{\gamma}_n^\train))} - \mathscr{W}(g_0,\gamma_0) \right] (\hat{\mathcal{E}}_{n,\tau}^\train-\mathcal{E}_{0,\tau}) \\
    &\quad+ \gamma_0 \left( \frac{1}{\gamma_0} - \frac{\Pi_{P^0}(\mathscr{W}(\hat{g}_n^\train,\hat{\gamma}_n^\train))}{\hat{\gamma}_n^\train \hat{\pi}_n} \right) P^0_{X \mid 1} \left[ \frac{\mathscr{W}(\hat{g}_n^\train,\hat{\gamma}_n^\train)}{\Pi_{P^0}(\mathscr{W}(\hat{g}_n^\train,\hat{\gamma}_n^\train))} - \mathscr{W}(g_0,\gamma_0) \right] \hat{\mathcal{E}}_{n,\tau}^\train \\
    &\quad+ \frac{\hat{\gamma}_n^\train-\gamma_0}{\gamma_0 (1-\gamma_0 )} P^0_{X \mid 0} (\hat{\mathcal{E}}_{n,\tau}^\train - \mathcal{E}_{0,\tau})
    \ (\hat{\gamma}_n^\train-\gamma_0) \left( \frac{1}{\hat{\gamma}_n^\train (1-\hat{\gamma}_n^\train)} - \frac{1}{\gamma_0 (1-\gamma_0)} \right) P^0_{X \mid 0} \hat{\mathcal{E}}_{n,\tau}^\train \\
    &\quad+ \frac{\gamma_0}{\hat{\gamma}_n^\train} \left( 1 - \frac{\Pi_{P^0}(\mathscr{W}(\hat{g}_n^\train,\hat{\gamma}_n^\train))}{\hat{\pi}_n} \right) P^0_{X \mid 0} (\hat{\mathcal{E}}_{n,\tau}^\train-\mathcal{E}_{0,\tau}) \\
    &\quad+ (P^{n,\test}-P^0) [\tilde{D}(\hat{\mathcal{E}}_{n,\tau}^\train,\hat{g}_n^\train,\hat{\gamma}_n^\train,\hat{\pi}_n) - \tilde{D}(\mathcal{E}_{0,\tau},g_0,\gamma_0,1)].
\end{align*}
One key condition on the positivity of asymptotic variance $\varsigma_{0,\tau}^2$ required by Theorem~\ref{thm: general CI coverage} follows from Theorem~\ref{thm: rejection sampling has larger variance}, which shows that $\tau \in \mathcal{T}^\epsilon$ implies $\varsigma_{0,\tau}^2 > \epsilon$. The result for the case in which $\Psi_\tau(P^0)=0$ can be proved by directly showing that $\breve{\psi}_{n,\tau}=\varsigma_{n,\tau}=0$ with probability tending to one under Condition~\ref{cond: constant Q for extreme tau}.
Therefore, we omit this proof.

\begin{proof}[Proof of Theorem~\ref{thm: rejection sampling has larger variance}]
We directly calculate $\varsigma_{0,\tau}^2$ from \eqref{varsigmantau}. We have the following equalities for the two components of $\varsigma_{0,\tau}^2$ in \eqref{varsigmantau}
\begin{align*}
    & \expect_{P^0} \left[ \frac{(A - \gamma_0)^2}{\gamma_0^2 (1-\gamma_0)^2} \Psi_\tau(P^0)^2 \right]
    = \frac{\gamma_0 (1-\gamma_0)}{\gamma_0^2 (1-\gamma_0)^2} \Psi_\tau(P^0)^2 = \frac{\Psi_\tau(P^0)^2}{\gamma_0 (1-\gamma_0)},
\end{align*}
and
\begin{align*}
    & \expect_{P^0} \Bigg[ \Bigg\{ \hat{B} \frac{A}{\gamma_0} \ind(\zeta \leq w_0(X)/\hat{B}) [Z_\tau(X,Y)-\Psi_\tau(P^0)] \\
    &\quad+ \frac{A [w_0(X) - 1]}{\gamma_0} \Psi_\tau(P^0) + \tilde{D}(\mathcal{E}_{0,\tau},g_0,\gamma_0,1)(O) \Bigg\}^2 \Bigg] \\
    &= \frac{(\hat{B})^2}{\gamma_0^2} \expect_{P^0}[A \ind(\zeta \leq w_0(X)/\hat{B})] \mathrm{Var}_{P^0}(Z_\tau(X,Y) \mid A=0) \\
    &\quad+ \expect_{P^0}[A (w_0(X)-1)^2] \frac{\Psi_\tau(P^0)^2}{\gamma_0^2} + \expect_{P^0}[\tilde{D}(\mathcal{E}_{0,\tau},g_0,\gamma_0,1)(O)^2] \\
    &\quad+ 2 \frac{\hat{B} \Psi_\tau(P^0)}{\gamma_0^2} \expect_{P^0}[A \ind(\zeta \leq w_0(X)/\hat{B}) (Z_\tau(X,Y) - \Psi_\tau(P^0)) (w_0(X)-1)] \\
    &\quad+ 2 \frac{\hat{B}}{\gamma_0} \expect_{P^0}[A \ind(\zeta \leq w_0(X)/\hat{B}) (Z_\tau(X,Y) - \Psi_\tau(P^0)) \tilde{D}(\mathcal{E}_{0,\tau},g_0,\gamma_0,1)(O)] \\
    &\quad+ 2 \frac{\Psi_\tau(P^0)}{\gamma_0} \expect_{P^0}[A (w_0(X)-1) \tilde{D}(\mathcal{E}_{0,\tau},g_0,\gamma_0,1)(O)].
\end{align*}
This further equals
\begin{align*}
    &\frac{\hat{B} \{ \mathrm{Var}_{P^0}(\expect_{P^0}[Z_\tau(X,Y) \mid A=0,X]) + \expect_{P^0}[\mathrm{Var}_{P^0}(Z_\tau(X,Y) \mid A=0,X)] \}}{\gamma_0} \\
    &\quad+ \frac{\Psi_\tau(P^0)^2 \mathrm{Var}_{P^0}(w_0(X)-1 \mid A=1)}{\gamma_0} \\
    &\quad+ \frac{\expect_{P^0}[\mathcal{E}_{0,\tau}(X)^2 w_0(X)^2 \mid A=1]}{\gamma_0} + \frac{\expect_{P^0}[\mathcal{E}_{0,\tau}^2 \mid A=0]}{1-\gamma_0} \\
    &\quad+ \frac{2 \Psi_\tau(P^0) \expect_{P^0}[(\expect_{P^0}[Z_\tau(X,Y) \mid X,A=0]-\Psi_\tau(P^0))(w_0(X)-1) \mid A=0]}{\gamma_0} \\
    &\quad- 2 \frac{\Psi_\tau(P^0)}{\gamma_0} \expect_{P^0}[(\expect_{P^0}[Z_\tau(X,Y) \mid X,A=0] - \Psi_\tau(P^0)) \mathcal{E}_{0,\tau}(X) w_0(X) \mid A=0] \\
    &\quad- 2 \frac{\Psi_\tau(P^0)}{\gamma_0} \expect_{P^0}[(w_0(X)-1) \mathcal{E}_{0,\tau}(X) w_0(X) \mid A=1] \\
    &= \frac{\hat{B} P^0_{X \mid 0} \mathcal{E}_{0,\tau_0} (1-\mathcal{E}_{0,\tau_0}) + \hat{B} P^0_{X \mid 0} (\mathcal{E}_{0,\tau_0} - \Psi_\tau(P^0))^2}{\gamma_0} \\
    &\quad+ \frac{\Psi_\tau(P^0)^2 (P^0_{X \mid 0} w_0 - 1)}{\gamma_0} + \frac{P^0_{X \mid 0} w_0 \mathcal{E}_{0,\tau_0}^2}{\gamma_0} + \frac{P^0_{X \mid 0} \mathcal{E}_{0,\tau_0}^2}{1-\gamma_0} \\
    &\quad+ \frac{2 \Psi_\tau(P^0) P^0_{X \mid 0} (\mathcal{E}_{0,\tau_0} - \Psi_\tau(P^0)) (w_0-1)}{\gamma_0} \\
    &\quad- \frac{2 P^0_{X \mid 0} w_0 \mathcal{E}_{0,\tau} (\mathcal{E}_{0,\tau}-\Psi_\tau(P^0))}{\gamma_0} - \frac{2 \Psi_\tau(P^0) P^0_{X \mid 0} (w_0-1) \mathcal{E}_{0,\tau}}{\gamma_0}.
\end{align*}
Therefore, with $\xi$ from \eqref{xi}, $\varsigma_{0,\tau}^2$ equals
\begin{align*}
    &\frac{1}{\xi} \frac{\Psi_\tau(P^0)^2}{\gamma_0 (1-\gamma_0)} + \frac{1}{1-\xi} \Bigg\{ \frac{\hat{B} P^0_{X \mid 0} \mathcal{E}_{0,\tau_0} (1-\mathcal{E}_{0,\tau_0}) + \hat{B} P^0_{X \mid 0} (\mathcal{E}_{0,\tau_0} - \Psi_\tau(P^0))^2}{\gamma_0} \\
    &\quad+ \frac{\Psi_\tau(P^0)^2 - P^0_{X \mid 0} w_0 (\mathcal{E}_{0,\tau_0}-\Psi_\tau(P^0))^2}{\gamma_0} + \frac{P^0_{X \mid 0} \mathcal{E}_{0,\tau_0}^2}{1-\gamma_0} \Bigg\}.
\end{align*}
Since $\hat{B} \geq 1$, $0 < \xi < 1$, $\frac{1}{\xi} \frac{\Psi_\tau(P^0)^2}{\gamma_0 (1-\gamma_0)} > 0$ whenever $\sigma_{0,\tau}^2 > 0$, and $P^0_{X \mid 0} \mathcal{E}_{0,\tau_0}^2 = \Psi_\tau(P^0)^2 + P^0_{X \mid 0} (\mathcal{E}_{0,\tau_0}-\Psi_\tau(P^0))^2$, the above expression is greater than
\begin{align*}
    & \frac{P^0_{X \mid 0} \mathcal{E}_{0,\tau_0} (1-\mathcal{E}_{0,\tau_0}) + P^0_{X \mid 0} (\mathcal{E}_{0,\tau_0} - \Psi_\tau(P^0))^2 + \Psi_\tau(P^0)^2- P^0_{X \mid 0} (\mathcal{E}_{0,\tau_0} - \Psi_\tau(P^0))^2}{\gamma_0} + \frac{P^0_{X \mid 0} \mathcal{E}_{0,\tau_0}^2}{1-\gamma_0} \\
    &> \frac{P^0_{X \mid 0} w_0 \mathcal{E}_{0,\tau_0} (1-\mathcal{E}_{0,\tau_0})}{\gamma_0} + \frac{P^0_{X \mid 0} (\mathcal{E}_{0,\tau_0} - \Psi_\tau(P^0))^2}{1-\gamma_0},
\end{align*}
which equals $\sigma_{0,\tau}^2$. We have thus proved that $\sigma_{0,\tau}^2 < \varsigma_{0,\tau}^2$ whenever $\sigma_{0,\tau}^2 > 0$.
\end{proof}

\subsection{Theoretical results for CV-TMLE} \label{sec: proof CVTMLE efficiency}

We first prove a general bound for targeted nuisance estimators based on sample splitting, which is used in general CV-TMLE. This result shows that the targeted nuisance estimator converges at about the same rate as the initial nuisance estimator, so that convergence rates of the initial nuisance estimator are inherited by the targeted nuisance estimator. This is useful for showing that the second-order remainder of CV-TMLE is negligible. We reuse some notations---with a slight abuse---for this general setting when presenting the result. 

Let $(X_i,Y_i) \in \mathcal{X} \times \real$ ($i \in [n]$) be an i.i.d. sample from $P^0$, and $f^*: \mathcal{X} \mapsto \real$ be a fixed function in a function class $\funclass$. In the context of CV-TMLE, we may treat $f^*$ as the initial nuisance estimator obtained from an independent sample. Let $H_1,\ldots,H_K: \mathcal{X} \mapsto \real$ be $K$ fixed functions. Consider a generalized linear model (GLM) to fit $F:x \mapsto \expect_{P^0}[Y \mid X=x]$. 
Let $I \subseteq \R$ be an interval containing the range of $Y$, and let $g$ be the canonical link function for the model, an invertible map $g:I \rightarrow \R$. For example, $g$ may be the identity function for ordinary least squares and unbounded $Y$, and may be the expit function for logistic regression and binary $Y$.

For each function $f:\mathcal{X} \mapsto \real$, let $\ell(f): \mathcal{X} \times \real \rightarrow \real$ be the corresponding loss used for fitting the GLM, which is typically a negative log (working) likelihood. Suppose that for any square-integrable function $f: \mathcal{X} \mapsto \real$, $P^0 \ell(F) \leq P^0 \ell(f)$ and $P^0 \ell(f) - P^0 \ell(F) \simeq \| f - F\|_{P^0,2}^2$.
For each $\beta \in \real^K$, let $f_\beta: x \mapsto g^{-1}\left( g(f^*(x)) + \sum_{k=1}^K \beta_k H_k(x) \right)$  be the parametric GLM to be fitted and define $\ell(\beta):=\ell(f_\beta)$. Let $\dot{\ell}(\beta)$ and $\ddot{\ell}(\beta)$ denote the partial derivative and Hessian matrix, respectively, of $\ell(\beta)$ with respect to $\beta$.

Let $\beta_n \in \real^K$ be the MLE, which is a solution in $\beta$ to $P^n \dot{\ell}(\beta) = 0$. In the context of CV-TMLE, we may treat $f_{\beta_n}$ as the targeted nuisance estimator. Let $\beta_0 \in \real^K$ be the true (population) risk minimizer in the parametric GLM, the solution in $\beta$ to $P^0 \dot{\ell}(\beta) = 0$. Suppose that $\beta \mapsto \ell(\beta)$ is strictly convex, so that $\beta_0$ is unique and $P^0 \ddot{\ell}(\beta_0)$ is invertible. Suppose that for all $f^* \in \funclass$, $\beta_n - \beta_0 = \smallo_p(1)$ and $(P^n-P^0) \dot{\ell}(\beta_n) = \bigO_p(n^{-1/2})$. Further, suppose that $\| f_{\beta_1} - f_{\beta_2} \|_{P^0,2} \lesssim \| \beta_1 - \beta_2 \|$ where $\| \cdot \|$ denotes the Euclidean norm, and that $P^0 \dot{\ell}(\beta_0)^2 < \infty$.

\begin{theorem}[Bound for targeted nuisance estimators] \label{thm: bound for targeted nuisance}
Under the conditions of the above three paragraphs, it holds that $$\| f_{\beta_n} - F \|_{P^0,2} \lesssim \| f^* - F \|_{P^0,2} + \bigO_p(n^{-1/2}).$$
\end{theorem}

This result formally shows that, in CV-TMLE, the convergence rate of the targeted nuisance estimator ($f_{\beta_n}$ above) is essentially the same as that of the initial nuisance estimator ($f^*$ above) under no Donsker conditions.
Hence, the rate of convergence of nuisance estimators is typically inherited from the initial nuisance estimators. For example, if the second order remainder is an inner product or a quadratic form of two nuisance functions, and it is known that the initial nuisance estimators both converge at an $\smallo_p(n^{-1/4})$-rate (for example, obtained via the highly-adaptive lasso \protect\citepsupp{VanderLaan2017,Benkeser2016}), then, by Theorem~\ref{thm: bound for targeted nuisance}, the targeted nuisance estimators also converge at an $\smallo_p(n^{-1/4})$-rate and thus the remainder is $\smallo_p(n^{-1/2})$, as desired.

In CV-TMLE, typical GLMs used include ordinary least-squares and logistic regression, and they usually both satisfy the assumptions that $F$ minimizes $f \mapsto P^0 \ell(f)$, $P^0 \ell(f) - P^0 \ell(F) \simeq \| f - F\|_{P^0,2}^2$ and $\beta \mapsto \ell(\beta)$ is strictly convex (as long as $H_1,\ldots,H_K$ are linearly independent, which often holds in CV-TMLE) \protect\citepsupp{Benkeser2016,Qiu2021,VanderLaan2017}. Moreover, it also often holds that $\beta_n-\beta_0=\smallo_p(1)$ since the working log likelihood is strictly concave on $\real^K$.

Theorem~\ref{thm: bound for targeted nuisance} also implies that, if the initial nuisance estimator is consistent, then both the initial and the targeted nuisance estimators converge to $F$ in probability in an $L^2(P^0)$-sense, and hence the amount $\beta_n$ of adjustment for targeting converges to zero in probability.

\begin{proof}[Proof of Theorem~\ref{thm: bound for targeted nuisance}]
The proof is an application of Z-estimation theory \protect\citepsupp[see e.g., Chapter~3 in][]{vandervaart1996}. We have that
\begin{align*}
    0 &= P^n \dot{\ell}(\beta_n) - P^0 \dot{\ell}(\beta_0)
    = P^0 (\dot{\ell}(\beta_n) - \dot{\ell}(\beta_0)) + (P^n - P^0) \dot{\ell}(\beta_n) \\
    &= P^0 \ddot{\ell}(\beta_0) \cdot (\beta_n - \beta_0) + \smallo(\| \beta_n-\beta_0 \|) + (P^n - P^0) \dot{\ell}(\beta_n).
\end{align*}
Rearrange the terms to see that
$$\beta_n - \beta_0 = -(P^0 \ddot{\ell}(\beta_0))^{-1} (P^n - P^0) \dot{\ell}(\beta_n) + \smallo(\| \beta_n - \beta_0 \|).$$
The right-hand side is $\bigO_p(n^{-1/2}) + \smallo(\| \beta_n - \beta_0 \|)$ by assumption, and hence $\beta_n - \beta_0 = \bigO_p(n^{-1/2})$. 
Noting that $P^0 \ell(F) \leq P^0 \ell(\beta_0) \leq P^0 \ell(0)$ since $F$ minimizes $f \mapsto P^0 \ell(f)$ among all square-integrable functions and $\beta_0$ minimizes $\beta \mapsto P^0 \ell(f_\beta)$, we have that $\| f_{\beta_0} - F \|_{P^0,2}^2 \lesssim P^0 \ell(\beta_0) - P^0 \ell(F) \leq P^0 \ell(0) - P^0 \ell(F) = P^0 \ell(f^*) - P^0 \ell(F) \lesssim \| f^* - F \|_{P^0,2}^2$. Therefore,
\begin{align*}
    &\| f_{\beta_n} - F \|_{P^0,2} \leq \| f_{\beta_n} - f_{\beta_0} \|_{P^0,2} + \| f_{\beta_0} - F \|_{P^0,2} \\
    &\lesssim \| \beta_n - \beta_0 \| + \| f^* - F \|_{P^0,2}
    = \| f^* - F \|_{P^0,2} + \bigO_p(n^{-1/2})
\end{align*}
for all $f^* \in \funclass$. We have shown the desired inequality.
\end{proof}

The procedure to obtain the targeted nuisance estimator $\tilde{\mathcal{E}}_{n,\tau}^v$ in Section~\ref{sec: TMLE} satisfies the conditions of Theorem~\ref{thm: bound for targeted nuisance} with $F$ being $\mathcal{E}_{0,\tau}$, $f^*$ being $\hat{\mathcal{E}}_{n,\tau}^{-v}$ and $f_{\beta_n}$ being $\tilde{\mathcal{E}}_{n,\tau}^v$, because ordinary least squares or logistic regression is used with one fixed function $H_1= \mathscr{W}(\hat{g}_n^{-v},\hat{\gamma}_n^v)$ as the covariate. Hence $\| \hat{\mathcal{E}}_{n,\tau}^v - \mathcal{E}_{0,\tau} \|_{P^0_{X \mid 1},2} \lesssim \| \hat{\mathcal{E}}_{n,\tau}^{-v} - \mathcal{E}_{0,\tau} \|_{P^0_{X \mid 1},2} + \bigO_p(n^{-1/2})$ for all $v \in [V]$.

\begin{proof}[Proof of Theorem~\ref{thm: CV-TMLE efficiency}]
The proof is similar to the proof of Theorem~\ref{thm: CV one-step efficiency} and hence the arguments are abbreviated. We have that $\Psi^\Gcomp(\tilde{P}_{\tau}^{n,v})$ is consistent for $\Psi_\tau(P^0)$ by Theorem~\ref{thm: bound for targeted nuisance} and Condition~\ref{cond: sufficient nuisance rate}. We then have that
\begin{align}
    & \Psi^\Gcomp(\tilde{P}_{\tau}^{n,v}) - \Psi_\tau(P^0)
    = (P^{n,v} - P^0) D_\tau(P^0,\mathcal{E}_{0,\tau},g_0,\gamma_0,1) - P^{n,v} D^\Gcomp_\tau(\tilde{P}_{\tau}^{n,v},\tilde{\mathcal{E}}_{n,\tau}^v,\hat{g}_n^{-v},\hat{\gamma}_n^v,1) \nonumber\\
    & + (P^{n,v} - P^0) [D^\Gcomp_\tau(\tilde{P}_{\tau}^{n,v},\tilde{\mathcal{E}}_{n,\tau}^v,\hat{g}_n^{-v},\hat{\gamma}_n^v,1) - D_\tau(P^0,\mathcal{E}_{0,\tau},g_0,\gamma_0,1)]\nonumber\\ 
    &+ R^\Gcomp_\tau(\tilde{P}_{\tau}^{n,v},P^0,\tilde{\mathcal{E}}_{n,\tau}^v,\hat{g}_n^{-v},\hat{\gamma}_n^v,1).
 \label{eq: TMLE proof expansion}
\end{align}
Under Condition~\ref{cond: sufficient nuisance rate}, by Theorem~\ref{thm: bound for targeted nuisance},
recalling the form of $R^\Gcomp_\tau$ from \eqref{rt}, we have
$R^\Gcomp_\tau$ $(\tilde{P}_{\tau}^{n,v},P^0,\tilde{\mathcal{E}}_{n,\tau}^v,\hat{g}_n^{-v},\hat{\gamma}_n^v,1)$  = $\smallo_p(n^{-1/2})$.

We first condition on $(\hat{\mathcal{E}}_{n,\tau}^{-v},\hat{g}_n^{-v})$. For a sequence $(\delta_n)_{n\ge 1}$, such that $\delta_n>0$ for all $n\ge 1$, converging to zero at an appropriate rate discussed below, 
consider the function class
\begin{align*}
    \funclass_{\delta_n}  &:= \{ o \mapsto \frac{a}{\gamma} \mathscr{W}(\hat{g}_n^{-v},\gamma)(x) \left\{ Z_\tau(x,y)- \expit \left\{ \logit \hat{\mathcal{E}}_{n,\tau}^{-v}(x) + \beta \frac{a}{\gamma} \mathscr{W}(\hat{g}_n^{-v},\hat{\gamma})(x) \right\} \right\} \\
    &+ \frac{1-a}{1-\gamma} \Bigg[ \expit \left\{ \logit \hat{\mathcal{E}}_{n,\tau}^{-v}(x) + \beta \frac{a}{\gamma} \mathscr{W}(\hat{g}_n^{-v},\gamma)(x) \right\} 
    - \Psi^\Gcomp(\tilde{P}_{\tau}^{n,v},\mathscr{W}(\hat{g}_n^{-v},\gamma)) \Bigg] \\
    &- D_\tau(P^0,\mathcal{E}_{0,\tau},g_0,\gamma_0,1): 
    \gamma \in [\gamma_0-\delta_n,\gamma_0+\delta_n], \beta \in [-\delta_n,\delta_n]\}.
\end{align*} 
Let $\{\delta_n\}_{n \geq 1}$ be a sequence converging to zero at a sufficiently slow rate such that, with probability tending to one, the fitted coefficient in Step~\ref{TMLE step: target reg} of the Algorithm in Section~\ref{sec: TMLE} lies in $[-\delta_n,\delta_n]$ and $\hat{\gamma}_n^v$ lies in $[\gamma_0-\delta_n,\gamma_0+\delta_n]$. By an argument similar to that in the proof of Theorem~\ref{thm: CV one-step efficiency}, we can show that the function class $\funclass_{\delta_n}$ is a BUEI class containing $D^\Gcomp_\tau(\tilde{P}_{\tau}^{n,v},\tilde{\mathcal{E}}_{n,\tau}^v,\hat{g}_n^{-v},\hat{\gamma}_n^v,1) - D_\tau(P^0,\mathcal{E}_{0,\tau},g_0,\gamma_0,1)$ with probability tending to one. Therefore,
$$\expect_{P^0} \sqrt{n} \left[ (P^{n,v} - P^0) [D^\Gcomp_\tau(\tilde{P}_{\tau}^{n,v},\hat{\mathcal{E}}_{n,\tau}^{-v},\hat{g}_n^{-v},\hat{\gamma}_n^v,1) - D_\tau(P^0,\mathcal{E}_{0,\tau},g_0,\gamma_0,1)]  \right] = \smallo(1)$$
under Condition~\ref{cond: sufficient nuisance rate}, and hence $(P^{n,v} - P^0) [D^\Gcomp_\tau(\tilde{P}_{\tau}^{n,v},\tilde{\mathcal{E}}_{n,\tau}^v,\hat{g}_n^{-v},\hat{\gamma}_n^v,1) - D_\tau(P^0,\mathcal{E}_{0,\tau},g_0,\gamma_0,1)] = \smallo_p(n^{-1/2})$.

By the construction of $\tilde{\mathcal{E}}_{n,\tau}^v$, which is obtained by solving the estimating equation corresponding to logistic regression or ordinary least squares \protect\citepsupp[see e.g., Line~6 on page~261, Section~6.5.1 in][]{Wakefield2013} in Step~\ref{TMLE step: target reg} in the algorithm in Section~\ref{sec: TMLE}, we have that
$$\frac{1}{|I_v|} \sum_{i \in I_v} A_i \mathscr{W}(\hat{g}_n^{-v},\hat{\gamma}_n^v)(X_i) (Z_\tau(X_i,Y_i) - \tilde{\mathcal{E}}_{n,\tau}^v(X_i)) = 0$$
and thus
$$P^{n,v} D^\Gcomp_\tau(\tilde{P}_{\tau}^{n,v},\tilde{\mathcal{E}}_{n,\tau}^v,\hat{g}_n^{-v},\hat{\gamma}_n^v,1) = 0.$$

Plugging all the above results into \eqref{eq: TMLE proof expansion}, we have that
$$\Psi^\Gcomp(\tilde{P}_{\tau}^{n,v}) - \Psi_\tau(P^0) = (P^{n,v} - P^0) D_\tau(P^0,\mathcal{E}_{0,\tau},g_0,\gamma_0,1) + \smallo_p(n^{-1/2})$$
and hence $\tilde{\psi}_{n,\tau} - \Psi_\tau(P^0) = (P^n - P^0) D_\tau(P^0,\mathcal{E}_{0,\tau},g_0,\gamma_0,1) + \smallo_p(n^{-1/2})$.
The claimed uniform convergence follows by noting that the above arguments apply uniformly to $\tau \in \mathcal{T}_n$.
\end{proof}

\begin{proof}[Proof of Theorem~\ref{thm: convergence rate of TMLE Wald CI}]
Similar arguments to those used to analyze the cross-fit one-step corrected estimator show that Lemmas~\ref{lemma: CV one-step empirical process and remainder bound}--\ref{lemma: CV one-step var convergence bound} hold with $(\hat{\psi}_{n,\tau},\hat{\sigma}_{n,\tau})$ replaced by $(\tilde{\psi}_{n,\tau},\tilde{\sigma}_{n,\tau})$. Therefore, we can apply Theorem~\ref{thm: general CI coverage}, and the desired result for $\mathcal{T}^\epsilon$ follows. For $\tau \in \mathcal{T}^-$, it is not hard to check that $\tilde{\psi}_{n,\tau}=\Psi_\tau(P^0)$ and $\tilde{\sigma}_{n,\tau}=0$, and hence the desired result follows.
\end{proof}

The proof of the corresponding version of Corollary~\ref{corollary: CV one-step APAC} is strikingly similar to that for the cross-fit one-step corrected estimator, and is thus omitted.

\subsection{Double robustness of PredSet-1Step in special cases (Section~\ref{sec: one step DR})}

We first show our claims about PredSet-1Step under Condition~\ref{cond: known nuisance} from Section~\ref{sec: one step DR known}. For conciseness of the proof, when the conditional coverage error rate $\mathcal{E}_{0,\tau}$ is known, we still use $\hat{\mathcal{E}}_{n,\tau}^{-v}$ to denote the coverage error being used to compute the estimator $\hat{\psi}_{n,\tau}$, and similarly for the propensity score estimator $\hat{g}_n^{-v}$ when $g_0$ is known.

We first note that, by the definition of $\hat{\psi}_{n,\tau}^v$ in \eqref{psi-n-v}, we have that
$$P^{n,v} \left\{ A \frac{1-\hat{g}_n^{-v}}{\hat{g}_n^{-v}} [Z_\tau - \hat{\mathcal{E}}_{n,\tau}^{-v}] + (1-A)[\hat{\mathcal{E}}_{n,\tau}^{-v}-\hat{\psi}_{n,\tau}^v] \right\}=0.$$
Subtract \eqref{eq: DR IF} from this equality  and rearrange terms to obtain that
\begin{align*}
    0 &= P^{n,v} \left\{ A \frac{1-\hat{g}_n^{-v}}{\hat{g}_n^{-v}} [Z_\tau - \hat{\mathcal{E}}_{n,\tau}^{-v}] + (1-A)[\hat{\mathcal{E}}_{n,\tau}^{-v}-\hat{\psi}_{n,\tau}^v] \right\} \\
    &\quad- P^0 \left\{ A \frac{1-g_\infty}{g_\infty} [Z_\tau - \mathcal{E}_{\infty,\tau}] + (1-A) [\mathcal{E}_{\infty,\tau} - \Psi_\tau(P^0)] \right\} \\
    &= (P^{n,v}-P^0) \left\{ A \frac{1-g_\infty}{g_\infty} [Z_\tau - \mathcal{E}_{\infty,\tau}] + (1-A) [\mathcal{E}_{\infty,\tau} - \Psi_\tau(P^0)] \right\} \\
    &\quad+ (P^{n,v}-P^0) \Bigg\{ \left\{ A \frac{1-\hat{g}_n^{-v}}{\hat{g}_n^{-v}} [Z_\tau - \hat{\mathcal{E}}_{n,\tau}^{-v}] + (1-A)[\hat{\mathcal{E}}_{n,\tau}^{-v}-\hat{\psi}_{n,\tau}^v] \right\} \\
    &\qquad- \left\{ A \frac{1-g_\infty}{g_\infty} [Z_\tau - \mathcal{E}_{\infty,\tau}] + (1-A) [\mathcal{E}_{\infty,\tau} - \Psi_\tau(P^0)] \right\} \Bigg\} \\
    &\quad+ P^0 \Bigg\{ \left\{ A \frac{1-\hat{g}_n^{-v}}{\hat{g}_n^{-v}} [Z_\tau - \hat{\mathcal{E}}_{n,\tau}^{-v}] + (1-A)[\hat{\mathcal{E}}_{n,\tau}^{-v}-\hat{\psi}_{n,\tau}^v] \right\} \\
    &\qquad- \left\{ A \frac{1-g_\infty}{g_\infty} [Z_\tau - \mathcal{E}_{\infty,\tau}] + (1-A) [\mathcal{E}_{\infty,\tau} - \Psi_\tau(P^0)] \right\} \Bigg\}.
\end{align*}
This further equals
\begin{align*}
    &(P^{n,v}-P^0) \left\{ A \frac{1-g_\infty}{g_\infty} [Z_\tau - \mathcal{E}_{\infty,\tau}] + (1-A) [\mathcal{E}_{\infty,\tau} - \Psi_\tau(P^0)] \right\} \\
    &\quad+ (P^{n,v}-P^0) \Bigg\{ \left\{ A \frac{1-\hat{g}_n^{-v}}{\hat{g}_n^{-v}} [Z_\tau - \hat{\mathcal{E}}_{n,\tau}^{-v}] + (1-A)[\hat{\mathcal{E}}_{n,\tau}^{-v}-\hat{\psi}_{n,\tau}^v] \right\} \\
    &\qquad- \left\{ A \frac{1-g_\infty}{g_\infty} [Z_\tau - \mathcal{E}_{\infty,\tau}] + (1-A) [\mathcal{E}_{\infty,\tau} - \Psi_\tau(P^0)] \right\} \Bigg\} \\
    &\quad- (1-\gamma_0)(\hat{\psi}_{n,\tau}^v-\Psi_\tau(P^0)) - \gamma_0 P^0_{X \mid 1} \left\{ \left( \frac{1-\hat{g}_n^{-v}}{\hat{g}_n^{-v}} - \frac{1-g_\infty}{g_\infty} \right) (\hat{\mathcal{E}}_{n,\tau}-\mathcal{E}_{0,\tau}) \right\} \\
    &\quad- \gamma_0 P^0_{X \mid 1} \left\{ \left( \frac{1-g_\infty}{g_\infty} - \frac{1-g_0}{g_0} \right) (\hat{\mathcal{E}}_{n,\tau}^{-v}-\mathcal{E}_{\infty,\tau}) \right\}.
\end{align*}
Thus,
\begin{align}
    \begin{split}
    \hat{\psi}_{n,\tau}^v - \Psi_\tau(P^0) &= \frac{1}{1-\gamma_0} (P^{n,v}-P^0) \left\{ A \frac{1-g_\infty}{g_\infty} [Z_\tau - \mathcal{E}_{\infty,\tau}] + (1-A) [\mathcal{E}_{\infty,\tau} - \Psi_\tau(P^0)] \right\} \\
    &\quad+ \frac{1}{1-\gamma_0} (P^{n,v}-P^0) \Bigg\{ \left\{ A \frac{1-\hat{g}_n^{-v}}{\hat{g}_n^{-v}} [Z_\tau - \hat{\mathcal{E}}_{n,\tau}^{-v}] + (1-A)[\hat{\mathcal{E}}_{n,\tau}^{-v}-\hat{\psi}_{n,\tau}^v] \right\} \\
    &\qquad- \left\{ A \frac{1-g_\infty}{g_\infty} [Z_\tau - \mathcal{E}_{\infty,\tau}] + (1-A) [\mathcal{E}_{\infty,\tau} - \Psi_\tau(P^0)] \right\} \Bigg\} \\
    &\quad- \frac{\gamma_0}{1-\gamma_0} P^0_{X \mid 1} \left\{ \left( \frac{1-\hat{g}_n^{-v}}{\hat{g}_n^{-v}} - \frac{1-g_\infty}{g_\infty} \right) (\hat{\mathcal{E}}_{n,\tau}^{-v}-\mathcal{E}_{0,\tau}) \right\} \\
    &\quad- \frac{\gamma_0}{1-\gamma_0} P^0_{X \mid 1} \left\{ \left( \frac{1-g_\infty}{g_\infty} - \frac{1-g_0}{g_0} \right) (\hat{\mathcal{E}}_{n,\tau}^{-v}-\mathcal{E}_{\infty,\tau}) \right\}.
    \end{split} \label{eq: DR one step}
\end{align}
Since $\hat{\mathcal{E}}_{n,\tau}^{-v}=\mathcal{E}_{\infty,\tau}=\mathcal{E}_{0,\tau}$ or $\hat{g}_n^{-v}=g_\infty=g_\infty$, we have that both terms below are zero.
\begin{align*}
    &\frac{\gamma_0}{1-\gamma_0} P^0_{X \mid 1} \left\{ \left( \frac{1-\hat{g}_n^{-v}}{\hat{g}_n^{-v}} - \frac{1-g_\infty}{g_\infty} \right) (\hat{\mathcal{E}}_{n,\tau}^{-v}-\mathcal{E}_{0,\tau}) \right\}, \\
    & \frac{\gamma_0}{1-\gamma_0} P^0_{X \mid 1} \left\{ \left( \frac{1-g_\infty}{g_\infty} - \frac{1-g_0}{g_0} \right) (\hat{\mathcal{E}}_{n,\tau}^{-v}-\mathcal{E}_{\infty,\tau}) \right\}.
\end{align*}
Therefore,
\begin{align*}
    \hat{\psi}_{n,\tau}^v - \Psi_\tau(P^0) &= \frac{1}{1-\gamma_0} (P^{n,v}-P^0) \left\{ A \frac{1-g_\infty}{g_\infty} [Z_\tau - \mathcal{E}_{\infty,\tau}] + (1-A) [\mathcal{E}_{\infty,\tau} - \Psi_\tau(P^0)] \right\} \\
    &\quad+ \frac{1}{1-\gamma_0} (P^{n,v}-P^0) \Bigg\{ \left\{ A \frac{1-\hat{g}_n^{-v}}{\hat{g}_n^{-v}} [Z_\tau - \hat{\mathcal{E}}_{n,\tau}^{-v}] + (1-A)[\hat{\mathcal{E}}_{n,\tau}^{-v}-\hat{\psi}_{n,\tau}^v] \right\} \\
    &\qquad- \left\{ A \frac{1-g_\infty}{g_\infty} [Z_\tau - \mathcal{E}_{\infty,\tau}] + (1-A) [\mathcal{E}_{\infty,\tau} - \Psi_\tau(P^0)] \right\} \Bigg\}.
\end{align*}
The rest of the proof is strikingly similar to the proof for the nonparametric case after replacing $(\mathcal{E}_{0,\tau},g_0)$ by $(\mathcal{E}_{\infty,\tau},g_\infty)$, and thus omitted. 
The difference in the convergence rate of CUB coverage is due to the fact that the mixed bias term now equals zero and thus the other terms that are dominated by the mixed bias term in the nonparametric case.

We next prove our claim about PredSet-1Step under Condition~\ref{cond: Hadamard differentiable nuisance} in Section~\ref{sec: one step DR parametric}. By a similar argument as above, we have that
\begin{align*}
    \hat{\psi}_{n,\tau} - \Psi_\tau(P^0) &= \frac{1}{1-\gamma_0} (P^n-P^0) \left\{ A \frac{1-g_\infty}{g_\infty} [Z_\tau - \mathcal{E}_{\infty,\tau}] + (1-A) [\mathcal{E}_{\infty,\tau} - \Psi_\tau(P^0)] \right\} \nonumber \\
    &\quad+ \frac{1}{1-\gamma_0} (P^n-P^0) \Bigg\{ \left\{ A \frac{1-\hat{g}_n}{\hat{g}_n} [Z_\tau - \hat{\mathcal{E}}_{n,\tau}] + (1-A)[\hat{\mathcal{E}}_{n,\tau}-\hat{\psi}_{n,\tau}] \right\} \nonumber \\
    &\qquad- \left\{ A \frac{1-g_\infty}{g_\infty} [Z_\tau - \mathcal{E}_{\infty,\tau}] + (1-A) [\mathcal{E}_{\infty,\tau} - \Psi_\tau(P^0)] \right\} \Bigg\} \nonumber \\
    &\quad- \frac{\gamma_0}{1-\gamma_0} P^0_{X \mid 1} \left\{ \left( \frac{1-\hat{g}_n}{\hat{g}_n} - \frac{1-g_\infty}{g_\infty} \right) (\hat{\mathcal{E}}_{n,\tau}-\mathcal{E}_{0,\tau}) \right\} \nonumber \\
    &\quad- \frac{\gamma_0}{1-\gamma_0} P^0_{X \mid 1} \left\{ \left( \frac{1-g_\infty}{g_\infty} - \frac{1-g_0}{g_0} \right) (\hat{\mathcal{E}}_{n,\tau}-\mathcal{E}_{\infty,\tau}) \right\}. \nonumber
\end{align*}
This further equals
\begin{align}
    \begin{split}
    &= \frac{1}{1-\gamma_0} (P^n-P^0) \left\{ A \frac{1-g_\infty}{g_\infty} [Z_\tau - \mathcal{E}_{\infty,\tau}] + (1-A) [\mathcal{E}_{\infty,\tau} - \Psi_\tau(P^0)] \right\} \\
    &\quad+ \frac{1}{1-\gamma_0} (P^n-P^0) \Bigg\{ \left\{ A \frac{1-\hat{g}_n}{\hat{g}_n} [Z_\tau - \hat{\mathcal{E}}_{n,\tau}] + (1-A)[\hat{\mathcal{E}}_{n,\tau}-\hat{\psi}_{n,\tau}] \right\} \\
    &\qquad- \left\{ A \frac{1-g_\infty}{g_\infty} [Z_\tau - \mathcal{E}_{\infty,\tau}] + (1-A) [\mathcal{E}_{\infty,\tau} - \Psi_\tau(P^0)] \right\} \Bigg\} \\
    &\quad- \frac{\gamma_0}{1-\gamma_0} P^0_{X \mid 1} \left\{ \left( \frac{1-\hat{g}_n}{\hat{g}_n} - \frac{1-g_\infty}{g_\infty} \right) (\hat{\mathcal{E}}_{n,\tau}-\mathcal{E}_{\infty,\tau}) \right\} \\
    &\quad- \frac{\gamma_0}{1-\gamma_0} P^0_{X \mid 1} \left\{ \left( \frac{1-\hat{g}_n}{\hat{g}_n} - \frac{1-g_\infty}{g_\infty} \right) (\mathcal{E}_{\infty,\tau}-\mathcal{E}_{0,\tau}) \right\} \\
    &\quad- \frac{\gamma_0}{1-\gamma_0} P^0_{X \mid 1} \left\{ \left( \frac{1-g_\infty}{g_\infty} - \frac{1-g_0}{g_0} \right) (\hat{\mathcal{E}}_{n,\tau}-\mathcal{E}_{\infty,\tau}) \right\}.
    \end{split} \label{eq: DR one step parametric}
\end{align}
We apply Lemma~\ref{sec: proof efficiency} to $(1-\hat{g}_n)/\hat{g}_n-(1-g_\infty)/g_\infty$ with function $f: x \mapsto (1-x)/x$ to obtain that
$$\frac{1-\hat{g}_n(x)}{\hat{g}_n(x)} - \frac{1-g_\infty(x)}{g_\infty(x)} = - \frac{1}{g_\infty(x)^2} P^n \IF^g(\cdot,x) + \bigO_p(n^{-1}).$$
Thus,
\begin{align*}
    &\frac{\gamma_0}{1-\gamma_0} P^0_{X \mid 1} \left\{ \left( \frac{1-\hat{g}_n}{\hat{g}_n} - \frac{1-g_\infty}{g_\infty} \right) (\hat{\mathcal{E}}_{n,\tau}-\mathcal{E}_{\infty,\tau}) \right\} = \bigO_p(n^{-1}), \\
    &\frac{\gamma_0}{1-\gamma_0} P^0_{X \mid 1} \left\{ \left( \frac{1-\hat{g}_n}{\hat{g}_n} - \frac{1-g_\infty}{g_\infty} \right) (\mathcal{E}_{\infty,\tau}-\mathcal{E}_{0,\tau}) \right\} \\
    &= -\frac{\gamma_0}{1-\gamma_0} \frac{1}{n} \sum_{i=1}^n P^0_{X \mid 1} \left\{ \frac{1}{g_\infty(\cdot)^2} [\mathcal{E}_{\infty,\tau}(\cdot)-\mathcal{E}_{0,\tau}(\cdot)] \IF^g(O_i,\cdot) \right\} + \bigO_p(n^{-1}), \\
    & \frac{\gamma_0}{1-\gamma_0} P^0_{X \mid 1} \left\{ \left( \frac{1-g_\infty}{g_\infty} - \frac{1-g_0}{g_0} \right) (\hat{\mathcal{E}}_{n,\tau}-\mathcal{E}_{\infty,\tau}) \right\} \\
    &= \frac{\gamma_0}{1-\gamma_0} \frac{1}{n} \sum_{i=1}^n P^0_{X \mid 1} \left\{ \left( \frac{1-g_\infty(\cdot)}{g_\infty(\cdot)} - \frac{1-g_0(\cdot)}{g_0(\cdot)} \right) \IF^\mathcal{E}_\tau(O_i,\cdot) \right\} + \bigO_p(n^{-1}).
\end{align*}
Plugging the above expansions into \eqref{eq: DR one step parametric} yields the claimed asymptotic linearity of $\hat{\psi}_{n,\tau}$ in \eqref{eq: DR one step parametric RAL}.

\section{Discussion of the PAC property and marginal validity} \label{sec: discuss PAC}

In the context of supervised learning, marginal validity of a prediction set $\hat{C}$---learned on training data---refers to a guarantee of the form
$$\Prob(Y \notin \hat{C}(X)) \leq \alpha,$$
where $(X,Y)$ is a new independent observation, and $1-\alpha$ is the confidence level specified by the user. 
In this guarantee, the probability statement marginalizes over both randomness in data and randomness in the new observation. We interpret this statement in the frequentist sense below. Define a statistical experiment as the following steps:
\begin{enumerate}
    \item collect training data $(X_i,Y_i)$ ($i=1,\ldots,n$);
    \item collect a new observation $(X_{n+1},Y_{n+1})$ with $Y_{n+1}$ unobserved;
    \item construct a prediction set $\hat{C}(X_{n+1})$ based on training data;
    \item observe $Y_{n+1}$ and check miscoverage $\ind(Y_{n+1} \notin \hat{C}(X_{n+1})$.
\end{enumerate}
Marginal validity means that, if we run this statistical experiment many times, then the proportion of runs where we observe miscoverage is approximately below $\alpha$.

In contrast, the PAC guarantee (i.e. training-set conditional validity) takes the form
$$\Prob( \Prob(Y \notin \hat{C}(X) \mid \hat{C}) \leq \alpha_\error ) \geq 1-\alpha_\conf.$$
Thus the randomness in the training data and the new observation is decoupled, and we condition on the training data that outputs $\hat{C}$ in the inner probability. We interpret this statement in the frequentist sense below. Define a statistical experiment as the following steps:
\begin{enumerate}
    \item collect training data $(X_i,Y_i)$ ($i=1,\ldots,n$);
    \item construct a prediction set $\hat{C}$ that assigns a set $\hat{C}(x)$ for each given covariate $x$;
    \item collect many new observations $(X_{n+j},Y_{n+j})$ ($j=1,\ldots,N$) with all $Y_{n+j}$ unobserved, and calculate a prediction set $\hat{C}(X_{n+j})$ for each $X_{n+j}$;
    \item observe $Y_{n+j}$ ($j=1,\ldots,N$) and calculate the proportion of miscoverage $\frac{1}{N} \sum_{j=1}^N \ind(Y_{n+j} \notin \hat{C}(X_{n+j}))$.
\end{enumerate}
The PAC guarantee means that, if we run this statistical experiment many times, as $N\rightarrow\infty$, the proportion of runs where we observe miscoverage proportion below $\alpha_\error$ is approximately above $1-\alpha_\conf$.

Comparing these two statistical experiments, it is evident that marginal validity is about prediction of one instance given the data at hand, while PAC guarantee is about prediction of (potentially inifinitely) many new instances given the data at hand. In many applications, we wish to use the given data to train one model and predict many new instances, and thus PAC guarantee might be be more desirable and meaningful than marginal validity.

As mentioned in Section~\ref{sec: intro}, PAC property results are known for inductive conformal prediction under no covariate shift \protect\citepsupp{Vovk2013} or known covariate shift \protect\citepsupp{park2021pac}. Inductive conformal prediction methods under unknown covariate shift have also been proposed \protect\citepsupp{Tibshirani2019}, but, to our best knowledge, conditional validity results are unknown for these methods. 
One possible challenge is that, under covariate shift, the distribution of the nonconformity scores that is used to obtain a quantile is no longer uniform and involves the likelihood ratio in \eqref{qdef}. 

In particular, this distribution involves the likelihood ratio evaluated at the new observed covariate, and it appears necessary to calculate the quantile for every new observation. Under known covariate shift, \protect\citetsupp{park2021pac} resolved this issue by obtaining a sample from the target population via rejection sampling and thus reducing the prediction set problem under covariate shift to the ordinary prediction set problem without covariate shift. However, under unknown covariate shift, as shown in Section~\ref{sec: rejection sampling method}, we only obtain a sample from an approximation to the target population via rejection sampling, which leads to further complications. Our proposed methods to construct prediction sets are the first to achieve asymptotic PAC guarantee under unknown covariate shift.

\section{Connection between causal inference and covariate shift} \label{section: causal and covariate shift}

There is a connection between counterfactuals in causal inference and covariate shift, as pointed out in \protect\citetsupp{Lei2021}. Here we provide some additional details. 
This connection allows us to apply our methods to problems in causal inference. 
In causal inference, with covariates $X$, we use $A \in \{0,1\}$ to denote a binary treatment, and $\tilde{Y}$ to denote the observed outcome. 
The counterfactual outcomes are $\tilde{Y}(0)$ and $\tilde{Y}(1)$, corresponding to setting the treatment $A$ to $0$ and $1$, respectively. 
In other words, $\tilde{Y}(a)$ is the outcome that would be observed if the treatment $A$ were set to $a$ ($a \in \{0,1\}$). In the observed data, the outcome is $\tilde{Y}=A \tilde{Y}(1) + (1-A) \tilde{Y}(0)$;  only the counterfactual outcome corresponding to the treatment taken is observed, while the other is missing. An important problem is then to predict the individual treatment effect $\tilde{Y}(1)-\tilde{Y}(0)$---or equivalently $\tilde{Y}(1)$---for an individual with treatment $A=0$.

If the treatment is not randomized (e.g., in observational data), it is well known that the distribution of $\tilde{Y} \mid A=a$ might not be identical to that of $\tilde{Y}(a)$ due to potential \textit{confounders} that affect both the treatment assignment and the counterfactual outcomes. For example, suppose that treatment 1 is believed to be more risky but in fact has a higher chance to be effective than treatment 0. Then, healthier patients might tend to choose treatment 0 while sicker patients might tend to choose treatment 1. In this case, with higher $\tilde{Y}$ denoting better outcomes, it is likely that $\expect[\tilde{Y} \mid A=0]$ is higher than $\expect[\tilde{Y}(0)]$ while $\expect[\tilde{Y} \mid A=1]$ is lower than $\expect[\tilde{Y}(1)]$. 
If all confounders are included in the covariate $X$, such a bias can be eliminated with techniques developed in causal inference. We make this no-unmeasured-confounding assumption throughout the rest of this section.

In traditional causal inference, the distribution of $(X,A,\tilde{Y})$ is decomposed into (i) the marginal distribution of $X$, (ii) the conditional distribution of $A$ given $X$, encoded by the propensity score function $g: x \mapsto \Prob(A=1 \mid X=x)$ \protect\citepsupp{Rosenbaum1983}, and (iii) the conditional distribution of $\tilde{Y}$ given $(A,X)$. In particular, for the purpose of predicting $\tilde{Y}(0)$, it suffices to consider the distribution of $\tilde{Y}$ given $A=0$ and $X$. The propensity score function $g$ is important because it contains the information on the treatment assignment mechanism and therefore the difference between the populations with $A=1$ and $A=0$.

We now describe how to view this problem of predicting $\tilde{Y}(1)$ for an individual with $A=0$ as a prediction problem under covariate shift, particularly Conditions~\ref{cond: positivity of P(A)}--\ref{cond: target dominated by source}. 
We may view $A=1$ as the source population and $A=0$ as the target population.
Condition~\ref{cond: positivity of P(A)} typically holds. Define $Y:=\tilde{Y}(1)$, which is observed in the population with $A=1$ but unobserved in the population with $A=0$. Under the no-unmeasured-confounding assumption, the distribution of $Y \mid A=1,X=x$ is identical to $Y \mid A=0,X=x$ for all $x \in \mathcal{X}$. In other words, Condition~\ref{cond: same Y|X} holds. By Bayes' theorem, with $P_{X,a}$ denoting the distribution of $X \mid A=a$, we have that
$$\frac{\intd P_{X \mid 0}}{\intd P_{X \mid 1}} (x) = \frac{1-g(x)}{g(x)} \frac{\Prob(A=1)}{\Prob(A=0)}.$$
Condition~\ref{cond: target dominated by source} is then equivalent to requiring that $g(X)>0$ a.s., which is a positivity assumption that is standard in the causal inference literature \protect\citepsupp[see, e.g.,][]{VanderLaan2018,Yang2018}. 
Thus, under the above standard causal assumptions, our methods can be applied to construct a PAC prediction set $\hat{C}_n$ for $\tilde{Y}(1)$ among the untreated group, which satisfies
$$\Prob_{P^0}(\Prob_{\bar{P}^0}(\tilde{Y}(1) \notin \hat{C}_n(X) \mid A=0,\hat{C}_n) \leq \alpha_\error) \geq 1-\alpha_\conf+\smallo(1).$$

\section{Literature on confidence interval coverage based on efficient estimators involving nuisance function estimation} \label{sec: CI coverage lit review}

A few recent works concern confidence interval (CI) coverage based on efficient and multiply robust estimators in causal inference or missing data applications, but they lack theoretical results showing that their proposed methods improve CI coverage. For example, in Chapter~28 of \citet{VanderLaan2018}, the authors presented several approaches to constructing CIs based on an efficient average treatment effect estimator and theoretically showed that they all attain the nominal coverage level asymptotically.
However, they provided no further statements about the convergence rates of CI coverage to nominal coverage, but only a simulation study. As another example, \citet{Tran2018} proposed a few alternative CIs for the average treatment effect. Although they theoretically showed that these approaches are based on asymptotically efficient estimators of the asymptotic variance, they did not theoretically show that this would lead to improved CI coverage; instead, they also presented evidence from simulations. Further, \citet{Bindele2018} proposed confidence intervals with asymptotic coverage guarantees, but only provided numerical evidence of improved coverage.

There are also some more distantly related works.
\citet{Rothe2017} proposed CIs with improved coverage for the average treatment effect under limited overlap.
The authors relied on an additional Gaussian assumption, which is not applicable in applications with binary outcomes such as the prediction sets in our current paper.
\citep{Matsouaka2023} studied CI coverage for a variety of average treatment effect estimation methods when nuisance function models are parametric; the authors provided evidence from empirical simulations only.

In conclusion, theoretical guidance on how to improve CI coverage for efficient estimators involving nuisance function estimation is lacking.

{\small 
\setlength{\bibsep}{0.2pt plus 0.3ex}
\bibliographystylesupp{chicago}
\bibliographysupp{ref}
}

\begin{figure}[h!]
    \centering
    \includegraphics[scale=0.8]{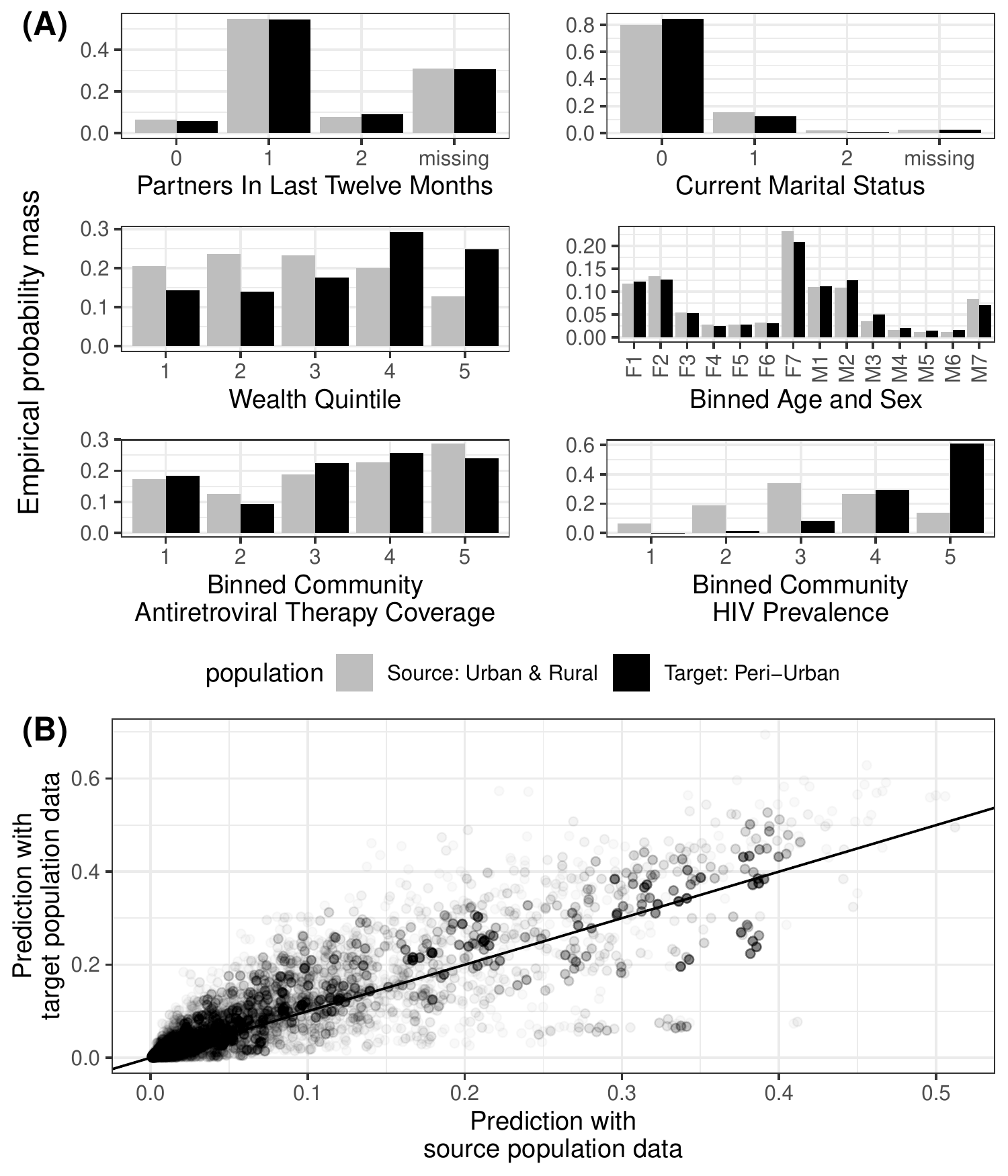}
    \caption{(A) Covariate distributions in the two populations of the data concerning HIV risk prediction in a South African cohort.
    A severe shift in community HIV prevalence and a moderate shift in wealth are present. (B) Outcome predictions for the entire data set with the predictors trained on the source population data and target population data, respectively, via Super Learner \protect\citepsupp{VanderLaan2007} with gradient boosting \protect\citepsupp{Friedman2001,Friedman2002,Mason1999,Mason2000} included in the library. The straight line $y=x$ represents identical predictions. Using data from the two populations, we obtain similar predictions for the entire data set, suggesting that the outcome distributions given covariates may be similar in the two populations.}
    \label{fig: cov shift}
\end{figure}

\begin{figure}
    \centering
    \includegraphics[scale=0.7]{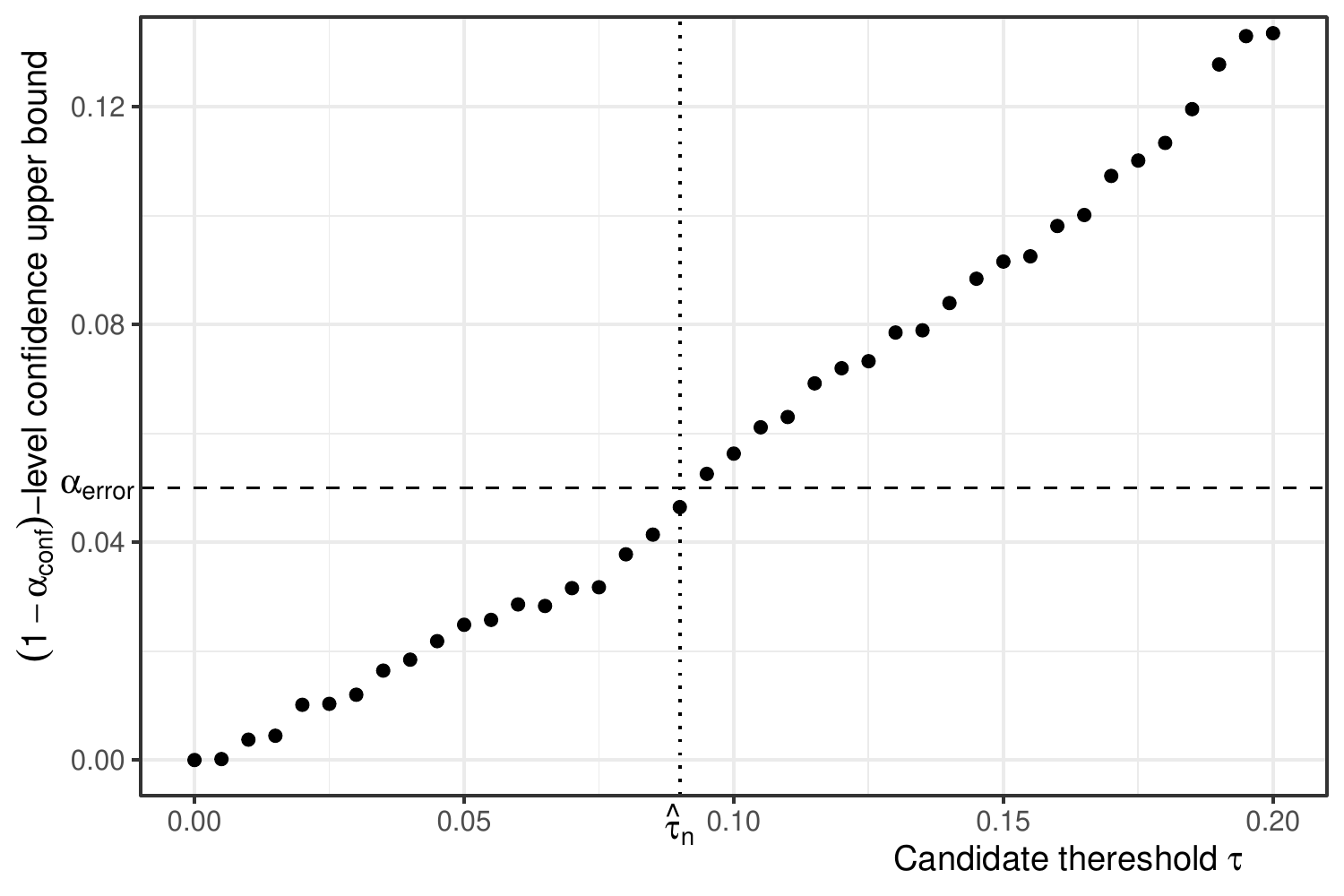}
    \caption{Illustration of threshold selection for prediction sets based on confidence upper bounds (CUBs). We first calculate approximate $(1-\alpha_\conf)$-level confidence upper bounds for the coverage error of prediction sets corresponding to a set $\mathcal{T}_n$ of candidate thresholds. We then select the threshold $\hat{\tau}_n$ to be the maximum threshold in the candidate set $\mathcal{T}_n$ such that, for any threshold $\tau \in \mathcal{T}_n$ less than or equal to $\hat{\tau}_n$, the CUB corresponding to $\tau$ is less than $\alpha_\error$.}
    \label{fig: illustrate}
\end{figure}

\begin{figure}
    \centering
    \includegraphics[scale=0.7]{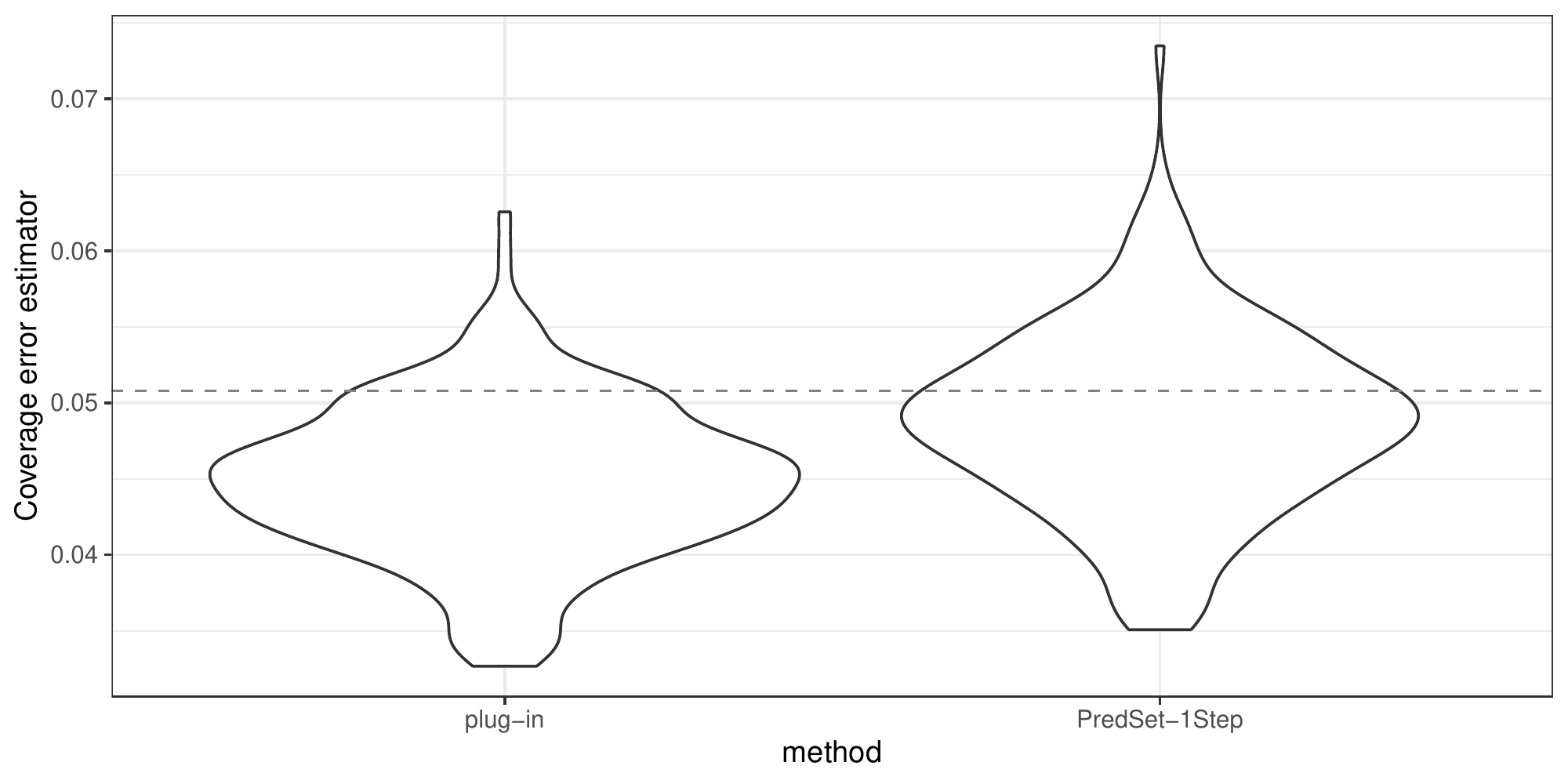}
    \caption{Violin plot of sampling distribution of the coverage error estimators with and without one-step correction (termed \textbf{PredSet-1Step} and \textbf{plug-in} respectively) for a given threshold, for a sample size of $n=4000$. The horizontal dashed line is the true coverage error corresponding to the given threshold. PredSet-1Step has significantly smaller bias than the plug-in estimator.}
    \label{fig: one step improvement}
\end{figure}

\begin{figure}
    \centering
    \includegraphics{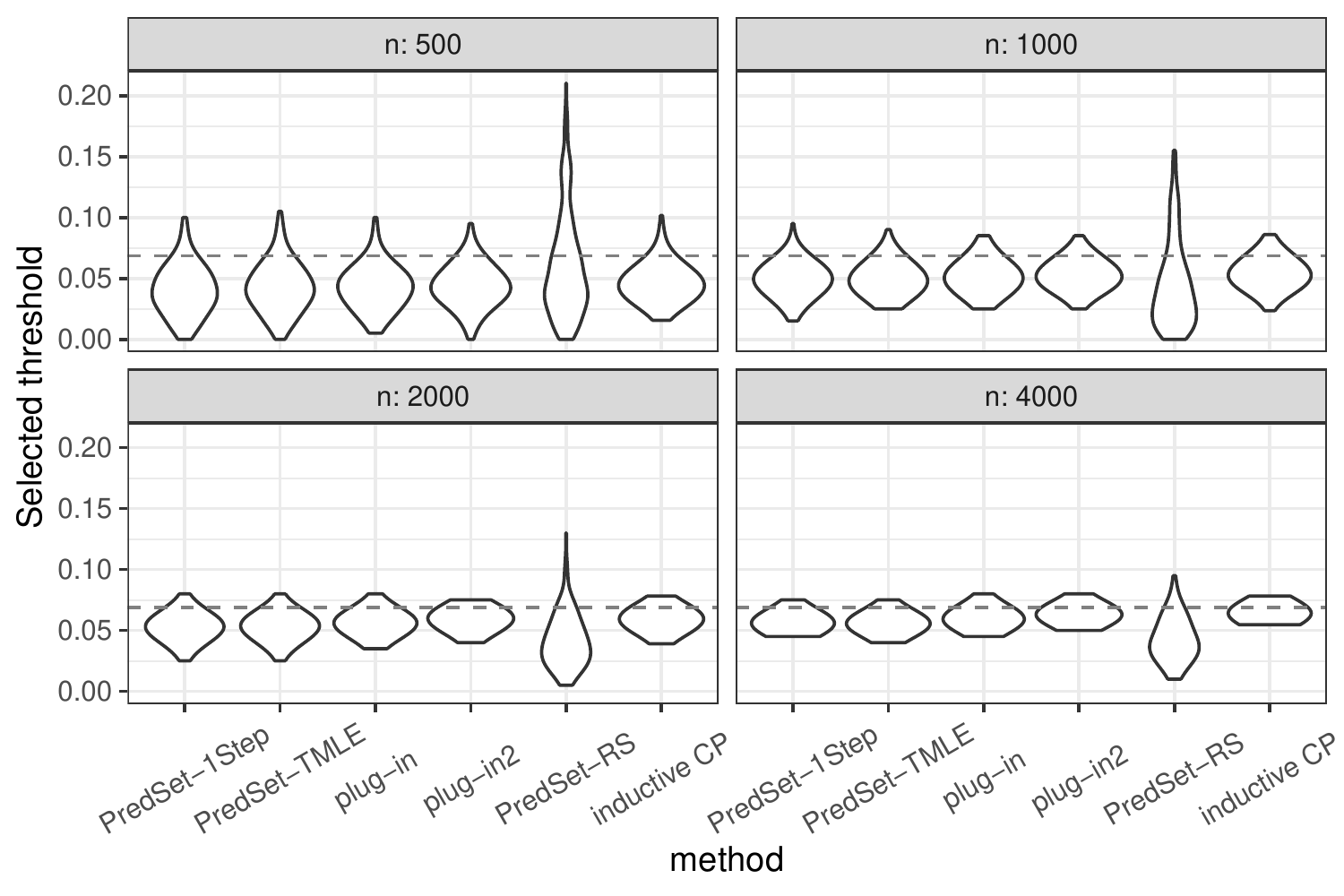}
    \caption{Violin plot of the sampling distribution of the selected threshold $\hat{\tau}_n$ in the moderate-to-high dimensional sparse setting. The gray horizontal dashed line is the true optimal threshold $\tau_0$.}
    \label{fig: high dim tauhat}
\end{figure}

\begin{figure}
    \centering
    \includegraphics{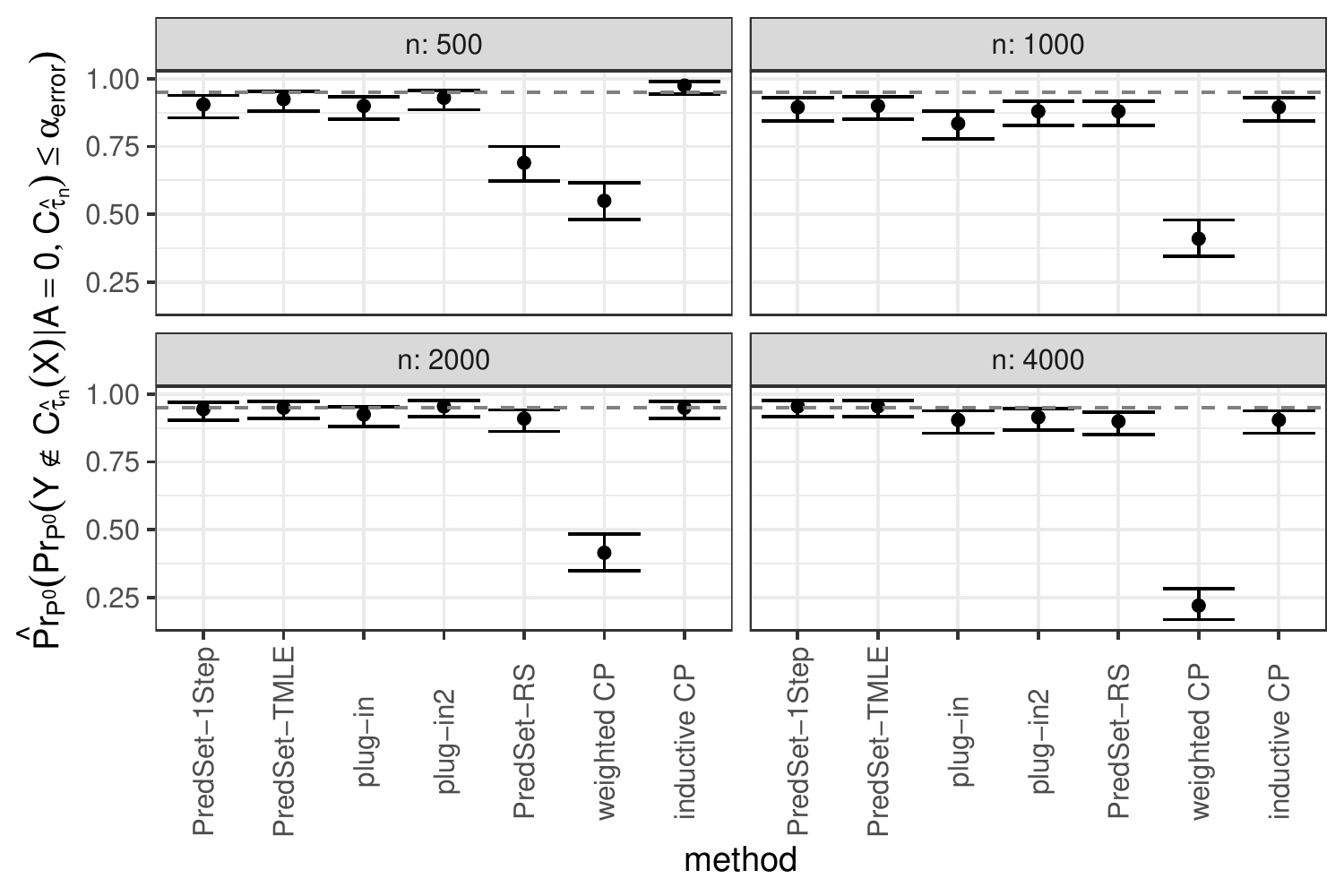}
    \caption{Empirical proportion of simulations where the estimated coverage error $\widehat\Prob_{P^0}(Y \notin C_{\hat{\tau}_n}(X) \mid A=0,C_{\hat{\tau}_n})$ does not exceed $\alpha_\error$, along with a 95\% Wilson score confidence interval, in the low dimensional setting. The gray horizontal dashed line is the desired confidence level $1-\alpha_\conf$.}
    \label{fig: low dim miscoverage}
\end{figure}

\begin{figure}
    \centering
    \includegraphics{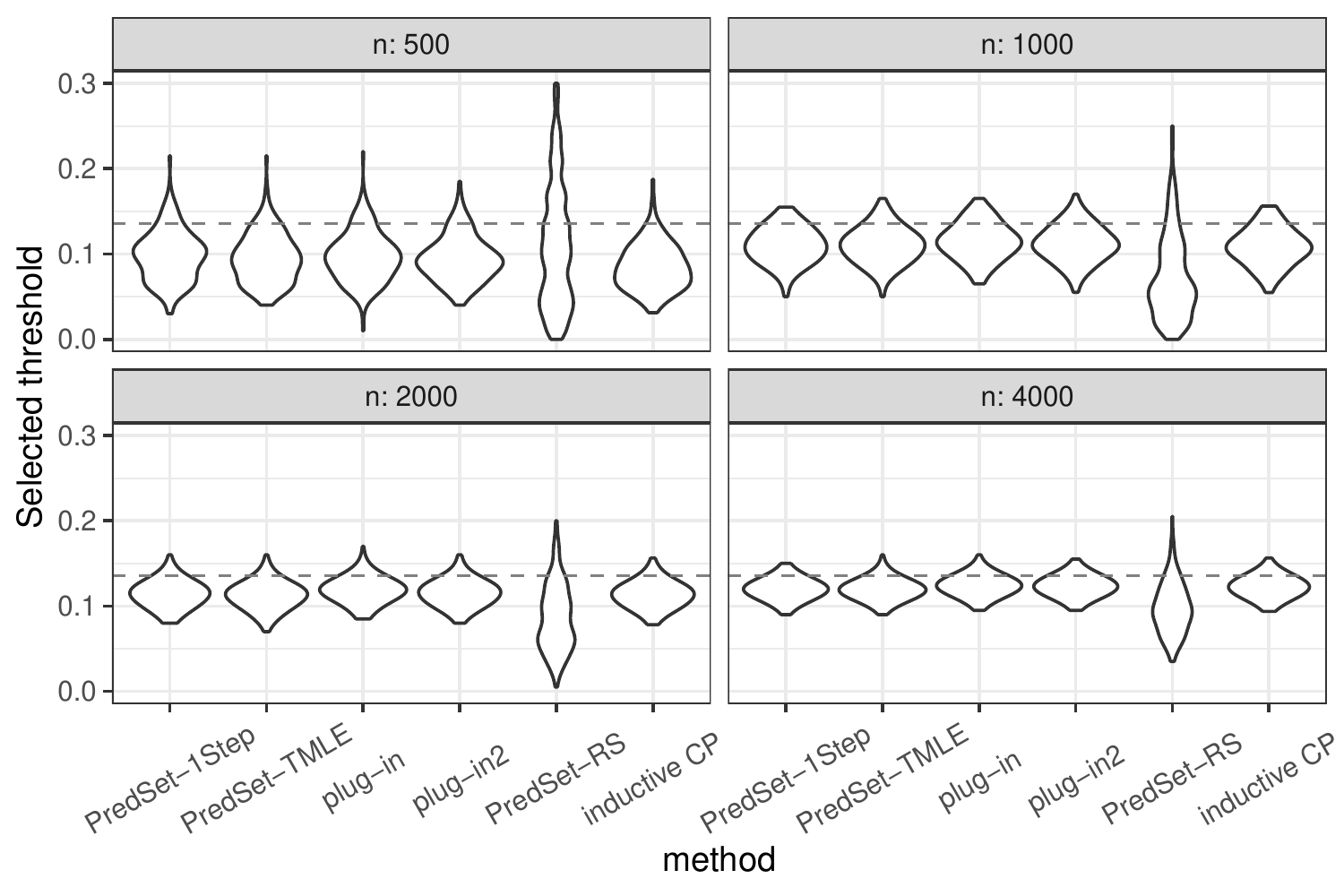}
    \caption{Violin plot of the sampling distribution of the selected threshold $\hat{\tau}_n$ in the low dimensional setting. The gray horizontal dashed line is the true optimal threshold $\tau_0$.}
    \label{fig: low dim tauhat}
\end{figure}

\begin{figure}
    \centering
    \includegraphics{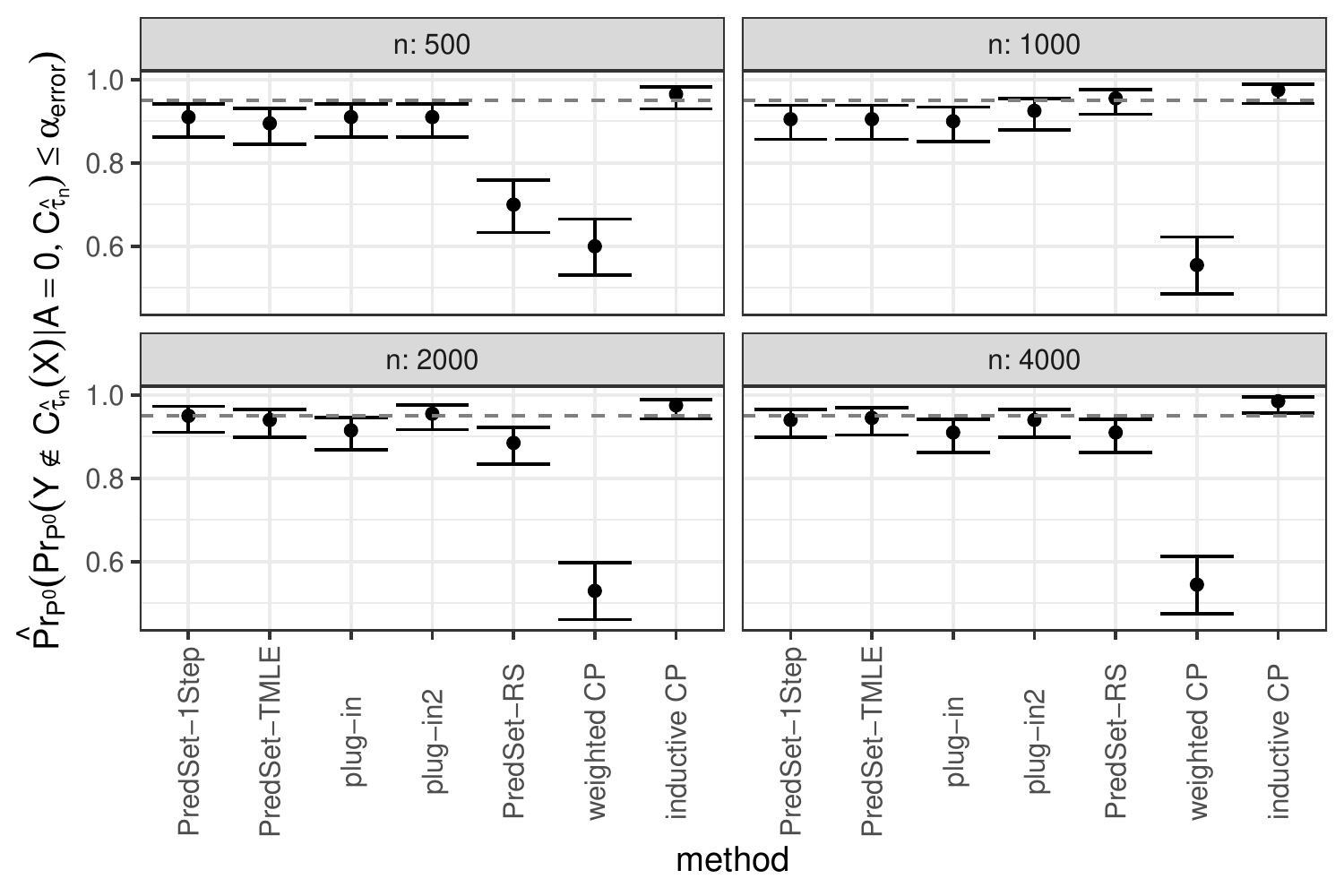}
    \caption{Empirical proportion of simulations where the estimated coverage error $\widehat\Prob_{P^0}(Y \notin C_{\hat{\tau}_n}(X) \mid A=0,C_{\hat{\tau}_n})$ does not exceed $\alpha_\error$, along with a 95\% Wilson score confidence interval in the low dimensional setting without covariate shift. The gray horizontal dashed line is the desired confidence level $1-\alpha_\conf$.}
    \label{fig: low dim noshift miscoverage}
\end{figure}

\begin{figure}
    \centering
    \includegraphics{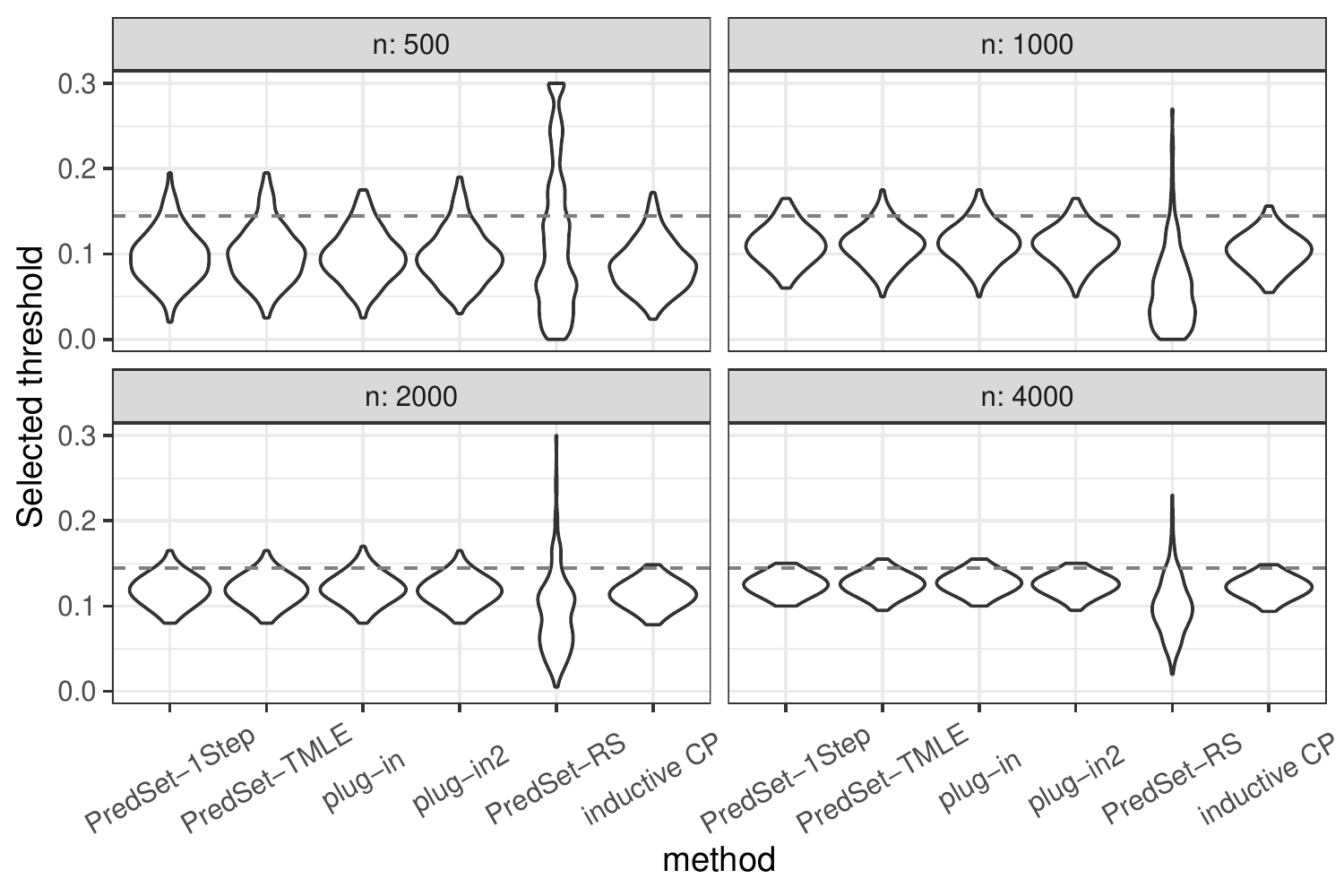}
    \caption{Violin plot of the sampling distribution of the selected threshold $\hat{\tau}_n$ in the low dimensional setting without covariate shift. The gray horizontal dashed line is the true optimal threshold $\tau_0$.}
    \label{fig: low dim noshift tauhat}
\end{figure}

\end{document}